\documentclass[11pt]{article}
\usepackage{graphicx}
\usepackage[dvipsnames]{xcolor}
\usepackage[section]{placeins}
\usepackage[top=1in, left=1in, right=1in, bottom=1in]{geometry}
\usepackage{caption}
\usepackage{subcaption}
\usepackage[export]{adjustbox}
\usepackage{amssymb}
\usepackage{bbm}
\usepackage{bm}
\usepackage{amssymb}
\usepackage{amsmath}
\usepackage{amsfonts}

\usepackage{enumitem}
\usepackage{setspace}
\usepackage{booktabs}
\usepackage{comment}
\usepackage[compact]{titlesec}

\usepackage{accents}

\usepackage{tabularray}

\usepackage[nocitation]{apacite}
\usepackage{natbib}
\bibpunct{(}{)}{;}{a}{}{,}
\usepackage{hyperref}
\usepackage{titling}
\usepackage{pifont}
\usepackage[capposition=top]{floatrow}
\newcommand{\subtitle}[1]{%
  \posttitle{%
    \par\end{center}
    \begin{center}\large#1\end{center}
    \vskip0.5em}%
}
\usepackage{amsthm}
\usepackage{amsmath} 
\usepackage{multirow}
\usepackage{multicol}
\usepackage{mathrsfs}

  {\end{smallmatrix}\right)}

\newcommand{\Var}{{\text{Var}}}
\newcommand{\Cov}{{\text{Cov}}}
\newcommand{\E}{{\mathbb{E}}}

\newtheorem{theorem}{Theorem}[section]

\newtheorem{proposition}[theorem]{Proposition}
\newtheorem{example}{Example}

\usepackage{mathrsfs}

\usepackage{color}

\allowdisplaybreaks

\usepackage{array}
\newcolumntype{C}[1]{>{\centering\arraybackslash}m{#1}}

\begin{document}
\pagestyle{plain}

\newtheoremstyle{mystyle}
{\topsep}
{\topsep}
{\it}
{}
{\bf}
{.}
{.5em}
{}
\theoremstyle{mystyle}
\newtheorem{assumptionex}{Assumption}
\newenvironment{assumption}
  {\pushQED{\qed}\renewcommand{\qedsymbol}{}\assumptionex}
  {\popQED\endassumptionex}
\newtheorem{assumptionexp}{Assumption}
\newenvironment{assumptionp}
  {\pushQED{\qed}\renewcommand{\qedsymbol}{}\assumptionexp}
  {\popQED\endassumptionexp}
\renewcommand{\theassumptionexp}{\arabic{assumptionexp}$'$}

\newtheorem{assumptionexpp}{Assumption}
\newenvironment{assumptionpp}
  {\pushQED{\qed}\renewcommand{\qedsymbol}{}\assumptionexpp}
  {\popQED\endassumptionexpp}
\renewcommand{\theassumptionexpp}{\arabic{assumptionexpp}$''$}

\newtheorem{assumptionexppp}{Assumption}
\newenvironment{assumptionppp}
  {\pushQED{\qed}\renewcommand{\qedsymbol}{}\assumptionexppp}
  {\popQED\endassumptionexppp}
\renewcommand{\theassumptionexppp}{\arabic{assumptionexppp}$'''$}

\renewcommand{\arraystretch}{1.3}

\newcommand\kosuke[1]{\cmnt{#1}{Kosuke}}
\newcommand\ambarish[1]{\cmnt{#1}{Ambarish}}
\newcommand\jose[1]{\cmnt{#1}{Jose}}

\newcommand{\argmin}{\mathop{\mathrm{argmin}}}
\makeatletter
\newcommand{\grande}{\bBigg@{2.25}}
\newcommand{\enorme}{\bBigg@{5}}

\newcommand{\blind}{0}

\newcommand{\tit}{Design-based inference for generalized network experiments with stochastic interventions}

\if0\blind

{\title{\tit\thanks{This work was supported through a grant from the Sloan Foundation (Economics Program; 2020-13946).}}
\author{Ambarish Chattopadhyay\thanks{Stanford Data Science, Stanford University, 450 Jane Stanford Way Wallenberg, Stanford, CA 94305; email: \url{hsirabma@stanford.edu}.}, \and Kosuke Imai\thanks{Department of Government and Department of Statistics, Harvard University, 1737 Cambridge Street, Institute for Quantitative Social Science, Cambridge, MA 02138; email: \url{imai@harvard.edu} URL: \url{https://imai.fas.harvard.edu}},\and Jos\'{e} R. Zubizarreta\thanks{Departments of Health Care Policy, Biostatistics, and Statistics, Harvard University, 180 Longwood Avenue, Office 307-D, Boston, MA 02115; email: \url{zubizarreta@hcp.med.harvard.edu}.}
}

\date{\today} 

\maketitle
}\fi

\if1\blind
\title{\bf \tit}
\maketitle
\fi

\begin{abstract}
A growing number of researchers are conducting randomized experiments to analyze causal relationships in network settings where units influence one another. A dominant methodology for analyzing these experiments is design-based, leveraging random treatment assignments as the basis for inference. In this paper, we generalize this design-based approach to accommodate complex experiments with a variety of causal estimands and different target populations.
An important special case of such generalized network experiments is a bipartite network experiment, in which treatment is randomized among one set of units, and outcomes are measured on a separate set of units. We propose a broad class of causal estimands based on stochastic interventions for generalized network experiments. Using a design-based approach, we show how to estimate these causal quantities without bias and develop conservative variance estimators. We apply our methodology to a randomized experiment in education where participation in an anti-conflict promotion program is randomized among selected students. Our analysis estimates the causal effects of treating each student or their friends among different target populations in the network. We find that the program improves the overall conflict awareness among students but does not significantly reduce the total number of such conflicts.
\end{abstract}


\clearpage
\doublespacing


\section{Introduction}
\label{sec_introduction}

In randomized experiments across the health and social sciences, units routinely interact with one another.
This often leads to the phenomenon of \textit{interference} where the outcome of one unit is influenced by the treatments assigned to other units \citep[e.g.,][]{halloran1995causal,nickerson2008voting,gupta2019top}.
Even when analyzing such complex experiments, the randomization of treatment assignment is under the control of the investigator and hence can serve as a ``reasoned basis'' for statistical inference \citep{fisher1935design}. 
This explains why the design-based or randomization-based inference has been a dominant approach to analyzing randomized experiments under interference \cite[e.g.,][]{ rosenbaum2007interference, hudgens2008toward, aronow2017estimating, athe:eckl:imbe:18}.

A strand of literature has developed design-based approaches for analyzing randomized experiments under {\it clustered network} or {\it partial} interference where spillover effects are assumed to occur only within the same cluster of units \citep[e.g.,][]{rosenbaum2007interference, hudgens2008toward, liu2014large, imai2021causal, park2022spillover}. 
Existing methods, however, can only be applied to experiments where all units in the network are eligible to receive the treatment.
This restriction represents an important limitation because many modern network experiments involve some units that are not eligible for treatment assignment or outcome measurement.
A prominent example is bipartite network experiments, where treatment is randomized among one set of units while the outcome is measured for a separate set of units \citep{doudchenko2020causal,harshaw2021design,zigler2021bipartite}. 
Conducting design-based inference for such experiments is challenging due to the inherent dependence within and between the ineligible and eligible units. 

In addition, although the existing methods focus on the average treatment effect among all treatment-eligible units, researchers may be interested in estimating causal effects for different target populations in the network, such as treatment-ineligible units or a group that includes some of both treatment-eligible and ineligible units.
For instance, in experiments on ride-sharing platforms such as Lyft and Uber, the treatment (e.g., price discount) may be applied only to riders while analysts wish to estimate causal effects separately for riders and drivers \citep{bajari2023experimental}.
In our motivating application \citep{paluck2016changing}, one question of interest is how popular students' participation in an anti-bullying program can influence the attitudes and behavior of their close friends who are ineligible for the program. 
To our knowledge, existing methods do not directly incorporate design-based inference for various target populations within the network.

In this paper, we propose a design-based causal inference framework and methodology for {\it generalized network experiments}, where an arbitrary subset of units are eligible to receive the treatment and a target population of interest may include both treatment-eligible and ineligible units.
Importantly, our framework does not make parametric assumptions nor imposes restrictions on the interference structure. 


As an important special case, generalized network experiments encompass bipartite network experiments, where treatment is randomized among one set of units while the outcome is measured for a separate set of units.
Bipartite network experiments are often used in two-sided markets where, for example, a price discount (treatment) is administered to a group of products (eligible units) whereas the amount purchased (outcome) is measured on buyers (ineligible units).
Although several scholars have recently proposed methods for analyzing bipartite network experiments \citep{doudchenko2020causal,harshaw2021design,zigler2021bipartite}, they are not fully design-based and are not applicable to other types of generalized network experiments, such as the school conflict experiment by \cite{paluck2016changing}. 

We first propose a broad class of causal estimands for generalized network experiments based on \textit{stochastic interventions}, which represent a probabilistic treatment assignment mechanism on the treatment-eligible units (Section~\ref{sec_ageneral}).
Under our framework, one can specify a stochastic intervention that assigns the treatment to each unit with different probabilities.
This allows us to formalize unit-level causal quantities under different treatment assignment mechanisms.
The proposed class of estimands extends the existing definitions of average direct, indirect, and total effects to generalized network experiments with arbitrary target populations \citep{hudgens2008toward, zigler2021bipartite}. 
These target populations may correspond to, for example, all treatment-eligible units, all treatment-ineligible units, all the units, or units defined by a set of covariates.

Second, we propose Horvitz-Thompson and H\'{a}jek estimators and develop design-based inferential approaches (Section~\ref{sec_designbased}). 
We show that the Horvitz-Thompson estimator is unbiased in finite samples whereas the H\'{a}jek estimator is unbiased in large samples.
Moreover, we obtain closed-form expressions for the design-based variances of these estimators. 
We show that under certain assumptions about the structure of interference, it is possible to obtain conservative estimators of these variances.

To this end, we consider two interference structures. The first extends the notion of stratified interference \citep{hudgens2008toward} to generalized network experiments. The second builds on recent works on semiparametric modeling by \cite{zhang2023individualized} to propose a flexible additive interference structure in a design-based setting. 
We show that while both approaches lead to conservative variance estimators in finite samples, the added flexibility of additive interference comes at the cost of greater standard errors.
In a simulation study, we find that the H\'{a}jek estimator systematically produces more efficient estimates when compared to the Horvitz-Thompson estimators across different simulation settings (Section~\ref{sec_simulation}).

Finally, we apply our methodology to reanalyze an influential randomized clustered network experiment \citep{paluck2016changing} concerning an anti-conflict program in public middle schools (Section~\ref{sec_empirical}).
The original analysis focused on understanding whether encouraging a group of students to take a public stance against conflict (i.e., treatment) can shift overall levels of conflict behavior in schools.
Our analysis, instead, examines whether and to what extent the behavior of students is influenced by their own treatment status or the treatment of their close friends. 
This alternative question is of interest because for 
any given student, their own treatment status and that of their close eligible friend are more likely to influence their behavior, when compared to the treatment status of the other students. Moreover, we examine the impact of the program on all the students and separately among eligible and ineligible students. 
We find that intervening on their close friends and themselves improves students' awareness and overall stance against conflict, while it does not significantly reduce the number of conflict cases in schools, on average. 

\paragraph{Related literature.}
There exists extensive literature on causal inference with interference \cite[see][for earlier reviews]{tchetgen2012causal,halloran2016dependent}.
The most commonly analyzed experimental design in this literature is two-stage randomization.
Building on the seminal work by \cite{hudgens2008toward}, many scholars have developed and applied methods to estimate various direct and spillover effects \citep[e.g.,][]{sinclair2012detecting,crepon2013labor,liu2014large,baird2018optimal,basse2018analyzing,imai2021causal}.
Similar to this literature, we allow general network interference within each cluster, but unlike two-stage randomized designs, we consider the possible existence of treatment-ineligible units and a broader class of spillover effects. 

Beyond this specific experimental design, \cite{aronow2017estimating} propose an exposure mapping approach by assuming that the potential outcome of one unit depends on the treatment assignments of other units in a network only through a known low-dimensional function of the treatment assignments \citep{toulis2013estimation,leung2020treatment}.
In practice, however, it is often impossible to observe all the ways in which units interact with one another.
As a result, the assumptions that severely restrict the structure of interference may be difficult to justify.

We take an alternative approach based on stochastic interventions that avoid the specification of an exposure map while maintaining the interpretability of empirical findings.
For variance estimation, however, we also assume a certain form of interference as done in the previous works that analyze randomized experiments under interference. 
In particular, our analysis incorporates extensions of stratified interference and a more flexible additive interference, each of which enables us to obtain conservative variance estimators \citep{yu2022estimating,zhang2023individualized}.

Our work also contributes to the fast growing literature on bipartite network experiments.
There are two basic approaches.
First, a series of recent works build upon the aforementioned exposure mapping approach under bipartite settings.
Under a linear exposure map, \cite{pouget2019variance} develop a clustering algorithm for estimating the global average treatment effect, i.e., the average effect of treating all eligible units.  
In addition, under an arbitrary but known exposure map, \cite{doudchenko2020causal} show how to estimate the global average treatment effect using regression and weighting methods based on generalized propensity scores.
More recently, \cite{harshaw2021design} assume a linear model for both the exposure and the response to estimate the global average treatment effect and develop an inferential approach using asymptotic approximations.
As mentioned above, we do not adopt the exposure mapping approach in order to avoid, whenever possible, restrictions on the structure of interference.

The second approach to bipartite network experiments is based on stochastic interventions.
\cite{zigler2021bipartite} introduce this alternative approach.
While they formulate a set of estimands using stochastic interventions, and propose Horvitz-Thompson-type estimators for these estimands, the authors do not consider formal variance estimation. 
Building on this seminal work, we generalize their estimands and develop design-based assumption-lean inference.

We also contribute to the recent literature on stochastic interventions.  
Two-stage randomization discussed above can be seen as an application of stochastic intervention.
More recently, stochastic interventions have been used in a variety of settings, including causal inference in longitudinal studies \citep{kennedy2019nonparametric}, mediation analysis \citep{diaz2019causal}, analysis of spatio-temporal data \citep{papadogeorgou2020causal}, and other types of observational studies \citep[e.g.,][]{munoz2012population,young2014identification,papadogeorgou2019causal,zigler2020bipartite}.
We further extend stochastic interventions to generalized network experiments.

Finally, another related literature focuses on the design and analysis of spatial experiments. 
In particular, \cite{wang2020design} define causal estimands by considering a circle-average outcome for each treatment-eligible point in space by focusing on Bernoulli assignments \citep[see also][]{wang2021causal}. 
In contrast, the outcomes and estimands in our framework are defined at the level of both eligible and ineligible units and allow for arbitrary assignment mechanisms.

\section{Effectiveness of anti-conflict interventions in schools}
\label{sec_theeffect}

In this section, we introduce the clustered network experiment analyzed later in the paper and discuss the substantive questions that motivate our proposed methodology. 

\subsection{Background}

An important question in the behavioral and social sciences is whether and how a shift in the attitude and behavior of a few individuals can be transmitted through social networks to induce community-wide changes.
\cite{paluck2016changing} used an innovative experimental design to study this question in the context of school conflicts, such as bullying, harassment, and other antagonistic interactions among students.
A primary goal of the study was to identify influential students who can effectively change the norms and behavior of other students in the same school. 

The authors conducted a randomized experiment across 56 public middle schools in New Jersey, of which 28 were randomly selected for an anti-conflict intervention (i.e., treatment). 
This program was designed to encourage participating students to take a public stance against school conflicts.
In each treated school, a group of students (called ``seed-eligible students'') was selected non-randomly.
Many of them were popular and reported having many friends within their schools.
On average, there were about 50 seed-eligible students and 200 ineligible students in each school.
Among the seed-eligible students, half of them (``seed students'') were randomly selected within a pre-defined stratum to participate in the anti-conflict intervention program. 
In addition, based on the number of social connections among students in each school, a group of highly connected seed students (``referent students'') were identified. On average, there were about five referent students per treated school.

A pre-experiment survey was fielded to collect student-level baseline data on demographics, social connections, and conflict behaviors and perceptions.
Each student reported up to 10 close friends that they spent time with during the last few weeks.
Post-experiment data were also collected using a similar survey at the end of the school year, along with the schools' administrative records.
The outcome variables measured awareness about conflict (e.g., whether students wore anti-conflict wristbands) and instances of conflict (e.g., number of cases of conflict).

\subsection{Motivating questions}

The authors of the original study were primarily interested in comparing the treated and control schools to estimate the causal effects of the intervention at the school level.
To this end, the authors conducted a model-based analysis by fitting a linear regression model of school-level outcomes (e.g., number of cases of conflict) on school-level characteristics (e.g., proportion of referent students) and school-level treatment status. 

Another key part of the original analysis focused on the effect of treating referent students on all the students in their social network.
Specifically, for the population of all the students in the network, the authors estimated the average causal effects of a new four-level treatment --- (i) having a seeded friend who is also a referent, (ii) having a seeded friend but no referent friends, (iii) being in a treated school but having no seeded friends, and (iv) being in a control school.

Estimation was done using a covariate-adjusted inverse probability weighted regression model with student-level data, and randomization-based inference was performed under the sharp null hypothesis of constant treatment effects.
The authors found that, in terms of peer-to-peer social influence, exposure to referent students increases awareness and perceived social norms against conflict, but it does not decrease instances of conflict.

In this paper, we provide an alternative approach to analyzing this experiment. Our analysis differs from that of \cite{paluck2016changing} in terms of the causal questions, the target populations, and the mode of estimation and inference.
First, unlike the original analysis, we examine how students' conflict behaviors are influenced by their own treatment status or that of their close friends, where closeness is determined by the information provided in the baseline survey.
We also examine whether, after taking into account the influence of close friends, the students' conflict behaviors are further affected by the referent students who are highly connected seed students.




Our proposed framework embeds these questions into a direct and indirect effects estimation problem under generalized network experiments (Section~\ref{sec_ageneral}).
A critical component of our framework is the notion of a \textit{key-intervention} unit.
For any unit in a network, its key-intervention units refer to one or more treatment-eligible units whose influence is of particular interest.
For instance, in this experiment, the key-intervention units of a seed-eligible student may be themselves, while the key-intervention units of a seed-ineligible student may be their closest seed-eligible friends.
The idea of a single key-intervention unit was introduced in \cite{zigler2021bipartite} for bipartite experiments.
We extend this notion by enabling multiple units to serve as key-intervention units under generalized network experiments.

Second, the original analysis estimates the average treatment effect on all the students within a school.
In contrast, our analysis separately estimates causal effects for seed-eligible and seed-ineligible students. The effects may vary between these two populations because the seed-eligible students are non-randomly selected and hence their characteristics differ.
Moreover, since the seed-ineligible students never receive the treatment, the effects of seed-eligible students' program participation on these students can be interpreted as peer effects. 
By comparing the average effects between these two target populations, we can examine the extent to which the treatment effect transmits from eligible students to their ineligible peers through their friendship network.

Finally, the original study used both model-based and design-based approaches.
In their model-based analysis, the inferential validity relies on the appropriateness of the assumed regression models.
In contrast, our analyses are fully design-based and do not require modeling assumptions. 
In their design-based analysis, the uncertainty quantification relies upon the assumption of constant additive treatment effects.
We address this limitation by providing design-based confidence intervals while allowing for heterogeneous treatment effects.

\section{Methodological framework}
\label{sec_ageneral}
\subsection{Setup and notation}
\label{sec_setup}

Consider a generalized network experiment on a finite population comprising a set of treatment-eligible or \textit{intervention} units $\mathcal{I}$ (e.g., seed-eligible students in our application) and a set of treatment-ineligible or \textit{non-intervention} units $\mathcal{O}$ (e.g., seed-ineligible students).
These units are grouped into $K \geq 1$ non-overlapping clusters (e.g., schools). 
While most of our proposed methodology applies to experiments with a single ($K = 1$) cluster, we retain the clustered setting throughout the paper to maintain consistency with our motivating application in Section \ref{sec_theeffect}.

We write $\mathcal{I} = \mathcal{I}_1 \cup ... \cup \mathcal{I}_K$ and $\mathcal{O} = \mathcal{O}_1 \cup ... \cup \mathcal{O}_K$, where $\mathcal{I}_k$ and $\mathcal{O}_k$ denote the sets of $n_k$ intervention units and $m_k$ non-intervention units in cluster $k$, respectively. 
By definition, $\mathcal{I}_k$ and $\mathcal{O}_k$ are disjoint.
Denote $\mathcal{S}_k \subseteq \mathcal{I}_k \cup \mathcal{O}_k$ as the target population of interest in cluster $k$ for which we wish to learn certain causal effects of the intervention. 
Finally, we use $\mathcal{S} = \mathcal{S}_1 \cup \mathcal{S}_2 \cup ... \cup \mathcal{S}_K$ to denote the combined target population across all clusters. 

\begin{figure}[!t]
  \centering
    \includegraphics[scale = 0.3]{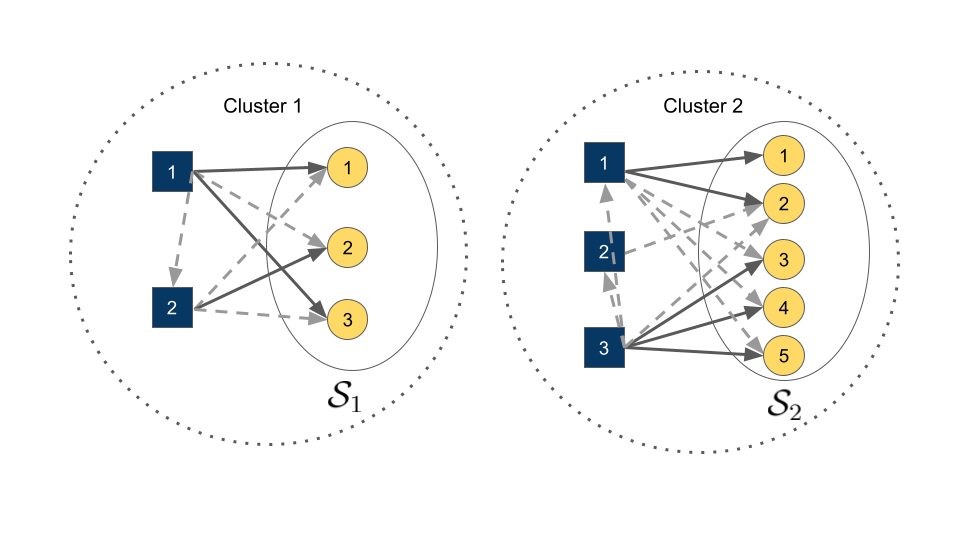}
    \caption{An example of a generalized network experiment with two clusters. An arrow from an intervention unit (blue square) to a non-intervention unit (yellow circle) indicates that the treatment received by the intervention unit may affect the outcome of the non-intervention unit. For each non-intervention unit, solid arrows correspond to the key-intervention units and the dashed arrows correspond to other intervention units. For each intervention unit, the corresponding key-intervention unit is itself. Here, the target population is the set of non-intervention units.} 
    \label{fig:bip}
\end{figure}

Figure~\ref{fig:bip} presents an example with $K=2$ clusters.
The intervention units are represented by blue squares while the non-intervention units are yellow circles. 
An arrow denotes a potential causal effect of an intervention unit's treatment on another unit's outcome, which may or may not be eligible for the treatment.
Here, the target population is the set of non-intervention units. 

With different choices of target population $\mathcal{S}$, this setup encompasses other common network experimental designs as illustrated in the following examples.
\begin{example}[Standard clustered network experiments] \normalfont \singlespacing
$\mathcal{S}_k = \mathcal{I}_k$, where inferences are made for the units that can be assigned to either the treatment or control condition.  
\end{example}
\begin{example}[Bipartite experiments] \normalfont \singlespacing
$\mathcal{S}_k = \mathcal{O}_k$, where inferences are made for the units that are not eligible to receive treatment.
\end{example}
\begin{example} \normalfont \singlespacing
$\mathcal{S}_k = \mathcal{I}_k \cup \mathcal{O}_k$, where inferences are made for the entire population of units.    
\end{example}
In addition, $\mathcal{S}_k$ may correspond to a population characterized by covariates.
For instance, in the school conflict experiment, the population of interest may be all the female students or all students who were involved in at least one case of conflict before the experiment took place.

Next, let $\bm{A}_k = (A_{ki}:i \in \mathcal{I}_k)$ denote the vector of treatment assignment indicators in cluster $k$, with $A_{ki} = 1$ if intervention unit $i \in \mathcal{I}_k$ receives the treatment, and $A_{ki} = 0$ otherwise. 
Also, we measure $d$ baseline covariates for each unit in cluster $k$.
Define $\bm{X}_k \in \mathbb{R}^{d(n_k + m_k)}$ as the stacked vector of observed covariates across all the units in cluster $k$, and write $\mathcal{X} = \{\bm{X}_1,...,\bm{X}_K\}$. 
For each unit $j \in \mathcal{S}_k$, let $Y_{kj}(\bm{a}_1,...,\bm{a}_K)$ represent its potential outcome (e.g., the number of conflict incidents) when the vectors of treatment levels in clusters $1,2,...,K$ equal $\bm{a}_1,\bm{a}_2,...,\bm{a}_K$, respectively.
We use $\mathcal{Y} = \{Y_{kj}(\bm{a}_1,...,\bm{a}_K): \bm{a}_k \in \{0,1\}^{n_k}, j\in \mathcal{S}_k, k\in\{1,...,K\} \}$ to represent the set of all possible potential outcomes across all units in $\mathcal{S}$.
Finally, for $j \in \mathcal{S}_k$, let $Y^{\text{obs}}_{kj} = Y_{kj}(\bm{A}_1,...,\bm{A}_K)$ be the corresponding observed outcome. 

Throughout the paper, we adopt a finite population causal inference framework \citep{neyman1923application}, where the sets of potential outcomes $\mathcal{Y}$ and the covariates $\mathcal{X}$ are fixed and randomness stems from the treatment assignments ($\bm{A}_1,...,\bm{A}_K$) alone. 
We assume that the potential outcomes of each unit in a cluster may depend on the treatments assigned to any intervention unit in the same cluster but do not depend on the treatment assignments of units in other clusters.
\begin{assumption}[Partial interference \citep{sobel2006randomized,hudgens2008toward}] \normalfont \singlespacing
For all $k \in \{1,2,...,K\}$ and $j \in \mathcal{I}_k \cup \mathcal{O}_k$, $Y_{kj}(\bm{a}_1,...,\bm{a}_K) = Y_{kj}(\bm{a}'_1,...,\bm{a}'_K)$ if $\bm{a}_k = \bm{a}'_k $.
\label{assump_partial}
\end{assumption}
Assumption~\ref{assump_partial} allows for interference within each cluster, but rules out interference across clusters. 
In our application, the assumption implies that students do not influence one another across schools.
Under this assumption, we can write $Y_{kj}(\bm{a}_k) = Y_{kj}(\bm{a}_1,...,\bm{a}_K)$ for all $k \in \{1,2,...,K\}$ and $j \in \mathcal{S}_k$.
We emphasize, however, that the subsequent theory and methods apply even if Assumption~\ref{assump_partial} is violated, as such experiments can be conceptualized as experiments with a single ($K=1$) cluster.

Finally, for each unit $j$ in cluster $k$, denote $i^* = i^*(j) \in \mathcal{I}_k$ as its \emph{key-intervention} unit \citep{zigler2021bipartite}. 
The key-intervention unit is often of interest as they are likely to influence the behavior of the corresponding non-intervention unit. 
In our application, the key intervention unit of a seed-eligible student may be themselves, while the key-intervention unit of a seed-ineligible student may be the best friend of the student. 
We can directly incorporate the role of key-intervention unit in the definition of causal estimands, as shown below.

\subsection{Estimands}
\label{sec_estimand}

Under the setup described above, we define a broad class of causal estimands for generalized network experiments using stochastic interventions, extending the existing estimands.

\paragraph{General formulation.}
For any unit $j \in \mathcal{S}_k$, we consider a stochastic intervention $\pi_{kj,\mathcal{X}}(\cdot): \{0,1\}^{n_k} \to [0,1]$ on the intervention units in cluster $k$. That is $\pi_{kj,\mathcal{X}}(\cdot)$ is a probability distribution over all possible treatment assignments on the units in $\mathcal{I}_k$, which may depend on unit $j$ and the set of covariates $\mathcal{X}$. 
For notational simplicity, we will omit the subscript $\mathcal{X}$ and write $\pi_{kj}(\cdot)$.

Our proposed causal estimand formalizes the notion of target population average potential outcome, where for each unit $j \in \mathcal{S}_k$, treatments are assigned to the units in $\mathcal{I}_k$ using the assignment mechanism $\pi_{kj}(\cdot)$. 
The formal definition is given by,
\begin{equation}
\tau^\pi = \frac{1}{K}\sum_{k=1}^{K}\left[\frac{1}{|\mathcal{S}_k|}\sum_{j \in \mathcal{S}_k}\left\{\sum_{\bm{a}\in \{0,1\}^{n_k}}\pi_{kj}(\bm{a})Y_{kj}(\bm{a})\right\}\right].
\label{eq:general_est}
\end{equation}
This estimand involves averaging at three levels. 
First, for each unit $j \in \mathcal{S}_k$, the potential outcomes $Y_{kj}(\bm{a})$ are averaged over all possible assignments $\bm{a}$ of the intervention units according to $\pi_{kj}(\cdot)$. 
Second, these unit-level average potential outcomes are further averaged over all units in $\mathcal{S}_k$ to obtain cluster-level average potential outcomes. 
Finally, they are averaged over all the clusters to obtain the target population average potential outcome. 

In this paper, we characterize the stochastic intervention corresponding to each $j \in \mathcal{S}_k$ in terms of a set of design-admissible treatment assignments $\mathcal{C}_{kj} \subseteq \{0,1\}^{n_k}$.
For instance, $\mathcal{C}_{kj} = \{\bm{a} \in \{0,1\}^{n_k}: A_{ki^*} = 1\}$ corresponds to the subset of possible treatment assignments where the key-intervention unit $i^*$ of unit $j$ receives the treatment. 
This leads to the stochastic intervention $\pi_{kj}(\cdot) = \pi_k(\cdot\mid\mathcal{C}_{kj})$, where $\pi_k(\cdot)$ is a probability distribution, free of $j$. 
For the school conflict experiment with $\mathcal{S}_k = \mathcal{O}_k$, we can interpret the resulting estimand $\tau^\pi$ as the target population average potential outcome (e.g., the average number of conflicts in all schools) where, for each seed-ineligible student $j$ in school $k$, the seed-eligible students in cluster $k$ receive the intervention with probability $\pi_k(\cdot)$, restricted to the set of design-admissible assignments $\mathcal{C}_{kj}$. 

\paragraph{Effects of a single key-intervention unit.}
Using $\tau^\pi$, we can encapsulate several practically relevant causal estimands as special cases, including the direct, indirect, and total effects in standard network experiments \citep{hudgens2008toward} and bipartite experiments \citep{zigler2021bipartite}. To see this, we first set $\mathcal{C}_{kj} = \{\bm{a} \in \{0,1\}^{n_k}: A_{ki^*} = a\}$ where $a \in \{0,1\}$ and note that $\tau^\pi$ can be written as, 
\begin{equation}
\mu^\pi_{a} = \frac{1}{K}\sum_{k=1}^{K} \left[\frac{1}{|\mathcal{S}_k|}\sum_{j \in \mathcal{S}_k}\left\{\sum_{\bm{s}}\pi_k(\bm{A}_{k(-i^*)} = \bm{s}\mid A_{ki^*} = a)Y_{kj}(A_{ki^*} = a, \bm{A}_{k(-i^*)} = \bm{s})\right\}\right].     
\end{equation}
In the school conflict experiment, setting $\mathcal{S}_k = \mathcal{O}_k$, $\mu^\pi_a$ represents the average potential outcome when in each school $k$, the best seed-eligible friend of each seed-ineligible student is assigned to the treatment condition $a \in \{0,1\}$, while the treatment assignment for all the other seed-eligible students in the school follows the distribution $\pi_k(\cdot)$.

Using the definition of $\mu^\pi_a$ above, we can write the direct effect as follows, 
\begin{equation}
\text{DE}^\pi = \mu^\pi_1 - \mu^\pi_0.   
\label{eq_DE}
\end{equation}
In the school conflict experiment, setting $\mathcal{S}_k  =\mathcal{O}_k$, we can interpret $\text{DE}^\pi$ as the average effect of treating the best seed-eligible friend of every seed-ineligible student, letting the treatment assignment of all the other seed-eligible students in the school $k$ follow $\pi_k(\cdot)$. 
In this case, $\text{DE}^\pi$ equals the existing definition of the direct effect in bipartite experiments \citep{zigler2021bipartite}.
Likewise, for $\mathcal{S}_k  =\mathcal{I}_k$, $\text{DE}^\pi$ is equivalent to the existing definition of the direct effect in standard network experiments \citep{hudgens2008toward}.

For a fixed treatment level $a \in \{0,1\}$, we can also formalize the indirect effect as
\begin{equation}
 \text{IE}^{\pi,\tilde{\pi}}_{a} = \mu^\pi_a - \mu^{\tilde{\pi}}_a,  
 \label{eq_IE}
\end{equation}
where $\tilde{\pi}_k(\cdot): \{0,1\}^{n_k} \to [0,1]$ is another stochastic intervention on the intervention units in cluster $k$. For $\mathcal{S}_k  = \mathcal{O}_k$, we can interpret $\text{IE}^{\pi,\tilde{\pi}}_{a}$ as the average effect of changing the treatment assignment mechanism of all but the best seed-eligible friend of every seed-ineligible student in school $k$ from $\tilde{\pi}_k(\cdot)$ to $\pi_k(\cdot)$, while holding the treatment level of the best-seed eligible friend fixed at $a$.
Here, $\pi_k(\cdot)$ and $\tilde{\pi}_k(\cdot)$ may correspond to assignment mechanisms where we treat a higher proportion of referent students in one and a lower proportion in the other.
Once again, this definition extends the existing notions of indirect effect to generalized network experiments. 

We can also contrast the treatment status of the key-intervention unit and two different stochastic interventions simultaneously to define an average total effect, $\text{TE}^{\pi,\tilde{\pi}} = \mu^\pi_1 - \mu^{\tilde{\pi}}_0$. For $\mathcal{S}_k  = \mathcal{O}_k$, we can interpret $\text{TE}^{\pi,\tilde{\pi}}$ as the average effect of providing the treatment to the best seed-eligible friend of every seed-ineligible student in school $k$, while changing the treatment assignment mechanism of the other seed-eligible students from $\tilde{\pi}_k(\cdot)$ to $\pi_k(\cdot)$.

\paragraph{Effects of multiple key-intervention units.}

In addition, $\tau^\pi$ can also incorporate causal quantities based on \textit{multiple} key-intervention units. For example, we can define the average potential outcome under a stochastic intervention that intervenes on a fixed proportion (e.g., 0.5) of the seed-eligible friends of a student. 
To formalize this intervention, for unit $j \in \mathcal{S}_k$, denote $\bm{i}^* = \{i^*_1,...,i^*_{r_j}\}$ as the corresponding set of $r_j$ seed-eligible key-intervention units. 
With multiple key-intervention units, an analog of $\mu^\pi_a$ can be obtained by setting $\mathcal{C}_{kj} = \{\bm{a} \in \{0,1\}^{n_k}: A_{ki^*_{s}} = a, s\in \{1,...,r_j\}\}$. More generally, we can set $\mathcal{C}_{kj} = \{\bm{a} \in \{0,1\}^{n_k}: \sum_{s=1}^{r}A_{ki^*_s}/r_j = p^*\}$, where $p^* \in [0,1]$. In this case, $\tau^\pi$ represents the target population average potential outcome under the intervention mechanism $\pi_k(\cdot)$, while fixing, for each unit in $\mathcal{S}_k$, the proportion of treated key-intervention units to $p^*$.

\subsection{Nonparametric identification and estimation}
\label{sec_identify}

Let $f_{\mathcal{X},\mathcal{Y}}(\cdot)$ and $f_{k,\mathcal{X},\mathcal{Y}}(\cdot)$ denote the joint distributions of the assignment mechanisms in the overall population and in cluster $k$, respectively, which may depend on the set of covariates and the potential outcomes.
Unless otherwise specified, for notational simplicity, we omit the additional subscripts $\mathcal{X}$ and $\mathcal{Y}$. 
The proposed estimand $\tau^\pi$ in Equation~\eqref{eq:general_est} can be non-parametrically identified under the following assumptions. 
\begin{assumption}[Identification conditions]
\normalfont \singlespacing
\begin{enumerate}[label=(\alph*)]
\item[]
    \item \textit{Overlap:} For all $k \in \{1,2,...,K\}$, $\text{Supp}(\pi_{kj,\mathcal{X}}) \subseteq \text{Supp}(f_{k,\mathcal{X},\mathcal{Y}})$.

    \item \textit{Unconfoundedness:} For all $k \in \{1,2,...,K\}$ and for all $\bm{a} \in \{0,1\}^{n_k}$, $f_{k,\mathcal{X},\mathcal{Y}}(\bm{a}) = f_{k,\mathcal{X}}(\bm{a})$. 
\end{enumerate}
\label{assump_identify}
\end{assumption}

Assumption~\ref{assump_identify}(a) states that any treatment assignment with a strictly positive probability under the stochastic intervention $\pi_{kj}(\cdot)$ also has a strictly positive probability under the actual intervention $f_{k}(\cdot)$. 
This assumption can be satisfied by choosing the intervention distribution appropriately.
Assumption~\ref{assump_identify}(b) is analogous to the usual unconfoundedness assumption in observational studies, stating that given the set of covariates, the assignment mechanism does not depend on the potential outcomes. In randomized experiments, this assumption is satisfied by design. 

Under Assumption \ref{assump_identify}, we can nonparametrically identify $\tau^\pi$ as
\begin{equation}
    \tau^\pi =  \frac{1}{K}\sum_{k=1}^{K}\left[\frac{1}{|\mathcal{S}_k|}\sum_{j \in \mathcal{S}_k}\mathbb{E}\left\{\sum_{\bm{a}\in \{0,1\}^{n_k}}\mathbbm{1}(\bm{A}_k = \bm{a})\frac{\pi_{kj}(\bm{a})}{f_k(\bm{a})}Y^{\text{obs}}_{kj}\right\}\right].
    \label{eq_identify}
\end{equation}
This identification result suggests the following Horvitz-Thompson-type estimator of $\tau^\pi$, 
\begin{equation}
\hat{\tau}^\pi_{\text{HT}} = \frac{1}{K}\sum_{k=1}^{K}\left\{\frac{1}{|\mathcal{S}_k|}\sum_{j \in \mathcal{S}_k}\frac{\pi_{kj}(\bm{A}_k)}{f_k(\bm{A}_k)}Y^{\text{obs}}_{kj}\right\}.    
\end{equation}
In the special case of $\tau^\pi = \mu^\pi_a$, we obtain, 
\begin{equation}
\hat{\tau}^\pi_{\text{HT}} = \frac{1}{K}\sum_{k=1}^{K}\left\{\frac{1}{|\mathcal{S}_k|}\sum_{j \in \mathcal{S}_k}\mathbbm{1}(A_{ki^*} = a)\frac{\pi_{k}(\bm{A}_{k(-i^*)} \mid A_{ki^*} = a)}{f_k(\bm{A}_k)}Y^{\text{obs}}_{kj}\right\} =: \hat{\mu}^\pi_{a,\text{HT}}.    
\end{equation}
For bipartite experiments (i.e., $\mathcal{S}_k = \mathcal{O}_k$), $\hat{\mu}^\pi_{a,\text{HT}}$ becomes the Horvitz-Thompson estimator of $\mu^\pi_a$ proposed by \cite{zigler2021bipartite}.

The theorem below shows that $\hat{\tau}^\pi_{\text{HT}}$ is unbiased for $\tau^\pi$ under the design-based framework.
\begin{theorem}[Unbiasedness] \normalfont \singlespacing
    $\mathbb{E}(\hat{\tau}^\pi_{\text{HT}}) = \tau^\pi$, where the expectation is taken over the assignment mechanism $f(\cdot)$.
    \label{thm_HT_unbiased}
\end{theorem}
Theorem \ref{thm_HT_unbiased} also suggests the following H\'{a}jek-type estimator of $\tau^\pi$,
\begin{equation}
\hat{\tau}^{\pi}_{\text{H\'{a}jek}} = \frac{\sum_{k=1}^{K}\left\{\frac{1}{|\mathcal{S}_k|}\sum_{j \in \mathcal{S}_k}\frac{\pi_{kj}(\bm{A}_k)}{f_k(\bm{A}_k)}Y^{\text{obs}}_{kj}\right\}}{\sum_{k=1}^{K}\left\{\frac{1}{|\mathcal{S}_k|}\sum_{j \in \mathcal{S}_k}\frac{\pi_{kj}(\bm{A}_k)}{f_k(\bm{A}_k)}\right\}},
\end{equation}
which replaces the denominator $K$ in $\hat{\tau}^{\pi}_{\text{HT}}$ with its Horvitz-Thompson estimator. In the special case of $\tau^\pi = \mu^\pi_a$, we have, 
\begin{equation}
\hat{\tau}^\pi_{\text{H\'{a}jek}} = \frac{\sum_{k=1}^{K}\left\{\frac{1}{|\mathcal{S}_k|}\sum_{j \in \mathcal{S}_k}\mathbbm{1}(A_{ki^*} = a)\frac{\pi_{k}(A_{ki^*} = a\mid\bm{A}_{k(-i^*)})}{f_k(\bm{A}_k)}Y^{\text{obs}}_{kj}\right\}}{\sum_{k=1}^{K}\left\{\frac{1}{|\mathcal{S}_k|}\sum_{j \in \mathcal{S}_k}\mathbbm{1}(A_{ki^*} = a)\frac{\pi_{k}(A_{ki^*} = a\mid\bm{A}_{k(-i^*)})}{f_k(\bm{A}_k)}\right\}} =: \hat{\mu}^\pi_{\text{H\'{a}jek}}.    
\end{equation}
Our H\'{a}jek estimators extends the existing H\'{a}jek estimators under interference (see, e.g., \citealt{wang2020design}) to general target populations, assignment mechanisms, and stochastic interventions.

While the H\'{a}jek estimator is not unbiased for $\tau^\pi$ in finite samples, we show that under certain regularity conditions, it is consistent for $\tau^\pi$.
The direct and indirect effects given in Equations~\eqref{eq_DE}~and~\eqref{eq_IE} are estimated analogously by replacing each component term by its Horvitz-Thompson and H\'{a}jek estimators.

\section{Design-based inference}
\label{sec_designbased}

In this section, we discuss design-based inference based on the estimators outlined in Section~\ref{sec_ageneral}. 
We derive the design-based variances of these estimators and obtain closed-form conservative estimators of these variances. 
For conciseness, we focus on the the Horvitz-Thompson estimators of the proposed causal quantities, and relegate related discussions on the H\'{a}jek estimators to Appendix \ref{appsec_hajek} of the Supplementary Materials.


\subsection{The general variance expression}
\label{sec_generalvar}

Throughout this section, we maintain the partial interference (Assumption~\ref{assump_partial}) and the identification assumptions (Assumption~\ref{assump_identify}). Additionally, we assume that the treatment assignment mechanisms are independent across clusters.
\begin{assumption}[Independence of treatment assignment mechanisms across clusters] \normalfont \singlespacing
    $\bm{A}_1,...,\bm{A}_K$ are mutually independent.
    \label{assump_indep}
\end{assumption}
We make this assumption to simplify the variance calculations, although it can be relaxed to incorporate dependence among clusters.
An example of such dependence includes the use of complete randomization across clusters in two-stage randomized experiments.
We note that this assumption is satisfied in the school conflict experiment.

Next, we consider the case of a single key-intervention unit and obtain a closed-form variance expression for $\hat{\mu}^\pi_{a,\text{HT}}$.
Appendix \ref{sec_appendix_multiple} presents the generalization of this result to the case of multiple key-intervention units.   
\begin{theorem}[Variance of the Horvitz-Thompson Estimator] \normalfont \singlespacing
Under Assumptions~\ref{assump_partial}--\ref{assump_indep},
\begin{align}
\Var(\hat{\mu}^\pi_{a,\text{HT}}) =   \frac{1}{K^2}\sum_{k = 1}^{K} \frac{1}{|\mathcal{S}_k|^2}\left(\sum_{j \in \mathcal{S}_k} \Lambda_{1,k,j} + \mathop{\sum\sum}_{j \neq j' \in \mathcal{S}_k}\Lambda_{2,k,j,j'} \right),  
\end{align}
where 
\begin{align}
\Lambda_{1,k,j} \ = \ & \sum_{\bm{s}}\frac{\pi^2_{k}(\bm{A}_{k(-i^*)} = \bm{s}\mid A_{ki^*} = a)}{f_k(A_{ki^*} = a, \bm{A}_{k(-i^*)} = \bm{s})}Y^2_{kj}(A_{ki^*} = a, \bm{A}_{k(-i^*)} = \bm{s}) \nonumber\\
& \quad - \left\{\sum_{\bm{s}}\pi_{k}(\bm{A}_{k(-i^*)} = \bm{s}\mid A_{ki^*} = a)Y_{kj}(A_{ki^*} = a, \bm{A}_{k(-i^*)} = \bm{s}) \right\}^2,
\label{eq_lambda1} \\
    \Lambda_{2,k,j,j'} \ = \  & \sum_{\tilde{\bm{s}}}\frac{\pi^2_k(A_{ki^{*}} = a, A_{ki^{*'}} = a, \bm{A}_{k(-i^*, -i^{*'})} = \tilde{\bm{s}})}{f_k(A_{ki^{*}} = a, A_{ki^{*'}} = a, \bm{A}_{k(-i^*, -i^{*'})} = \tilde{\bm{s}})\pi_k(A_{ki^*} = a)\pi_k(A_{ki^{*'}} = a)} \nonumber\\
& \quad  \times Y_{kj}(A_{ki^{*}} = a, A_{ki^{*'}} = a, \bm{A}_{k(-i^*,-i^{*'}) } = \tilde{\bm{s}})Y_{kj'}(A_{ki^{*}} = a,A_{ki^{*'}} = a, \bm{A}_{k(-i^*,-i^{*'})} = \tilde{\bm{s}}) \nonumber\\
& - \left\{\sum_{\bm{s}}\pi_{k}(\bm{A}_{k(-i^*)} = \bm{s}\mid A_{ki^*} = a)Y_{kj}(A_{ki^*} = a, \bm{A}_{k(-i^*)} = \bm{s}) \right\} \nonumber \\ 
& \times \left\{\sum_{\bm{s}}\pi_{k}(\bm{A}_{k(-i^{*'})} = \bm{s}\mid A_{ki^{*'}} = a)Y_{kj'}(A_{ki^{*'}} = a, \bm{A}_{k(-i^{*'})} = \bm{s}) \right\},
\end{align}
and $i^{*'} = i^*(j')$ is the key-intervention unit of unit $j'$.
\label{thm_generalvar}   
\end{theorem}
Theorem~\ref{thm_generalvar} implies that in general, the variance of $\hat{\mu}^\pi_{a,\text{HT}}$ is non-identifiable. 
The second term in the expression of $\Lambda_{1,kj}$ involves products of potential outcomes that are not observable simultaneously, e.g., $Y_{kj}(A_{ki^*} = a, \bm{A}_{k(-i^*)} = \bm{s})Y_{kj}(A_{ki^*} = a, \bm{A}_{k(-i^*)} = \bm{s}')$, where $\bm{s} \ne \bm{s}'$. Therefore, to identify $\Var(\hat{\mu}^\pi_{a,\text{HT}})$, we need to invoke some additional assumptions. In Appendix \ref{sec_partial_ht}, we discuss an approach to partially identify $\Var(\hat{\mu}^\pi_{a,\text{HT}})$ in completely randomized experiments, assuming a form of Lipschitz continuity for the potential outcomes.  

\subsection{Variance estimation under stratified interference}
\label{sec_stratified}

To point identify the variance of $\hat{\mu}^\pi_{a,\text{HT}}$, we need stronger restrictions on the pattern of interference.
To this end, we extend the assumption of stratified interference \citep[e.g.,][]{hudgens2008toward, liu2014large, imai2021causal}, which is often used to identify variances of estimated treatment effects in the presence of interference, to generalized network experiments. 
In particular, we assume that the potential outcomes of a unit in the target population depend on the treatments assigned to the intervention units in its cluster only through the assignment of its key-intervention unit and the proportion of treated intervention units in the cluster.
\begin{assumption}[Stratified interference] \normalfont \singlespacing
   For each unit $j \in \mathcal{S}_k$, if $\bm{a}, \bm{a}' \in \{0,1\}^{n_k}$ are such that $a_{i^*} = a'_{i^*}$ and $\bm{a}^\top\bm{1} = \bm{a}'^\top\bm{1}$, then $Y_{kj}(\bm{a}) = Y_{kj}(\bm{a}')$.\label{assump_stratified}
\end{assumption} 
In addition, we impose a mild design restriction that within each cluster $k$, we treat a fixed proportion $p_k$ of intervention units. Complete and stratified randomized experiments satisfy this restriction. More generally, this restriction is also satisfied by designs that allow differential assignment probabilities on each assignment vector having a proportion $p_k$ of treated units.  
\begin{assumption}[Fixed proportion treated] \normalfont \singlespacing
 For cluster $k \in \{1,2,...,K\}$, $\frac{\bm{A}^\top_k \bm{1}}{|\mathcal{I}_k|} = p_k$ for some fixed $p_k \in (0,1)$.
 \label{assump_fixedprop}
\end{assumption}
This assumption is satisfied in the school conflict experiment with $p_k = 0.5$ for all $k$. Under Assumptions~\ref{assump_stratified}~and~\ref{assump_fixedprop}, we can write $Y_{kj}(\bm{a}) = Y_{kj}(a_{i^*},p_k)$ for all $\bm{a}$ such that $\frac{\bm{a}^\top \bm{1}}{n_k} = p_k$.
Although existing literature provides variance estimators for standard experiments under Assumptions \ref{assump_stratified} and \ref{assump_fixedprop}, the direct application of these results to generalized network experiments is not evident due to the dependence between the eligible and ineligible units. For instance, two ineligible units $j,j' \in \mathcal{S}_k$ may share the same key-intervention unit $i^*$, introducing additional dependence between their outcomes, even under Assumptions \ref{assump_stratified} and \ref{assump_fixedprop}.

To this end, let us further denote $\mathbbm{1}(j \leftarrow i)$ as an indicator variable that equals one if intervention unit $i$ is the key-intervention unit of unit $j$, and equals zero otherwise. For intervention unit $i$ in cluster $k$ and $a \in \{0,1\}$, we define the \textit{pooled potential outcome} $\Tilde{Y}_{ki}(a,p_k) = \sum_{j \in \mathcal{S}_k}\mathbbm{1}(j \leftarrow i) Y_{kj}(a,p_k)$. In other words, $\Tilde{Y}_{ki}(a,p_k)$ sums up the potential outcomes of all the units in the target population whose key-intervention unit is $i$. Accordingly, we denote the pooled observed outcome as $\tilde{Y}^{\text{obs}}_{ki} = \tilde{Y}_{ki}(A_i, p_k)$. Under Assumptions~\ref{assump_partial}--\ref{assump_fixedprop}, Theorem \ref{thm_var1} provides a closed form expression of the variance of $\hat{\mu}^\pi_{a,\text{HT}}$ in terms of the pooled potential outcomes. 
\begin{theorem}[Variance under stratified interference] \normalfont \singlespacing
Under Assumptions~\ref{assump_partial}--\ref{assump_fixedprop}, 
\begin{align}
\Var(\hat{\mu}^\pi_{a,\text{HT}}) = \frac{1}{K^2}\sum_{k=1}^{K}\frac{1}{|\mathcal{S}_k|^2}\left\{\sum_{i=1}^{n_k}c_{i,a} \tilde{Y}^2_{ki}(a,p_k) + \mathop{\sum\sum}_{i \neq i'}d_{ii',a}\tilde{Y}_{ki}(a,p_k)\tilde{Y}_{ki'}(a,p_k)\right\}\label{eq_quadform}
\end{align}
for $a \in \{0,1\}$, where
$$\begin{aligned}
  c_{i,a} = & \frac{1}{\pi^2_k(A_{ki} = a)}\sum_{\bm{s}}\frac{\pi^2_k(A_{ki} = a, \bm{A}_{k(-i)} = \bm{s})}{f_k(A_{ki} = a, \bm{A}_{k(-i)} = \bm{s})} - 1, \\
  d_{ii',a} = & \frac{1}{\pi_k(A_{ki} = a)\pi_k(A_{ki'} = a)}\sum_{\bm{s}}\frac{\pi^2_k(A_{ki} = a, A_{ki'} = a, A_{k(-i,i')} = \bm{s})} {f_k(A_{ki} = a, A_{ki'} = a, A_{k(-i,i')} = \bm{s})} - 1.
\end{aligned}
$$ \label{thm_var1}
\end{theorem}
An unbiased estimator of this variance can be obtained by considering the Horvitz-Thompson estimator of each term in Equation~\eqref{eq_quadform}.
\begin{align*}
    \widehat{\Var}(\hat{\mu}^\pi_{a,\text{HT}}) = \frac{1}{K^2}\sum_{k=1}^{K}\frac{1}{|\mathcal{S}_k|^2}\left\{\sum_{i=1}^{n_k}\frac{\mathbbm{1}(A_{ki} = a)}{f_k(A_{ki} = a)}c_{i,a}\tilde{Y}_{ki}^2 + \mathop{\sum\sum}_{i \neq i'} \frac{\mathbbm{1}(A_{ki} = a, A_{ki'} = a)}{f_k(A_{ki} = a,A_{ki'} = a)} d_{ii',a}\tilde{Y}_{ki}\tilde{Y}_{ki'}\right\}
\end{align*}
Using the stratified interference assumption, we can also obtain the variance of the Horvitz-Thompson estimator of the direct effect.
\begin{theorem}[Variance of the direct effect estimator] \normalfont \singlespacing
Under Assumptions~\ref{assump_partial}--\ref{assump_fixedprop},
\begin{align*}
    \Var(\widehat{\text{DE}}^\pi_{\text{HT}}) = & \Var(\hat{\mu}^\pi_{1,\text{HT}}) + \Var(\hat{\mu}^\pi_{0,\text{HT}}) \nonumber \\ 
    &- 2\frac{1}{K^2}\sum_{k=1}^{K}\frac{1}{|\mathcal{S}_k|^2}\left[\mathop{\sum\sum}_{i \neq i'}g_{ii'}\tilde{Y}_{ki}(1,p_k)\tilde{Y}_{ki'}(0,p_k) - \sum_{i = 1}^{n_k} \tilde{Y}_{ki}(1,p_k)\tilde{Y}_{ki}(0,p_k)\right],
\end{align*}
where $g_{ii'} = \frac{1}{\pi_k(A_{ki} = 1)\pi_k(A_{ki'} = 0)}\sum_{\bm{s}}\frac{\pi^2_k(A_{ki} = 1, A_{ki'} = 0, A_{k(-i,i')} = \bm{s})}{f_k(A_{ki} = 1, A_{ki'} = 0, A_{k(-i,i')} = \bm{s})} - 1.$ \label{thm_varDE}
\end{theorem}
Theorem~\ref{thm_varDE} implies that even under stratified interference, the variance of the direct effect estimator $\Var(\widehat{\text{DE}}^\pi_{\text{HT}})$ is not identifiable because the cross product of the pooled potential outcome terms $\tilde{Y}_{ki}(1,p_k)\tilde{Y}_{ki}(0,p_k)$ cannot be observed for any $i \in \{1,...,n_k\}$. 
However, we can use the following upper bound of the variance to obtain a conservative estimator of  $\Var(\widehat{\text{DE}}^\pi_{\text{HT}})$,
\begin{align*}
\Var(\widehat{\text{DE}}^\pi_{\text{HT}}) \leq & \Var(\hat{\mu}^\pi_{1,\text{HT}}) + \Var(\hat{\mu}^\pi_{0,\text{HT}}) \\ 
    & - 2\frac{1}{K^2}\sum_{k=1}^{K}\frac{1}{|\mathcal{S}_k|^2}\left[\mathop{\sum\sum}_{i \neq i'}g_{ii'}\tilde{Y}_{ki}(1,p_k)\tilde{Y}_{ki'}(0,p_k) - \frac{1}{2}\sum_{i = 1}^{n_k} \{\tilde{Y}^2_{ki}(1,p_k) + \tilde{Y}^2_{ki}(0,p_k)\}\right].
\end{align*}
This upper bound can be estimated without bias as,
\begin{align*}
     \widehat{\Var}(\widehat{\text{DE}}^\pi_{\text{HT}}) = & \widehat{\Var}(\hat{\mu}^\pi_{1,\text{HT}}) + \widehat{\Var}(\hat{\mu}^\pi_{0,\text{HT}}) \\ 
    &- 2\frac{1}{K^2}\sum_{k=1}^{K}\frac{1}{|\mathcal{S}_k|^2}\Bigg[\mathop{\sum\sum}_{i \neq i'}\frac{\mathbbm{1}(A_{ki} = 1, A_{ki'} = 0)}{f(A_{ki} = 1,A_{ki'} = 0)}g_{ii'}\tilde{Y}_{ki}\tilde{Y}_{ki'}  \\
    & \hspace{1.2in} - \frac{1}{2}\sum_{i = 1}^{n_k} \left\{\frac{\mathbbm{1}(A_{ki} = 1)}{f_k(A_{ki} = 1)}\tilde{Y}^2_{ki} + \frac{\mathbbm{1}(A_{ki} = 0)}{f_k(A_{ki} = 0)}\tilde{Y}^2_{ki} \right\}\Bigg].
\end{align*}
Moreover, this variance estimator is unbiased for ${\Var}(\widehat{\text{DE}}^\pi_{\text{HT}})$ if $\tilde{Y}_{ki}(1,p_k) = \tilde{Y}_{ki}(0,p_k)$ for all $i$, i.e., the pooled potential outcomes for intervention unit $i$ under treatment and control are the same. This condition is analogous to Fisher's sharp null hypothesis of zero unit-level causal effect (see, e.g., \citealt{imbens2015causal}, Chapter 5).

In the special case of a completely randomized experiment with $\pi_k(\cdot) = f_k(\cdot)$, the exact variances of $\hat{\mu}^\pi_{a,\text{HT}}$ and $\widehat{\text{DE}}^\pi_{\text{HT}}$
and their estimators can be simplified, resembling the standard Neymanian variance estimators for the population mean and average treatment effect (Proposition \ref{prop_htvar1} in the Appendix). 
Under stratified interference, we also obtain the exact variance of the indirect effects. Proposition \ref{prop_varIE} in the Appendix provides exact closed-form expressions of this variance and shows that, like $\Var(\hat{\mu}^\pi_{a,\text{HT}})$, $\Var(\widehat{\text{IE}}^{\pi, \tilde{\pi}}_{a,\text{HT}})$ is a quadratic form in the pooled potential outcomes.
Thus, this variance can be estimated analogously using a Horvitz-Thompson estimator.

\subsection{Variance estimation under additive interference}
\label{sec_additive}

The stratified interference assumption (as stated in Assumption~\ref{assump_stratified}) inherently assumes a uniformity in spillover effects; i.e., it presumes that treated intervention units influence the outcome of unit $j$ to the same extent. This assumption can be overly restrictive in scenarios where unit $j$ is likely to be more influenced by one intervention unit (e.g., a close friend) than by another. 
In view of this, building on the work by \cite{zhang2023individualized}, we propose a more flexible assumption regarding the interference structure within each cluster.
\begin{assumption}[Additive interference] \normalfont \singlespacing
For unit $j\in \mathcal{S}_k$ and $\bm{a}_k = (a_{k1},...,a_{kn_k})^\top\in \{0,1\}^{n_k}$, the potential outcome $Y_{kj}(\bm{a}_k)$ satisfies
$Y_{kj}(\bm{a}_k) = \beta^{(0)}_{kj} + \sum_{i=1}^{n_k}\beta^{(i)}_{kj}a_{ki},$
where $\tilde{\bm{\beta}}_{kj}  = (\beta^{(0)}_{kj},\beta^{(1)}_{kj},...,\beta^{(n_k)}_{kj})^\top$ is a vector of unknown constants. 
\label{assump_additive}
\end{assumption}
Assumption~\ref{assump_additive} posits that the potential outcome of unit $j$ in cluster $k$ is additive in the treatment levels of all the intervention units in cluster $k$.
The vector of coefficients, $\tilde{\bm{\beta}}_{kj}$, is completely arbitrary and may depend on the observed covariates vector, $\bm{X}_k$. Since these coefficients can vary across $i$, Assumption~\ref{assump_additive} accommodates differential and flexible spillover effects of the intervention units on unit $j$'s outcome.
In particular, if $\beta^{(i)}_{kj} = c_{kj}$ for all $i \neq i^*$ and some arbitrary constants $c_{kj}$, this assumption becomes analogous to stratified interference. Finally, Assumption~\ref{assump_additive} can be further generalized to incorporate more complex interference patterns, such as interactions among the treatment levels of the intervention units (see \citealt{zhang2023individualized}).

We now turn to estimating the design-based variances of our proposed estimators under additive interference. For brevity, we focus on the variance of $\hat{\mu}^\pi_{a,\text{HT}}$ ($a \in \{0,1\}$). See Appendix~\ref{appsec_var_additive} in the Appendix for related discussions concerning the variances of the other estimators. 
Proposition~\ref{prop_additive_mu1} provides closed-form expression of the variance of $\hat{\mu}^\pi_{a,\text{HT}}$.
\begin{proposition}[Variance under additive interference] \normalfont \singlespacing
 Under Assumptions~\ref{assump_partial}--\ref{assump_indep} ~and~\ref{assump_additive},  
\begin{align*}
\Var(\hat{\mu}^\pi_{a,\text{HT}}) =   \frac{1}{K^2}\sum_{k = 1}^{K} \frac{1}{|\mathcal{S}_k|^2}\left(\sum_{j \in \mathcal{S}_k} \Lambda_{1,k,j} + \mathop{\sum\sum}_{j \neq j' \in \mathcal{S}_k}\Lambda_{2,k,j,j'} \right),
\end{align*}
where
\begin{align*}
\Lambda_{1,k,j} \ = \ & \sum_{\bm{s}}\frac{\pi^2_{k}(\bm{A}_{k(-i^*)} = \bm{s}\mid A_{ki^*} = a)}{f_k(A_{ki^*} = a, \bm{A}_{k(-i^*)} = \bm{s})}\left\{(1,\bm{a}^\top_{kj})\tilde{\bm{\beta}}_{kj} \right\}^2 - \left\{(1,\bm{\pi}^\top_k(\cdot\mid A_{ki^*} = a))\tilde{\bm{\beta}}_{kj} \right\}^2,
\label{eq_lambda1_add} \\
       \Lambda_{2,k,j,j'} \ = \ &\sum_{\tilde{\bm{s}}}\frac{\mathbbm{\pi}^2_k(A_{ki^{*}} = a, A_{ki^{*'}} = a, \bm{A}_{k(-i^{*},-i^{*'})} = \tilde{\bm{s}}) \left\{(1,\bm{a}^\top_{kjj'})\tilde{\bm{\beta}}_{kj} \right\}\left\{(1,\bm{a}^\top_{kjj'})\tilde{\bm{\beta}}_{kj'} \right\}}{f_k(A_{ki^*} = a, A_{ki^{*'}} = a, \bm{A}_{k(-i^{*},-i^{*'})} = \tilde{\bm{s}})\pi_k(A_{ki^*} = a)\pi_k(A_{ki^{*'}} = a)} \nonumber\\
    &\hspace{0.5cm} - \left\{(1,\bm{\pi}^\top_k(\cdot\mid A_{ki^*} = a))\tilde{\bm{\beta}}_{kj} \right\}\left\{(1,\bm{\pi}^\top_k(\cdot\mid A_{ki^{*'}} = a))\tilde{\bm{\beta}}_{kj'} \right\}. 
\end{align*}
Note that $\bm{a}_{kj}$ is the vector of treatment assignments with $A_{ki^{*}} = a$ and $\bm{A}_{k(-i^{*})} = \bm{s}$; $\bm{a}_{kjj'}$ is the vector of treatment assignments with $A_{ki^{*}} = a, A_{ki^{*'}} = a, \bm{A}_{k(-i^{*},-i^{*'})} = \tilde{\bm{s}}$; $\bm{\pi}_k(\cdot\mid A_{ki^*} = a)$ is the vector of conditional probabilities whose $i$th element is $\pi_k(A_{ki} = 1 \mid A_{ki^*} = a)$.
\label{prop_additive_mu1}
\end{proposition}
Proposition~\ref{prop_additive_mu1} shows that the variance of $\hat{\mu}^\pi_{a,\text{HT}}$ can be written as a quadratic function of the coefficients $\{\tilde{\bm{\beta}}_{kj}, j \in \mathcal{S}_k, k \in \{1,2,...,K\}\}$. 
The variances of the estimated direct and indirect effects can also be expressed as similar quadratic forms; see Propositions~\ref{prop_varDE_additive}~and~\ref{prop_additive_IE} in the Appendix. Thus, an estimator of these variances can be obtained by plugging in appropriate estimators of these coefficients. To this end, we consider the following estimator of $\tilde{\bm{\beta}}_{kj}$:
\begin{equation}
    \hat{\tilde{\bm{\beta}}}_{kj} = \E(\tilde{\bm{A}}_k\tilde{\bm{A}}^\top_k)^{-1}\tilde{\bm{A}}_k Y^{\text{obs}}_{kj},
    \label{eq_beta_hat}
\end{equation}
where $\tilde{\bm{A}}_k = (1,A_{k1},...,A_{kn_k})^\top$.  If $\E(\tilde{\bm{A}}_k\tilde{\bm{A}}^\top_k)$ is not invertible, we use the Moore-Penrose pseudoinverse instead. Under Assumption~\ref{assump_additive}, $Y^{\text{obs}}_{kj} = Y_{kj}(\bm{A}_k) = \tilde{\bm{A}}^\top_k\hat{\tilde{\bm{\beta}}}_{kj}$, and in this sense, $\hat{\tilde{\bm{\beta}}}_{kj}$ can be interpreted as an estimator of the population regression coefficient $\E(\tilde{\bm{A}}_k\tilde{\bm{A}}^\top_k)^{-1}\E(\tilde{\bm{A}}_k Y^{\text{obs}}_{kj})$. Moreover, it is straightforward to see that, $\hat{\tilde{\bm{\beta}}}_{kj}$ is design-unbiased for $\tilde{\bm{\beta}}_{kj}$, i.e., $\E(\hat{\tilde{\bm{\beta}}}_{kj}) = \tilde{\bm{\beta}}_{kj}$, where the expectation is taken with respect to the distribution of $\tilde{\bm{A}}_k$. 

Theorem~\ref{thm_conservative_mu1} shows that the variance estimator, derived by plugging in the estimated coefficients is always conservative. This conservative nature extends to variance estimators for other causal quantities. Specifically, Theorems~\ref{thm_conservative_DE}~and~\ref{thm_conservative_IE} in the Appendix show that the variance estimators of the estimated direct and indirect effects (based on $\hat{\tilde{\bm{\beta}}}_{kj}$) are conservative. 
\begin{theorem}[Conservative variance estimator under additive interference] \normalfont \singlespacing
  Let $\widehat{\Var}(\hat{\mu}^\pi_{a,\text{HT}})$ be the estimator of $\Var(\hat{\mu}^\pi_{a,\text{HT}})$ based on $\hat{\tilde{\bm{\beta}}}_{kj}$. Then, under Assumptions~\ref{assump_partial}--\ref{assump_indep} ~and~\ref{assump_additive},
$$\E\{\widehat{\Var}(\hat{\mu}^\pi_{a,\text{HT}})\} \geq \Var(\hat{\mu}^\pi_{a,\text{HT}}).$$
\label{thm_conservative_mu1}
\end{theorem}
Recall that, under stratified interference, we can obtain an unbiased estimator of $\Var(\hat{\mu}^\pi_{a,\text{HT}})$ (see Theorem~\ref{thm_var1}). In contrast, under additive interference, we lose unbiasedness and end up with an upwardly biased variance estimator. Thus, the additional flexibility in the interference pattern provided by additive interference comes at the cost of a more conservative variance estimation.

The additive interference assumption and the resulting estimation strategy also flexibly incorporate design-based inference for more general causal quantities. To this end, we consider the general target population average potential outcome $\tau^\pi$ in Equation~\eqref{eq:general_est} and its corresponding estimator $\hat{\tau}^{\pi}_{\text{HT}}$. 
Theorem \ref{thm_general_additive} provides a closed-form expression of the variance of $\hat{\tau}^{\pi}_{\text{HT}}$ and shows that the resulting estimated variance based on $\hat{\tilde{{\bm{\beta}}}}_{kj}$ is conservative.
\begin{theorem}[Variance of a general estimator]\normalfont \singlespacing
Consider the Horvitz-Thompson estimator $\hat{\tau}^{\pi}_{\text{HT}}$ of $\tau^\pi$, defined in Equation~\eqref{eq:general_est}. Under Assumptions~\ref{assump_partial}--\ref{assump_indep} ~and~\ref{assump_additive},  
\begin{align*}
\Var(\hat{\tau}^\pi_{\text{HT}}) =   \frac{1}{K^2}\sum_{k = 1}^{K} \frac{1}{|\mathcal{S}_k|^2}\left[\sum_{\bm{a} \in \text{Supp}(f_k)}f_k(\bm{a}) \{1 - f_k(\bm{a})\}\zeta^2_k(\bm{a}) - \mathop{\sum\sum}_{\bm{a} \neq \bm{a}' \in \text{Supp}(f_k)}f_k(\bm{a})f_k(\bm{a}')\zeta_k(\bm{a})\zeta_k(\bm{a}') \right],
\end{align*}
where $\zeta_k(\bm{a}) = \sum_{j \in \mathcal{S}_k}\frac{\pi_{kj}(\bm{a})}{f_k(\bm{a})}(1,\bm{a}^\top)\tilde{\bm{\beta}}_{kj}$. Moreover, let $\widehat{\Var}(\hat{\tau}^\pi_{\text{HT}})$ be the estimator of $\Var(\hat{\tau}^\pi_{\text{HT}})$ based on $\hat{\tilde{\bm{\beta}}}_{kj}$. Then,
$$\E\{\widehat{\Var}(\hat{\tau}^\pi_{\text{HT}})\} \geq \Var(\hat{\tau}^\pi_{\text{HT}}).$$
\label{thm_general_additive}
\end{theorem}


\section{Simulation study}
\label{sec_simulation}

We now evaluate the finite-sample performance of the proposed Horvitz-Thompson and H\'{a}jek estimators in a simulation study. 

\subsection{Setup}

In this study, the target population $\mathcal{S}$ is the population of non-intervention units, corresponding to a bipartite randomized experiment.
We consider two different numbers of clusters, $K = 10, 50$. 
For each $K$, we allocate an equal number of intervention units to each cluster, setting it at $n_k = 32$. 
Similarly, we posit that the number of non-intervention units in each cluster is uniform, but it can take one of four possible values, namely $m_k \in \{50, 100, 250, 500 \}$.

For each intervention unit, we generate two continuous covariates independently from the standard Normal distribution, $W_1, W_2 \sim \mathcal{N}(0,1)$.
We also incorporate a binary covariate, $W_3$, which equals one for exactly half of the units within each cluster.
These covariates serve as the basis for building the potential outcomes using the following two models,
\begin{itemize}
\item[{\bf M1:}] $Y_{kj}(A_{ki^*} = a, A_{k(-i*)} = \bm{s}) = 5 - 2.5a - 1.5p_k + W_{1,ki^*} - 0.5W_{2,ki^*} + 3W_{3,ki^*} +  ap_k$, 
\item [\bf{M2:}] $Y_{kj}(A_{ki^*} = a, A_{k(-i*)} = \bm{s}) = 5 - 2.5a - 1.5p_k + W_{1,ki^*} - 0.5W_{2,ki^*} + 3W_{3,ki^*} + ap_k + 2(W_{1,ki^*} + W_{2,ki^*})a$, 
\end{itemize}
where $p_k (=0.5)$ is the proportion of treated units in cluster $k$.
Finally, in each cluster, we set the actual intervention $f_k(\cdot)$ to correspond to complete randomization with equal allocation and consider two stochastic interventions: $\pi^{(1)}_{k}(\cdot)$, which corresponds to complete randomization with equal allocation (i.e., $\pi^{(1)}_{k}(\cdot) = f_{k}(\cdot)$), and $\pi^{(2)}_{k}(\cdot)$ which corresponds to stratified randomization with equal allocation within strata defined by $W_3$.
\subsection{Results}

Figures~\ref{fig:simu_mu1}~and~\ref{fig:simu_DE} display the bias, standard error (SE), and coverage of the 95\% confidence intervals for the Horvitz-Thompson and H\'{a}jek estimators of $\mu^\pi_1$ and $\text{DE}^\pi$, under stochastic intervention, $\pi^{(1)}(\cdot)$ and $\pi^{(2)}(\cdot)$, respectively. 
The corresponding measures under $\pi^{(2)}(\cdot)$ are shown in Figures~\ref{fig:simu_mu1_strat}~and~\ref{fig:simu_DE_strat} in the Appendix. The coverages are computed under stratified intervention. 

Regarding bias, the Horvitz-Thompson estimator is design-unbiased across all scenarios, which is reflected in the simulation results. In general, the H\'{a}jek estimator is not design-unbiased in finite samples. More importantly, the H\'{a}jek estimator is undefined when the observed treatment assignment in each cluster falls outside the support of the stochastic intervention. To alleviate the latter, we rerandomize (i.e., reject the draw and simulate again) until the assignment in at least one cluster falls within the support. Our simulation results indicate that, under this rerandomization scheme, the bias of the H\'{a}jek estimator is close to zero across all scenarios.

When considering the SE, the H\'{a}jek estimators for both $\mu^\pi_1$ and $\text{DE}^\pi$ consistently outperform the corresponding Horvitz-Thompson estimators across all scenarios. 
The difference in SE between the Horvitz-Thompson and H\'{a}jek estimators for each estimand and stochastic intervention is especially noticeable for smaller values of $K$ and $m_k$. 
As expected, the SE for each estimator tends to decrease as $K$ or $m_k$ increases.  
Furthermore, this difference in SE is more pronounced under $\pi^{(2)}_k$ compared to $\pi^{(1)}_k$ for each estimand. 
This finding indicates that the H\'{a}jek estimator is more precise when the stochastic intervention deviates from the actual intervention.

Regarding coverage, the Horvitz-Thompson estimator of $\mu^\pi_1$ exhibits coverage that is nearly at the nominal level of 95\% across the two outcome models.
This result is expected because under stratified interference, our variance estimator is unbiased.
When the stochastic intervention is $\pi^{(1)}_k$, the coverage for $\text{DE}^\pi$ is closer to the nominal level under model M1 than under M2. 
This difference arises because M1 assumes homogeneous treatment effects (i.e., $\tilde{Y}_{ki}(1,p_k) - \tilde{Y}_{ki}(0,p_k)$ is constant), implying that the variance estimator is unbiased (see Proposition~\ref{prop_htvar1}). 
Under M2, however, treatment effects are heterogeneous, resulting in a conservative variance estimator.

When the stochastic intervention is $\pi^{(2)}_k$, the coverage for $\text{DE}^\pi$ is near the nominal level under both M1 and M2.
For the H\'{a}jek estimator of $\mu^\pi_1$, the coverage is approximately at the nominal level, with a few exceptions where the number of clusters is small (see Figure~\ref{fig:simu_mu1_strat}). 
This is reasonable because the variance estimator for the H\'{a}jek estimator is based on asymptotic approximations with a  large number of clusters.
Nonetheless, regardless of the number or size of the clusters, for $\text{DE}^\pi$, the H\'{a}jek estimator tends to be conservative, with coverages nearing 100\%.

\begin{figure}[!t]
    \centering 
    \includegraphics[scale = 0.53]{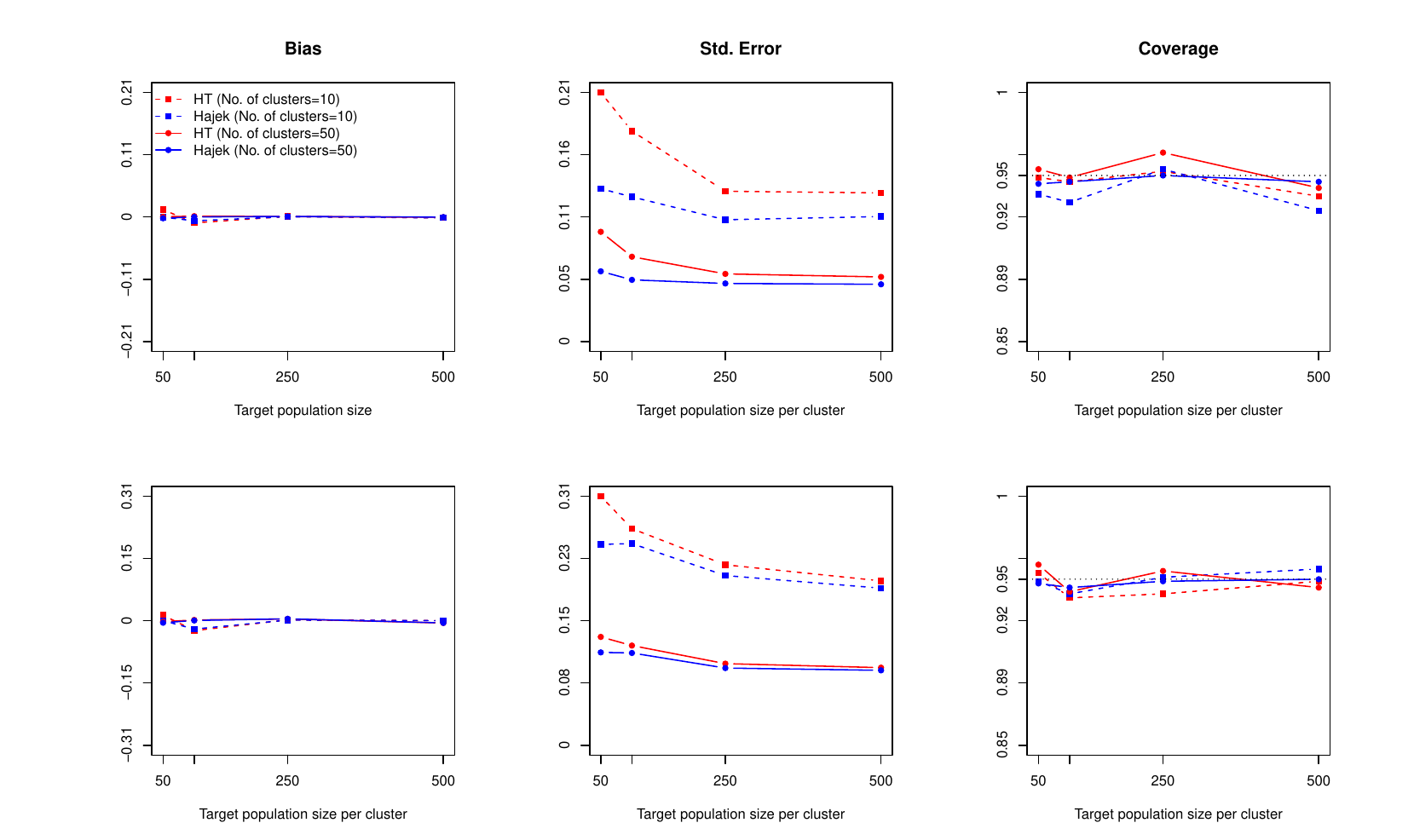}
    \caption{Bias, standard error, and coverage of 95\% confidence intervals for the Horvitz-Thompson and H\'{a}jek estimators of $\mu^\pi_1$ under outcome models M1 and M2 and stochastic intervention $\pi^{(1)}_k(\cdot)$. The first and second row correspond to models M1 and M2, respectively.}
    \label{fig:simu_mu1}
  \end{figure}
  
\begin{figure}[!t]
    \centering 
    \includegraphics[scale = 0.53]{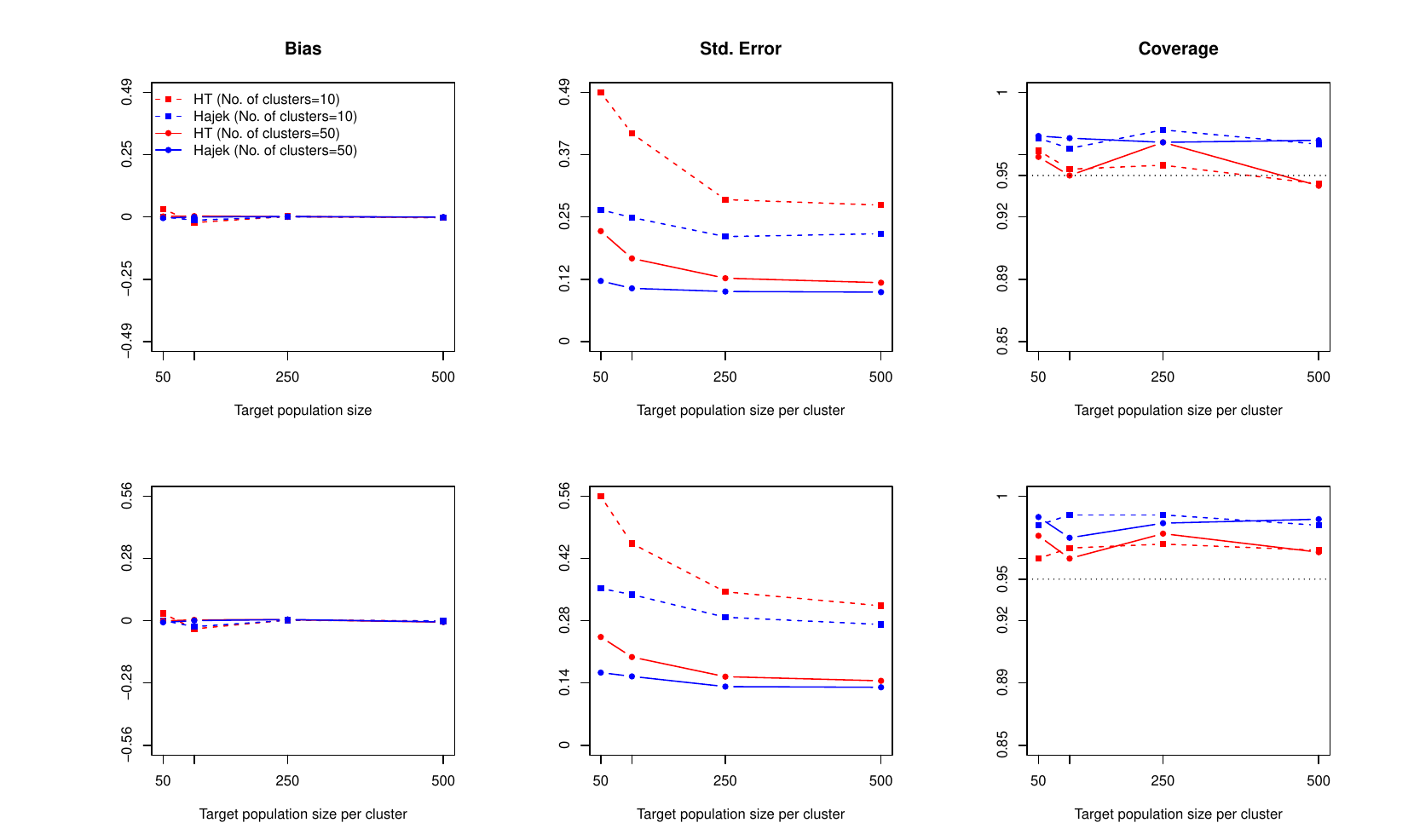}
    \caption{Bias, standard error, and coverage of 95\% confidence intervals for the Horvitz-Thompson and H\'{a}jek estimators of $\text{DE}^\pi$ under outcome models M1 and M2 and stochastic intervention $\pi^{(1)}_k(\cdot)$. The first and second row correspond to models M1 and M2, respectively.}
    \label{fig:simu_DE}
\end{figure}
\section{Empirical application}
\label{sec_empirical}

In this section, we implement our proposed inferential methods using the dataset from the school conflict experiment introduced in Section~\ref{sec_theeffect}.
Our analysis focuses on two indicators of \textit{awareness} about conflict and one outcome about the \textit{instances} of conflict: talking about conflict (yes or no), wearing anti-conflict wristbands (yes or no), and the number of conflict incidents. 

The main questions of interest are: (i) On average, what is the effect of each seed-eligible student's \textit{own} treatment status on their subsequent conflict behavior? (ii) What is the average effect of the treatment status of the \textit{best seed-eligible friend} of each seed-ineligible student on their conflict behavior? (iii) How do the above effects vary with different \textit{proportions} of referent students receiving treatment? (iv) What is the average effect of treating a fixed proportion (e.g., 0.5) of the \textit{close} seed-eligible friends of each student on their subsequent conflict behavior?

To address question~(i), we define the target population as all seed-eligible students in the network, and for every seed-eligible student $j$, we designate their key-intervention unit as themselves, i.e., $i^*(j) = j$. 
To address (ii), we define the target population as all seed-ineligible students in the network.  For every seed-ineligible student $j$, we designate their key-intervention unit as their self-reported closest seed-eligible friend $i^*(j)$. 
In both cases, we set the stochastic intervention $\pi_{k}(\cdot)$ to the actual intervention $f_k(\cdot)$.
Table~\ref{tab:q1} reports the point estimates, SEs, and 95\% confidence intervals for $\mu^\pi_1$ and $\text{DE}^\pi$ under stratified interference across both scenarios. The corresponding values under additive interference are provided in Table~\ref{tab:q1_additive} of the Appendix.

\begin{singlespacing}
\begin{table}[!t]
\centering
\scalebox{0.75}{
\begin{tabular}{cccccccc}
\toprule
  &   &  \multicolumn{3}{c}{Seed-ineligible population} &  \multicolumn{3}{c}{Seed-eligible population}\\
  \cline{3-8}
Outcome &  & Estimate & Std. Error & 95\% CI & Estimate & Std. Error & 95\% CI \\
\hline
\multirow{4}{*}{Talking about conflict} & $\hat{\mu}^{\pi}_{1,\text{HT}}$   & 0.41     & 0.01  & (0.38, 0.43) & 0.44 & 0.01 & (0.42, 0.47)    \\
& $\hat{\mu}^{\pi}_{1,\text{H\'{a}jek}}$   & 0.40  & 0.01 & (0.39, 0.41)  & 0.44 & 0.01 & (0.42, 0.47)  \\
& $\widehat{\text{DE}}^{\pi}_{\text{HT}}$  & 0.04   & 0.02    & (-0.01, 0.08) &  0.07 & 0.03 & (0.02, 0.12)   \\
& $\widehat{\text{DE}}^{\pi}_{\text{H\'{a}jek}}$   & 0.02   & 0.03  & (-0.04, 0.07) & 0.07 & 0.03 & (0.02, 0.12)  \\
\hline
\multirow{4}{*}{Wearing anti-conflict wristbands} & $\hat{\mu}^{\pi}_{1,\text{HT}}$     & 0.19               & 0.01       & (0.18, 0.21)  & 0.28 & 0.01 & (0.26, 0.30)    \\
& $\hat{\mu}^{\pi}_{1,\text{H\'{a}jek}}$  & 0.19   & 0.01  & (0.17, 0.20) & 0.28 & 0.01 & (0.26, 0.30)    \\
& $\widehat{\text{DE}}^{\pi}_{\text{HT}}$    & 0.03  & 0.01   & (-0.002, 0.05) & 0.14 & 0.02 & (0.10, 0.18)  \\
 & $\widehat{\text{DE}}^{\pi}_{\text{H\'{a}jek}}$  & 0.02   & 0.02   & (-0.02, 0.05)  & 0.14 & 0.02 & (0.10, 0.18)   \\
\hline
\multirow{4}{*}{Cases of conflict  }  & $\hat{\mu}^{\pi}_{1,\text{HT}}$   & 0.16               & 0.01       & (0.14, 0.18) & 0.15 & 0.02  & (0.11, 0.18)    \\
& $\hat{\mu}^{\pi}_{1,\text{H\'{a}jek}}$   & 0.16               & 0.01       & (0.14,0.18)  & 0.15 & 0.02 & (0.11, 0.18)   \\
& $\widehat{\text{DE}}^{\pi}_{\text{HT}}$  & 0.00               & 0.02       & (-0.04, 0.05)  & 0.00 & 0.03 & (-0.07, 0.07)  \\
& $\widehat{\text{DE}}^{\pi}_{\text{H\'{a}jek}}$  & -0.01  & 0.02   & (-0.05, 0.04)  & 0.00 & 0.03 & (-0.07, 0.07)  \\
\bottomrule
\end{tabular}
}
\caption{Estimates, standard errors (SE) and $95\%$ confidence intervals (CI) of the average potential outcomes and direct effects under stratified interference for the two target populations, where the stochastic intervention equals the actual intervention.}
\label{tab:q1}
\end{table}
\end{singlespacing}

Table~\ref{tab:q1} shows that the point estimates and SEs for the Horvitz-Thompson and H\'{a}jek estimators are similar across all scenarios. 
This finding aligns with those from the simulation study in Section \ref{sec_simulation}, where both the Horvitz Thompson and H\'{a}jek estimators performed similarly for large $m_k$.
When comparing the point estimates of $\mu^\pi_1$ across the two target populations, we find that under the intervention, the overall level of anti-conflict activities (such as talking about conflict and wearing anti-conflict wristbands) is higher in the seed-eligible population than in the ineligible population. 
A similar pattern is noted for estimates of $\text{DE}^\pi$. 

These patterns intuitively make sense because the intervention is expected to have a more pronounced effect on students directly involved (i.e., the eligible students) than on their friends who are not eligible. 
However, this pattern is not as apparent when considering the instances of conflict. 
Finally, the confidence intervals for $\text{DE}^\pi$ indicate that while the intervention on the key-intervention unit increases awareness about conflict (at least among the seed-eligible students), it does not significantly decrease the actual instances of conflict. 

Next, we address question (iii) by incorporating information on the referent students in our causal estimands. 
For school $k$, we consider a stochastic intervention $\pi_k(\cdot)$ that treats a fixed proportion $\alpha$ of referent students (see the Appedix for details). 
With varying values of $\alpha$, namely $0.1$, $0.3$, $0.5$, $0.7$, and $0.9$, we plot the corresponding point estimates and 95\% confidence intervals for $\mu^\pi_1$ and $\text{DE}^\pi$ under stratified interference in Figures \ref{fig:mu1_estimates} and \ref{fig:DE_estimates}, respectively. The corresponding plots under additive interference are provided in Figures \ref{fig:mu1_estimates_additive} and \ref{fig:DE_estimates_additive} in the Appendix.

\begin{figure}[!t]
    \centering 
    \includegraphics[scale = 0.7]{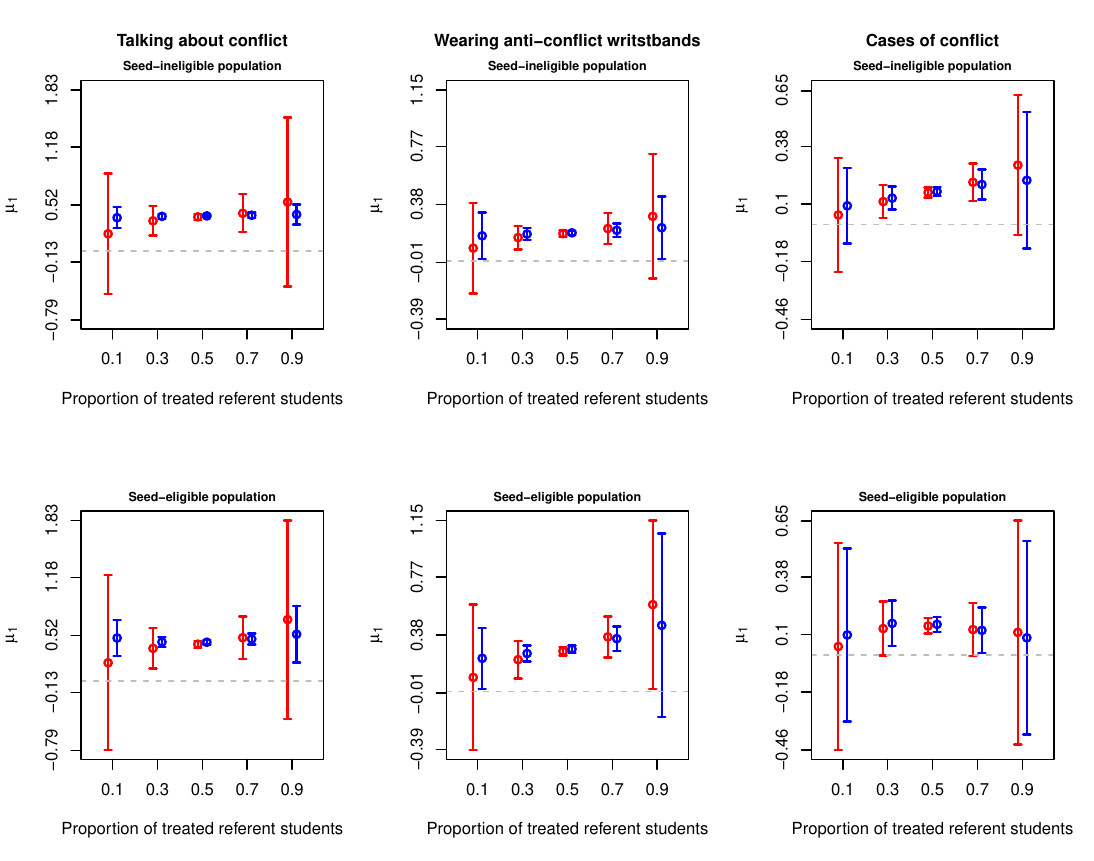}
    \caption{Point estimates and $95\%$ confidence intervals under stratified interference for the Horvitz-Thompson (red) and H\'{a}jek (blue) estimators of $\mu^\pi_1$ for two target populations.}
    \label{fig:mu1_estimates}
  \end{figure}
  
\begin{figure}[!t]
    \centering 
    \includegraphics[scale = 0.7]{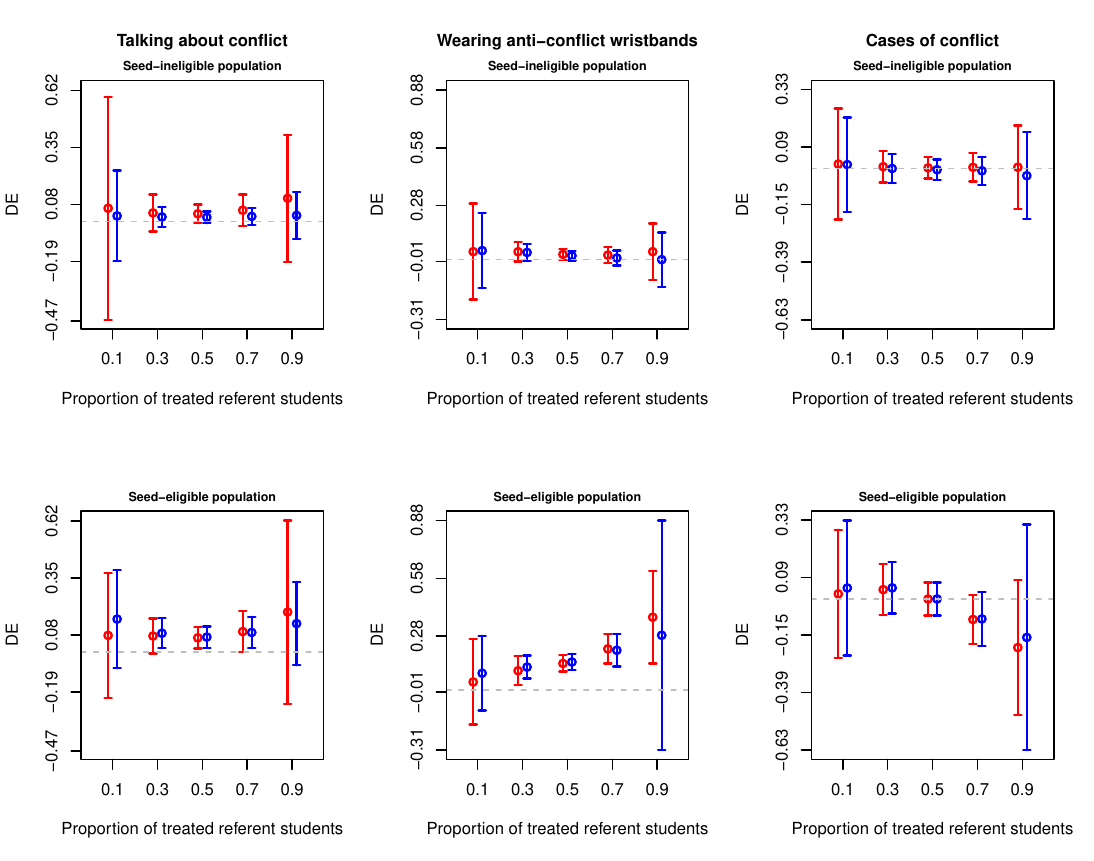}
    \caption{Point estimates and $95\%$ confidence intervals under stratified interference for the Horvitz-Thompson (red) and H\'{a}jek (blue) estimators of $\text{DE}^\pi$ for two target populations. }
    \label{fig:DE_estimates}
  \end{figure}
  
Figures~\ref{fig:mu1_estimates}~and~\ref{fig:DE_estimates} show that the point estimates and SEs of the Horvitz-Thompson and H\'{a}jek estimators exhibit more pronounced differences compared to the previous $\pi_k(\cdot) = f_k(\cdot)$ scenario.  
Generally, the H\'{a}jek estimator yields smaller SEs than the Horvitz-Thompson estimator, leading to narrower confidence intervals. 
The contrast in the overall performance of the estimators for the seed-eligible and ineligible target populations is similar to the previous scenario.

Furthermore, as the proportion of treated referent student $\alpha$ increases, the estimators of $\mu^\pi_1$ corresponding to conflict awareness tend to increase. 
However, the estimators linked to instances of conflict do not decrease as $\alpha$ increases; in fact, they tend to slightly increase among the seed-ineligible population. 
Our findings align with those in \cite{paluck2016changing}, which showed that the peer-to-peer social influence effects of the referent seeds on the instances of conflict are not significant.
Also, the confidence intervals for $\text{DE}^{\pi}$ suggest that while in some cases there are significant (positive) direct effects of the intervention on conflict awareness (e.g., on wearing wristbands with $\alpha = 0.7$), the effects on instances of conflict are not significant.

Across all scenarios, the standard errors under additive interference are uniformly larger than those under stratified interference, leading to wider confidence intervals; see Table~\ref{tab:q1_additive} and Figures~\ref{fig:mu1_estimates_additive}~and~\ref{fig:DE_estimates_additive} in the Appendix. These observations align with the results in Section~\ref{sec_additive}, which indicate that while additive interference allows for more flexible and complex patterns of spillover in the network compared to stratified interference, the resulting variance estimators are more conservative. Nevertheless, similar to stratified interference, the estimated standard errors under additive interference across different scenarios are relatively small when the stochastic intervention $\pi_k(\cdot)$ resembles the actual intervention $f_k(\cdot)$, i.e., in scenarios with $\alpha$ close to $0.5$.

Finally, to address question (iv), we define the target population as the set of \textit{all} students, and for each student $j$, we consider their self-reported close friends (up to 10) as the set of key-intervention units. 
If $j$ is a seed-eligible student, we include them in the set of key-intervention units. 
We consider the stochastic intervention as in Section \ref{sec_estimand} (see the Appendix for details). 
With varying values of $p^*$, namely $0.1$, $0.3$, $0.5$, $0.7$, and $0.9$, we depict the corresponding point estimates and 95\% confidence intervals for $\tau^\pi$ (as defined in Section \ref{sec_estimand}) in Figure \ref{fig:tau_estimates}.

\begin{figure}[!t]
    \centering
    \includegraphics[scale = 0.7]{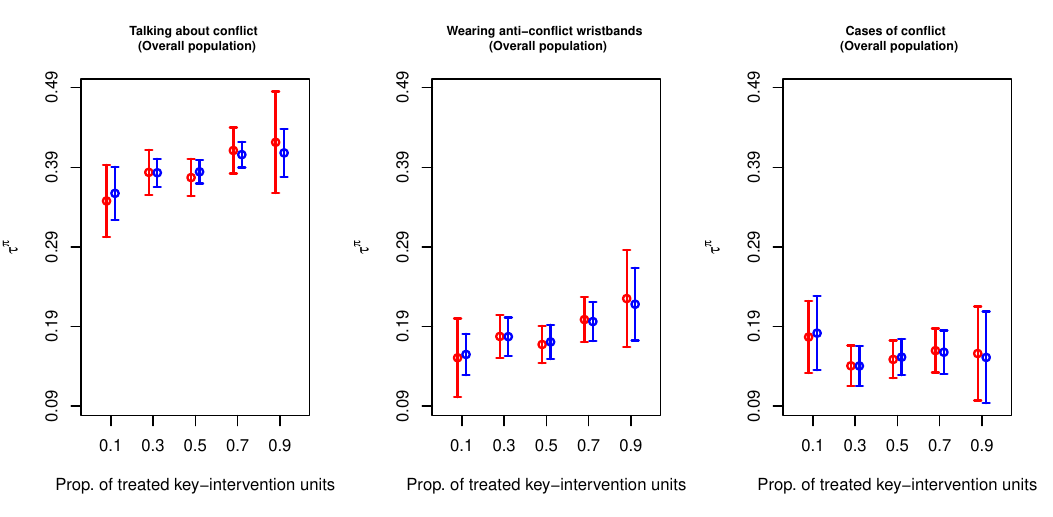}
    \caption{Point estimates and $95\%$ confidence intervals for the Horvitz-Thompson (red) and H\'{a}jek (blue) estimators of $\tau^\pi$ with multiple key-intervention units. }
    \label{fig:tau_estimates}
  \end{figure}
  
Figure~\ref{fig:tau_estimates} shows that, similar to the previous cases, the H\'{a}jek estimator typically yields narrower confidence intervals for $\tau^\pi$ compared to the Horvitz-Thompson estimator. 
Furthermore, as the proportion of treated close friends increases, the average levels of conflict awareness also tend to increase. 
The average instances of conflict initially decrease but then plateau as $p^*$ is increased from 0.5 to 0.9.

In summary, the results suggest that treating a higher proportion of close friends of each student could be beneficial in enhancing awareness about conflict behaviors; however, it may not lead to a reduction in the actual cases of conflict beyond a certain threshold.

\section{Concluding remarks}
\label{sec_summary}

In this paper, we established a design-based framework for the analysis of generalized network experiments, accommodating arbitrary interference and arbitrary target populations in the network.
We introduced a class of causal estimands using stochastic interventions and proposed Horvitz-Thompson and H\'{a}jek estimators under general interference. 
We addressed the challenge of identifying the design-based variances of these estimators by developing their conservative estimators under certain assumptions on interference.

We implemented the proposed estimation methods in a simulated experiment and a real-world experiment focused on anti-conflict interventions in schools. 
Both studies suggested that the H\'{a}jek estimators tend to produce more precise estimates of causal effects than the Horvitz-Thompson estimators. 
Our analysis of the school-conflict experiment revealed that intervening on a higher proportion of close friends or referent (i.e., influential) students increases awareness regarding conflict on average, though it does not significantly reduce the average number of conflict cases in schools.

The proposed framework for generalized network experiments can be extended to incorporate more complex estimands and assignment mechanisms. 
For instance, one could consider treatment assignment mechanisms (both counterfactual and actual) that are dependent across clusters, such as two-stage randomized experiments \citep{hudgens2008toward}. 
Potential extensions of this framework include the derivation of large-sample properties of the proposed estimators under weaker assumptions on interference \citep{savje2021average}.
 
\bibliographystyle{asa}
\bibliography{main}.

\clearpage
\appendix
\section*{Supplementary Materials} 

\addcontentsline{toc}{section}{Supplementary Materials}
\renewcommand{\thesubsection}{\Alph{subsection}}
\setcounter{table}{0}
\renewcommand{\thetable}{A\arabic{table}}
\setcounter{figure}{0}
\renewcommand\thefigure{A\arabic{figure}}
\setcounter{theorem}{0}
\renewcommand\thetheorem{A\arabic{theorem}}
\setcounter{lemma}{0}
\renewcommand\thelemma{A\arabic{lemma}}
\setcounter{equation}{0}
\renewcommand\theequation{A\arabic{equation}}

\subsection{Proofs of propositions and theorems}

\subsubsection{Proof of Theorem \ref{thm_HT_unbiased}}

\begin{align*}
\E(\hat{\tau}^\pi_{\text{HT}}) &= \E\left[\frac{1}{K}\sum_{k=1}^{K}\left\{\frac{1}{|\mathcal{S}_k|}\sum_{j \in \mathcal{S}_k}\frac{\pi_{kj}(\bm{A}_k)}{f_k(\bm{A}_k)}Y^{\text{obs}}_{kj}\right\}\right], \nonumber \\
&=  \frac{1}{K}\sum_{k=1}^{K}\frac{1}{|\mathcal{S}_k|}\sum_{j \in \mathcal{S}_k} \E\left\{\sum_{\bm{a}}\mathbbm{1}(\bm{A}_k = \bm{a})  \frac{\pi_{kj}(\bm{a})}{f_k(\bm{a})}Y_{kj}(\bm{a})\right\}, \nonumber \\
& = \frac{1}{K}\sum_{k=1}^{K}\frac{1}{|\mathcal{S}_k|}\sum_{j \in \mathcal{S}_k} \sum_{\bm{a}}\E\{\mathbbm{1}(\bm{A}_k = \bm{a})\}  \frac{\pi_{kj}(\bm{a})}{f_k(\bm{a})}Y_{kj}(\bm{a}) \nonumber,\\
& = \frac{1}{K}\sum_{k=1}^{K}\frac{1}{|\mathcal{S}_k|}\sum_{j \in \mathcal{S}_k} \sum_{\bm{a}}f_k(\bm{a})  \frac{\pi_{kj}(\bm{a})}{f_k(\bm{a})}Y_{kj}(\bm{a}) \\
& = \tau^\pi. 
\end{align*}
\qed

\subsubsection{Proof of Theorem \ref{thm_generalvar}}

First, we can write $\hat{\mu}^\pi_{a,\text{HT}}$ as
\begin{align}
    \hat{\mu}^\pi_{a,\text{HT}} & = \frac{1}{K}\sum_{k=1}^{K}\frac{1}{|\mathcal{S}_k|}\sum_{j \in \mathcal{S}_k}\sum_{\bm{s}}\mathbbm{1}\{A_{ki^*} = a, \bm{A}_{k(-i^*)} = \bm{s}\}\frac{\pi_{k}(\bm{A}_{k(-i^*)} = \bm{s} \mid A_{ki^*} = a)}{f_k(A_{ki^*} = a, \bm{A}_{k(-i^*)} =\bm{s})} \nonumber\\
    & \hspace{4cm} \times Y_{kj}(A_{ki^*} = a, \bm{A}_{k(-i^*)} =\bm{s}).
    \label{eq_A_muahat}
\end{align}
Thus, the variance of $\hat{\mu}^\pi_{a,\text{HT}}$ can be written as,
\begin{align*}
&\Var(\hat{\mu}^\pi_{a,\text{HT}}) \nonumber \\
= & \frac{1}{K^2}\sum_{k=1}^{K}\frac{1}{|\mathcal{S}_k|^2}\Var\left(\sum_{j \in \mathcal{S}_k}\sum_{\bm{s}}\mathbbm{1}\{A_{ki^*} = a, \bm{A}_{k(-i^*)} = \bm{s}\}\frac{\pi_{k}(\bm{A}_{k(-i^*)} = \bm{s} \mid A_{ki^*} = a)}{f_k(A_{ki^*} = a, \bm{A}_{k(-i^*)} =\bm{s})} \right. \nonumber\\
& \hspace{4cm}  \times Y_{kj}(A_{ki^*} = a, \bm{A}_{k(-i^*)} =\bm{s}) \Bigg).
\end{align*}
The variance term inside the first summation can be decomposed as $\sum_{j \in \mathcal{S}_k} \Lambda_{1,k,j} + \mathop{\sum\sum}_{j \neq j' \in \mathcal{S}_k}\Lambda_{2,k,j,j'}$, where
\begin{align*}
    &\Lambda_{1,k,j} \nonumber\\
     = & \Var\left\{ \sum_{\bm{s}}\mathbbm{1}(A_{ki^*} = a, \bm{A}_{k(-i^*)} = \bm{s})\frac{\pi_{k}(\bm{A}_{k(-i^*)} = \bm{s} \mid A_{ki^*} = a)}{f_k(A_{ki^*} = a, \bm{A}_{k(-i^*)} =\bm{s})}Y_{kj}(A_{ki^*} = a, \bm{A}_{k(-i^*)} =\bm{s})  \right\} \nonumber \\
     = & \sum_{\bm{s}}\Var\left\{ \mathbbm{1}(A_{ki^*} = a, \bm{A}_{k(-i^*)} = \bm{s})\right\}\frac{\pi^2_{k}(\bm{A}_{k(-i^*)} = \bm{s} \mid A_{ki^*} = a)}{f^2_k(A_{ki^*} = a, \bm{A}_{k(-i^*)} =\bm{s})}Y^2_{kj}(A_{ki^*} = a, \bm{A}_{k(-i^*)} =\bm{s}) \nonumber \\
    & + \mathop{\sum\sum}_{\bm{s} \neq \bm{s}'} \Cov\left\{\mathbbm{1}(A_{ki^*} = a, \bm{A}_{k(-i^*)} = \bm{s}), \mathbbm{1}(A_{ki^*} = a, \bm{A}_{k(-i^*)} = \bm{s}')\right\}\frac{\pi_{k}(\bm{A}_{k(-i^*)} = \bm{s} \mid A_{ki^*} = a)}{f_k(A_{ki^*} = a, \bm{A}_{k(-i^*)} =\bm{s})} \nonumber\\
    & \hspace{.75in} \times \frac{\pi_{k}(\bm{A}_{k(-i^*)} = \bm{s}' \mid A_{ki^*} = a)}{f_k(A_{ki^*} = a, \bm{A}_{k(-i^*)} =\bm{s}')}Y_{kj}(A_{ki^*} = a, \bm{A}_{k(-i^*)} =\bm{s})Y_{kj}(A_{ki^*} = a, \bm{A}_{k(-i^*)} =\bm{s}'). \nonumber\\
     = & \sum_{\bm{s}}\frac{\pi^2_{k}(\bm{A}_{k(-i^*)} = \bm{s}|A_{ki^*} = a) \{1 - f_k(A_{ki^*} = a, \bm{A}_{k(-i^*)} = \bm{s})\}}{f_k(A_{ki^*} = a, \bm{A}_{k(-i^*)} = \bm{s})}Y^2_{kj}(A_{ki^*} = a, \bm{A}_{k(-i^*)} = \bm{s}) \nonumber\\
 & - \mathop{\sum\sum}_{\bm{s} \neq \bm{s}'} \pi_{k}(\bm{A}_{k(-i^*)} = \bm{s} \mid A_{ki^*} = a) \pi_{k}(\bm{A}_{k(-i^*)} = \bm{s}' \mid A_{ki^*} = a) \nonumber \\
 & \hspace{1in} \times Y_{kj}(A_{ki^*} = a, \bm{A}_{k(-i^*)} =\bm{s}) Y_{kj}(A_{ki^*} = a, \bm{A}_{k(-i^*)} =\bm{s}') \nonumber\\
 = & \sum_{\bm{s}}\frac{\pi^2_{k}(\bm{A}_{k(-i^*)} = \bm{s}|A_{ki^*} = a)}{f_k(A_{ki^*} = a, \bm{A}_{k(-i^*)} = \bm{s})}Y^2_{kj}(A_{ki^*} = a, \bm{A}_{k(-i^*)} = \bm{s}) \nonumber \\
 & \quad - \left\{\sum_{\bm{s}}\pi_{k}(\bm{A}_{k(-i^*)} = \bm{s}|A_{ki^*} = a)Y_{kj}(A_{ki^*} = a, \bm{A}_{k(-i^*)} = \bm{s}) \right\}^2
 \end{align*} 
and
\begin{align*}
&\Lambda_{2,k,j,j'}  \nonumber \\
 = & \Cov\left(\sum_{\bm{s}}\mathbbm{1}(A_{ki^*} = a, \bm{A}_{k(-i^*)} = \bm{s})\frac{\pi_{k}(\bm{A}_{k(-i^*)} = \bm{s} \mid A_{ki^*} = a)}{f_k(A_{ki^*} = a, \bm{A}_{k(-i^*)} =\bm{s})}Y_{kj}(A_{ki^*} = a, \bm{A}_{k(-i^*)} =\bm{s}),\right. \nonumber \\  
& \hspace{.5in} \left.\sum_{\bm{s}}\mathbbm{1}(A_{ki^{*'}} = a, \bm{A}_{k(-i^{*'})} = \bm{s})\frac{\pi_{k}(\bm{A}_{k(-i^{*'})} = \bm{s} \mid A_{ki^{*'}} = a)}{f_k(A_{ki^{*'}} = a, \bm{A}_{k(-i^{*'})} =\bm{s})}Y_{kj'}(A_{ki^{*'}} = a, \bm{A}_{k(-i^{*'})} =\bm{s})\right),  \nonumber \\
= & \sum_{\bm{s}}\pi_{k}(\bm{A}_{k(-i^*)} = \bm{s}\mid A_{ki^*} = a) \pi_{k}(\bm{A}_{k(-i^{*'})} = \bm{s}\mid A_{ki^{*'}} = a) Y_{kj}(A_{ki^*} = a, \bm{A}_{k(-i^*)} = \bm{s})\nonumber \\
& \quad \times  Y_{kj'}(A_{ki^{*'}} = a, \bm{A}_{k(-i^{*'})} = \bm{s})\left( \frac{\mathbbm{1}\{A_{ki^*} = a,A_{ki^{*'}} = a, \bm{A}_{k(-i^*)} = \bm{s}, \bm{A}_{k(-i^{*'})} = \bm{s}\}}{f_k(A_{ki^*} = a, \bm{A}_{k(-i^*)} = \bm{s})} - 1 \right) \nonumber \\
& - \mathop{\sum\sum}_{\bm{s} \neq \bm{s}'}\pi_{k}(\bm{A}_{k(-i^*)} = \bm{s}\mid A_{ki^*} = a) \pi_{k}(\bm{A}_{k(-i^{*'})} = \bm{s}\mid A_{ki^{*'}} = a)\\ 
&\hspace{.75in} \times Y_{kj}(A_{ki^*} = a, \bm{A}_{k(-i^*)} = \bm{s})Y_{kj'}(A_{ki^*} = a, \bm{A}_{k(-i^{*'})} = \bm{s}') \nonumber \\
= & \sum_{\bm{s}}\pi_{k}(\bm{A}_{k(-i^*)} = \bm{s}\mid A_{ki^*} = a) \pi_{k}(\bm{A}_{k(-i^{*'})} = \bm{s}\mid A_{ki^{*'}} = a) Y_{kj}(A_{ki^*} = a, \bm{A}_{k(-i^*)} = \bm{s})\nonumber \\
& \quad \times Y_{kj'}(A_{ki^{*'}} = a, \bm{A}_{k(-i^{*'})} = \bm{s})\frac{\mathbbm{1}\{A_{ki^*} = a,A_{ki^{*'}} = a, \bm{A}_{k(-i^*)} = \bm{s}, \bm{A}_{k(-i^{*'})} = \bm{s}\}}{f_k(A_{ki^*} = a, \bm{A}_{k(-i^*)} = \bm{s})} \nonumber \\
& - \left\{\sum_{\bm{s}}\pi_{k}(\bm{A}_{k(-i^*)} = \bm{s}\mid A_{ki^*} = a)Y_{kj}(A_{ki^*} = a, \bm{A}_{k(-i^*)} = \bm{s}) \right\} \nonumber \\ 
& \times \left\{\sum_{\bm{s}}\pi_{k}(\bm{A}_{k(-i^{*'})} = \bm{s}\mid A_{ki^{*'}} = a)Y_{kj'}(A_{ki^{*'}} = a, \bm{A}_{k(-i^{*'})} = \bm{s}) \right\}. \nonumber\\
 = & \sum_{\tilde{\bm{s}}}\frac{\pi^2_k(A_{ki^{*}} = a, A_{ki^{*'}} = a, \bm{A}_{k(-i^*, -i^{*'})} = \tilde{\bm{s}})}{f_k(A_{ki^{*}} = a, A_{ki^{*'}} = a, \bm{A}_{k(-i^*, -i^{*'})} = \tilde{\bm{s}})\pi_k(A_{ki^*} = a)\pi_k(A_{ki^{*'}} = a)} \nonumber\\
& \quad  \times Y_{kj}(A_{ki^{*}} = a, A_{ki^{*'}} = a, \bm{A}_{k(-i^*,-i^{*'}) } = \tilde{\bm{s}})Y_{kj'}(A_{ki^{*}} = a,A_{ki^{*'}} = a, \bm{A}_{k(-i^*,-i^{*'})} = \tilde{\bm{s}}) \nonumber\\
& - \left\{\sum_{\bm{s}}\pi_{k}(\bm{A}_{k(-i^*)} = \bm{s}\mid A_{ki^*} = a)Y_{kj}(A_{ki^*} = a, \bm{A}_{k(-i^*)} = \bm{s}) \right\} \nonumber \\ 
& \times \left\{\sum_{\bm{s}}\pi_{k}(\bm{A}_{k(-i^{*'})} = \bm{s}\mid A_{ki^{*'}} = a)Y_{kj'}(A_{ki^{*'}} = a, \bm{A}_{k(-i^{*'})} = \bm{s}) \right\}.
\end{align*}
\qed


\subsubsection{Proof of Theorem \ref{thm_var1}}

Without loss of generality, we set $a = 1$. The proof for $a = 0$ is analogous.
Under Assumptions~\ref{assump_stratified}~and~\ref{assump_fixedprop}, $\hat{\mu}^{\pi}_{1,\text{HT}}$ can be written as,
\begin{align*}
    \hat{\mu}^{\pi}_{1,\text{HT}} &= \frac{1}{K}\sum_{k=1}^{K}\frac{1}{|\mathcal{S}_k|}\sum_{j \in \mathcal{S}_k}\sum_{\bm{s}: \frac{\bm{s}\top \bm{1}}{n_k} = p_k} \frac{\mathbbm{1}(A_{ki^*} = 1,\bm{A}_{k(-i^*)} = \bm{s})\pi_k(\bm{A}_{k(-i^*)} = \bm{s}\mid A_{ki^*} = 1)}{f_k(A_{ki*} = 1, \bm{A}_{k(-i^*)} = \bm{s})}Y_{kj}(1,p_k) \nonumber \\
    & = \frac{1}{K}\sum_{k=1}^{K}\frac{1}{|\mathcal{S}_k|}\sum_{i = 1}^{n_k}\left(\sum_{j \in \mathcal{S}_k}\mathbbm{1}(j \leftarrow i) Y_{kj}(1,p_k) \right)\sum_{\bm{s}} \frac{\mathbbm{1}(A_{ki} = 1,\bm{A}_{k(-i)} = \bm{s})\pi_k(\bm{A}_{k(-i)} = \bm{s}\mid A_{ki} = 1)}{f_k(A_{ki} = 1, \bm{A}_{k(-i)} = \bm{s})} \nonumber \\
    & = \frac{1}{K}\sum_{k=1}^{K}\frac{1}{|\mathcal{S}_k|}\sum_{i = 1}^{n_k}\tilde{Y}_{ki}(1,p_k)\sum_{\bm{s}} \frac{\mathbbm{1}(A_{ki} = 1,\bm{A}_{k(-i)} = \bm{s})\pi_k(\bm{A}_{k(-i)} = \bm{s}\mid A_{ki} = 1)}{f_k(A_{ki} = 1, \bm{A}_{k(-i)} = \bm{s})}
\end{align*}
Thus, by Assumption~\ref{assump_indep},
\begin{align*}
&\Var(\hat{\mu}^\pi_{1,\text{HT}})  \nonumber \\
& = \frac{1}{K^2}\sum_{k =1}^{K}\frac{1}{|\mathcal{S}_k|^2}\left[\sum_{i=1}^{n_k}\Var\left\{ \tilde{Y}_{ki}(1,p_k)\sum_{\bm{s}} \frac{\mathbbm{1}(A_{ki} = 1,\bm{A}_{k(-i)} = \bm{s})\pi_k(\bm{A}_{k(-i)} = \bm{s}\mid A_{ki} = 1)}{f_k(A_{ki} = 1, \bm{A}_{k(-i)} = \bm{s})}\right\} \right.\nonumber \\
& \quad + \mathop{\sum\sum}_{i \neq i'}\Cov\left(\tilde{Y}_{ki}(1,p_k)\sum_{\bm{s}} \frac{\mathbbm{1}(A_{ki} = 1,\bm{A}_{k(-i)} = \bm{s})\pi_k(\bm{A}_{k(-i)} = \bm{s}\mid A_{ki} = 1)}{f_k(A_{ki} = 1, \bm{A}_{k(-i)} = \bm{s})}, \right.\nonumber\\
& \hspace{1.25in}\left.\left. \tilde{Y}_{ki'}(1,p_k)\sum_{\bm{s}} \frac{\mathbbm{1}(A_{ki'} = 1,\bm{A}_{k(-i')} = \bm{s})\pi_k(\bm{A}_{k(-i')} = \bm{s}\mid A_{ki'} = 1)}{f_k(A_{ki'} = 1, \bm{A}_{k(-i')} = \bm{s})} \right) \right] \nonumber \\
& = \frac{1}{K^2}\sum_{k =1}^{K}(G_{1k} + G_{2k}),
\end{align*}
where
\begin{align*}
 G_{1k}
 &= \sum_{i=1}^{n_k}\Var\left( \tilde{Y}_{ki}(1,p_k)\sum_{\bm{s}} \frac{\mathbbm{1}(A_{ki} = 1,\bm{A}_{k(-i)} = \bm{s})\pi_k(\bm{A}_{k(-i)} = \bm{s}\mid A_{ki} = 1)}{f_k(A_{ki} = 1, \bm{A}_{k(-i)} = \bm{s})}\right) \nonumber\\
 & = \sum_{i=1}^{n_k} \tilde{Y}^2_{ki}(1,p_k) \left[\sum_{\bm{s}}\frac{\Var\{\mathbbm{1}(A_{ki} = 1,\bm{A}_{k(-i)} = \bm{s})\}\pi^2_k(\bm{A}_{k(-i)} = \bm{s}\mid A_{ki} = 1)}{f^2_k(A_{ki} = 1, \bm{A}_{k(-i)} = \bm{s})} \right.\nonumber \\
 & \hspace{1.5in} + \mathop{\sum\sum}_{\bm{s}\neq \bm{s}'}\Cov\{\mathbbm{1}(A_{ki} = 1,\bm{A}_{k(-i)} = \bm{s}), \mathbbm{1}(A_{ki} = 1,\bm{A}_{k(-i)} = \bm{s}')\} \nonumber\\
 & \hspace{2in} \left. \times \frac{\pi_k(\bm{A}_{k(-i)} = \bm{s}\mid A_{ki} = 1)\pi_k(\bm{A}_{k(-i)} = \bm{s}'\mid A_{ki} = 1)}{f_k(A_{ki} = 1, \bm{A}_{k(-i)} = \bm{s})f_k(A_{ki} = 1, \bm{A}_{k(-i)} = \bm{s}')}\right] \nonumber \\
 & = \sum_{i=1}^{n_k} \tilde{Y}^2_{ki}(1,p_k) \left[ \sum_{\bm{s}} \pi^2_k(\bm{A}_{k(-i)} = \bm{s}\mid A_{ki} = 1) \frac{1 - f_k(A_{ki} = 1, \bm{A}_{k(-i)} = \bm{s})}{f_k(A_{ki} = 1, \bm{A}_{k(-i)} = \bm{s})} \right.\nonumber\\
 & \hspace{1.5in} \left. - \mathop{\sum\sum}_{\bm{s}\neq \bm{s}'}\pi_k(\bm{A}_{k(-i)} = \bm{s}\mid A_{ki} = 1)\pi_k(\bm{A}_{k(-i)} = \bm{s}'\mid A_{ki} = 1)\right] \nonumber \\
 &= \sum_{i=1}^{n_k} c_{i,1}\tilde{Y}^2_{ki}(1,p_k), \\
 G_{2k} &= \mathop{\sum\sum}_{i \neq i'}\Cov\left(\tilde{Y}_{ki}(1,p_k)\sum_{\bm{s}} \frac{\mathbbm{1}(A_{ki} = 1,\bm{A}_{k(-i)} = \bm{s})\pi_k(\bm{A}_{k(-i)} = \bm{s}\mid A_{ki} = 1)}{f_k(A_{ki} = 1, \bm{A}_{k(-i)} = \bm{s})}, \right. \nonumber\\
& \hspace{1.1in}\left. \tilde{Y}_{ki'}(1,p_k)\sum_{\bm{s}} \frac{\mathbbm{1}(A_{ki'} = 1,\bm{A}_{k(-i')} = \bm{s})\pi_k(\bm{A}_{k(-i')} = \bm{s}\mid A_{ki'} = 1)}{f_k(A_{ki'} = 1, \bm{A}_{k(-i')} = \bm{s})} \right) \nonumber \\
&= \mathop{\sum\sum}_{i \neq i'} \tilde{Y}_{ki}(1,p_k) \tilde{Y}_{ki'}(1,p_k) \nonumber \\
& \hspace{1in} \times \left[ \sum_{\bm{s}}\frac{\pi_k(\bm{A}_{k(-i)} = \bm{s}\mid A_{ki} = 1)\pi_k(\bm{A}_{k(-i')} = \bm{s}\mid A_{ki'} = 1)}{f_k(A_{ki} = 1, \bm{A}_{k(-i)} = \bm{s})f_k(A_{ki'} = 1, \bm{A}_{k(-i')} = \bm{s})} \right.\nonumber \\
& \hspace{1.5in} \times \Cov\{\mathbbm{1}(A_{ki} = 1,\bm{A}_{k(-i)} = \bm{s}),\mathbbm{1}(A_{ki'} = 1,\bm{A}_{k(-i')} = \bm{s})\} \nonumber\\
& \hspace{1.25in} + \mathop{\sum\sum}_{\bm{s} \neq \bm{s}'}\frac{\pi_k(\bm{A}_{k(-i)} = \bm{s}\mid A_{ki} = 1)\pi_k(\bm{A}_{k(-i')} = \bm{s}'\mid A_{ki'} = 1)}{f_k(A_{ki} = 1, \bm{A}_{k(-i)} = \bm{s})f_k(A_{ki'} = 1, \bm{A}_{k(-i')} = \bm{s}')} \nonumber\\
& \hspace{1.5in} \times \Cov\{\mathbbm{1}(A_{ki} = 1,\bm{A}_{k(-i)} = \bm{s}),\mathbbm{1}(A_{ki'} = 1,\bm{A}_{k(-i')} = \bm{s}')\}\Big] \nonumber \\
& = \mathop{\sum\sum}_{i \neq i'} \tilde{Y}_{ki}(1,p_k) \tilde{Y}_{ki'}(1,p_k) \times \frac{1}{\pi_k(A_{ki} = 1)\pi_k(A_{ki'} = 1)} \nonumber \\
& \hspace{1in} \times \left[ \sum_{\bm{s}}\frac{\pi_k(A_{ki} = 1,\bm{A}_{k(-i)} = \bm{s})\pi_k(A_{ki'} = 1, \bm{A}_{k(-i')} = \bm{s})}{f_k(A_{ki} = 1, \bm{A}_{k(-i)} = \bm{s})f_k(A_{ki'} = 1, \bm{A}_{k(-i')} = \bm{s})} \right.\nonumber \\
& \hspace{1.5in} \times \{f_k(A_{ki} = 1,\bm{A}_{k(-i)} = \bm{s},A_{ki'} = 1,\bm{A}_{k(-i')} = \bm{s}) \nonumber \\
& \hspace{1.5in} - f_k(A_{ki} = 1, \bm{A}_{k(-i)} = \bm{s})f_k(A_{ki'} = 1, \bm{A}_{k(-i')} = \bm{s})\} \nonumber \\
& \hspace{1.5in} + \mathop{\sum\sum}_{\bm{s} \neq \bm{s}'}\frac{\pi_k(A_{ki} = 1,\bm{A}_{k(-i)} = \bm{s})\pi_k(A_{ki'} = 1,\bm{A}_{k(-i')} = \bm{s}')}{f_k(A_{ki} = 1, \bm{A}_{k(-i)} = \bm{s})f_k(A_{ki'} = 1, \bm{A}_{k(-i')} = \bm{s}')} \nonumber \\
& \hspace{2in} \times \{- f_k(A_{ki} = 1, \bm{A}_{k(-i)} = \bm{s})f_k(A_{ki'} = 1, \bm{A}_{k(-i')} = \bm{s}')\} \Big]\nonumber \\
& = \mathop{\sum\sum}_{i \neq i'} \tilde{Y}_{ki}(1,p_k) \tilde{Y}_{ki'}(1,p_k) \times \frac{1}{\pi_k(A_{ki} = 1)\pi_k(A_{ki'} = 1)} \nonumber \\
& \quad \times \Big[\sum_{\bm{s}}\frac{\pi_k(A_{ki} = 1,\bm{A}_{k(-i)} = \bm{s})\pi_k(A_{ki'} = 1,\bm{A}_{k(-i)} = \bm{s})}{f_k(A_{ki} = 1,\bm{A}_{k(-i)} = \bm{s})f_k(A_{ki'} = 1,\bm{A}_{k(-i')} = \bm{s})}f_k(A_{ki} = 1,\bm{A}_{k(-i)} = \bm{s},A_{ki'} = 1,\bm{A}_{k(-i')} = \bm{s})
\nonumber \\
& \quad - \{\sum_{\bm{s}}\pi_k(A_{ki'} = 1,\bm{A}_{k(-i)} = \bm{s})\}\{\sum_{\bm{s}'}\pi_k(A_{ki'} = 1,\bm{A}_{k(-i)} = \bm{s}')\}\Big] \nonumber \\
& = \mathop{\sum\sum}_{i \neq i'} d_{ii',1}\tilde{Y}_{ki}(1,p_k) \tilde{Y}_{ki'}(1,p_k).
\end{align*}
\qed

\subsubsection{Proof of Theorem \ref{thm_varDE}}
\begin{align*}
    \Var(\widehat{\text{DE}}^\pi_{\text{HT}}) &= \Var(\hat{\mu}^\pi_{1,\text{HT}}) + \Var(\hat{\mu}^\pi_{0,\text{HT}}) - 2 \Cov(\hat{\mu}^\pi_{1,\text{HT}},\hat{\mu}^\pi_{0,\text{HT}}).
\end{align*}
Now, following similar steps as in the proof of Theorem \ref{thm_var1},
\begin{align*}
&\Cov(\hat{\mu}^\pi_{1,\text{HT}},\hat{\mu}^\pi_{0,\text{HT}}) \nonumber\\
    &= \Cov\left(\frac{1}{K}\sum_{k=1}^{K}\frac{1}{|\mathcal{S}_k|}\sum_{i=1}^{n_k} \frac{\mathbbm{1}(A_{ki} = 1)\pi_k(\bm{A}_{k(-i)}\mid A_{i} = 1)}{f_k(\bm{A}_k)}\tilde{Y}_{ki}(1,p_k),\right.\nonumber\\
    & \quad \quad \left.\frac{1}{K}\sum_{k=1}^{K}\frac{1}{|\mathcal{S}_k|}\sum_{i=1}^{n_k} \frac{\mathbbm{1}(A_{ki} = 0)\pi_k(\bm{A}_{k(-i)}\mid A_{i} = 0)}{f_k(\bm{A}_k)}\tilde{Y}_{ki}(0,p_k) \right) \nonumber\\
    & = \frac{1}{K^2}\sum_{k=1}^{K}\frac{1}{|\mathcal{S}_k|^2}(G_{k1} + G_{k2}),
\end{align*}
where
\begin{align*}
  G_{1k} 
 & = \sum_{i=1}^{n_k}\tilde{Y}_{ki}(1,p_k)\tilde{Y}_{ki}(0,p_k)\sum_{\bm{s}}\sum_{\bm{s}'}\frac{\pi_k(\bm{A}_{k(-i)} = \bm{s}\mid A_{ki} = 1)\pi_k(\bm{A}_{k(-i)} = \bm{s}'\mid A_{ki} = 0)}{f_k(A_{ki} = 1,\bm{A}_{k(-i)} = \bm{s})f_k(A_{ki} = 0,\bm{A}_{k(-i)} = \bm{s}')} \nonumber\\
 & \hspace{2in} \times \Cov\{\mathbbm{1}(A_{ki} = 1,\bm{A}_{k(-i)} = \bm{s}), \mathbbm{1}(A_{ki} = 0,\bm{A}_{k(-i)} = \bm{s}')\} \nonumber\\
 &= -\sum_{i=1}^{n_k}\tilde{Y}_{ki}(1,p_k)\tilde{Y}(0,p_k), \\
 G_{2k}
& = \mathop{\sum\sum}_{i\neq i'}\tilde{Y}_{ki}(1,p_k)\tilde{Y}_{ki'}(0,p_k)\sum_{\bm{s}}\sum_{\bm{s}'}\frac{\pi_k(\bm{A}_{k(-i)} = \bm{s}\mid A_{ki} = 1)\pi_k(\bm{A}_{k(-i')} = \bm{s}'\mid A_{ki'} = 0)}{f_k(A_{ki} = 1,\bm{A}_{k(-i)} = \bm{s})f_k(A_{ki'} = 0,\bm{A}_{k(-i')} = \bm{s}')}\nonumber\\
& \hspace{1in} \times\Cov\{\mathbbm{1}(A_{ki} = 1,\bm{A}_{k(-i)} = \bm{s}), \mathbbm{1}(A_{ki'} = 0,\bm{A}_{k(-i')} = \bm{s}')\} \nonumber \\
&= \mathop{\sum\sum}_{i\neq i'}\frac{\tilde{Y}_{ki}(1,p_k)\tilde{Y}_{ki'}(0,p_k)}{\pi_k(A_{ki}=1)\pi_k(A_{ki'}=0)}\sum_{\bm{s}}\sum_{\bm{s}'}\pi_k(A_{ki} = 1, \bm{A}_{k(-i)} = \bm{s})\pi_k(A_{ki'} = 0,\bm{A}_{k(-i')} = \bm{s}') \nonumber\\
& \hspace{1in} \times \left\{\frac{f_k(A_{ki} = 1,\bm{A}_{k(-i)} = \bm{s}, A_{ki'} = 0,\bm{A}_{k(-i')} = \bm{s}')}{f_k(A_{ki} = 1,\bm{A}_{k(-i)} = \bm{s})f_k(A_{ki'} = 0,\bm{A}_{k(-i')} = \bm{s}')} - 1\right\}. \nonumber \\
& = \mathop{\sum\sum}_{i\neq i'}g_{ii'}{\tilde{Y}_{ki}(1,p_k)\tilde{Y}_{ki'}(0,p_k)}.
\end{align*}
\qed


\subsubsection{Proof of Proposition \ref{prop_additive_mu1}}

Without loss of generality, we set $a = 1$. Now, from Theorem \ref{thm_generalvar}, we have,
\begin{align}
\Lambda_{1,k,j} \ = \ & \sum_{\bm{s}}\frac{\pi^2_{k}(\bm{A}_{k(-i^*)} = \bm{s}\mid A_{ki^*} = 1)}{f_k(A_{ki^*} = 1, \bm{A}_{k(-i^*)} = \bm{s})}Y^2_{kj}(A_{ki^*} = 1, \bm{A}_{k(-i^*)} = \bm{s}) \nonumber\\
& \quad - \left\{\sum_{\bm{s}}\pi_{k}(\bm{A}_{k(-i^*)} = \bm{s}\mid A_{ki^*} = 1)Y_{kj}(A_{ki^*} = 1, \bm{A}_{k(-i^*)} = \bm{s}) \right\}^2
\end{align}
Now, write $\bm{a}_{kj} = (a_{kj1},....,a_{kjn_k})^\top$ as the vector of assignments in cluster $k$ corresponding to $A_{ki^*} = 1$ and $\bm{A}_{k(-i^*)} = \bm{s}$. 
\begin{align}
    &\sum_{\bm{s}}\pi_{k}(\bm{A}_{k(-i^*)} = \bm{s}\mid A_{ki^*} = 1)Y_{kj}(A_{ki^*} = 1, \bm{A}_{k(-i^*)} = \bm{s}) \nonumber\\
    & = \sum_{\bm{s}}\pi_{k}(\bm{A}_{k(-i^*)} = \bm{s}\mid A_{ki^*} = 1)\left(\beta^{(0)}_{kj} + \sum_{i=1}^{n_k}\beta^{(i)}_{kj}a_{kji}\right). \nonumber\\
    & = \beta^{(0)}_{kj} + \sum_{i=1}^{n_k}\beta^{(i)}_{kj}\sum_{\bm{s}:a_{kji = 1}}\pi_{k}(\bm{A}_{k(-i^*)} = \bm{s}\mid A_{ki^*} = 1) \nonumber\\
    & = \beta^{(0)}_{kj} + \sum_{i=1}^{n_k}\beta^{(i)}_{kj}\pi_k(A_{ki} = 1|A_{ki^*} = 1) = (1,\bm{\pi}^\top_k(\cdot|A_{ki^*} = 1))\tilde{\bm{\beta}}_{kj}. \label{eq_A4.6_1}
\end{align}
Therefore,
\begin{align}
\Lambda_{1,k,j} \ = \ & \sum_{\bm{s}}\frac{\pi^2_{k}(\bm{A}_{k(-i^*)} = \bm{s}\mid A_{ki^*} = 1)}{f_k(A_{ki^*} = 1, \bm{A}_{k(-i^*)} = \bm{s})}\left\{(1,\bm{a}^\top_j)\tilde{\bm{\beta}}_{kj} \right\}^2 - \left\{(1,\bm{\pi}^\top_k(\cdot|A_{ki^*} = 1))\tilde{\bm{\beta}}_{kj} \right\}^2.
\end{align}
Moreover, 
\begin{align}
    \Lambda_{2,k,j,j'} \ = \  = & \sum_{\tilde{\bm{s}}}\frac{\pi^2_k(A_{ki^{*}} = 1, A_{ki^{*'}} = 1, \bm{A}_{k(-i^*, -i^{*'})} = \tilde{\bm{s}})}{f_k(A_{ki^{*}} = 1, A_{ki^{*'}} = 1, \bm{A}_{k(-i^*, -i^{*'})} = \tilde{\bm{s}})\pi_k(A_{ki^*} = 1)\pi_k(A_{ki^{*'}} = 1)} \nonumber\\
& \quad  \times Y_{kj}(A_{ki^{*}} = 1, A_{ki^{*'}} = 1, \bm{A}_{k(-i^*,-i^{*'})} = \tilde{\bm{s}})Y_{kj'}(A_{ki^{*}} = 1,A_{ki^{*'}} = 1, \bm{A}_{k(-i^*,-i^{*'})} = \tilde{\bm{s}}) \nonumber\\
& - \left\{\sum_{\bm{s}}\pi_{k}(\bm{A}_{k(-i^*)} = \bm{s}\mid A_{ki^*} = 1)Y_{kj}(A_{ki^*} = 1, \bm{A}_{k(-i^*)} = \bm{s}) \right\} \nonumber \\ 
& \times \left\{\sum_{\bm{s}}\pi_{k}(\bm{A}_{k(-i^{*'})} = \bm{s}\mid A_{ki^{*'}} = 1)Y_{kj'}(A_{ki^{*'}} = 1, \bm{A}_{k(-i^{*'})} = \bm{s}) \right\} \nonumber \\
\ = \ &\sum_{\tilde{\bm{s}}}\frac{\mathbbm{\pi}^2_k(A_{ki^{*}} = 1, A_{ki^{*'}} = 1, \bm{A}_{k(-i^{*},-i^{*'})} = \tilde{\bm{s}}) \left\{(1,\bm{a}^\top_{kjj'})\tilde{\bm{\beta}}_{kj} \right\}\left\{(1,\bm{a}^\top_{kjj'})\tilde{\bm{\beta}}_{kj'} \right\}}{f_k(A_{ki^*} = 1, A_{ki^{*'}} = 1, \bm{A}_{k(-i^{*},-i^{*'})} = \tilde{\bm{s}})\pi_k(A_{ki^*} = 1)\pi_k(A_{ki^{*'}} = 1)} \nonumber\\
    &\hspace{0.5cm} - \left\{(1,\bm{\pi}^\top_k(\cdot|A_{ki^*} = 1))\tilde{\bm{\beta}}_{kj} \right\}\left\{(1,\bm{\pi}^\top_k(\cdot|A_{ki^{*'}} = 1))\tilde{\bm{\beta}}_{kj'} \right\},
\end{align}
where the last equality holds due to Equation \ref{eq_A4.6_1}.

\qed
\subsubsection{Proof of Theorem \ref{thm_conservative_mu1}}

\begin{align}
\hat{\mu}^\pi_{a,\text{HT}} & = \frac{1}{K}\sum_{k=1}^{K}\left\{\frac{1}{|\mathcal{S}_k|}\sum_{j \in \mathcal{S}_k}\mathbbm{1}(A_{ki^*} = a)\frac{\pi_{k}(\bm{A}_{k(-i^*)} \mid A_{ki^*} = a)}{f_k(\bm{A}_k)}Y^{\text{obs}}_{kj}\right\} \nonumber \\
& = \sum_{k=1}^{K}\left\{\sum_{j \in \mathcal{S}_k}\frac{\mathbbm{1}(A_{ki^*} = a)}{K|\mathcal{S}_k|}\frac{\pi_{k}(\bm{A}_{k(-i^*)} \mid A_{ki^*} = a)}{f_k(\bm{A}_k)}(1,\bm{A}^\top_{k})\tilde{\bm{\beta}}_{kj}\right\} 
\end{align}
Thus, $\hat{\mu}^\pi_{a,\text{HT}}$ is of the form $\hat{\mu}^\pi_{a,\text{HT}} = \sum_{k}\sum_{j}\tilde{\bm{\beta}}^\top_{kj}\bm{\psi}_{kj}$, for some $|\mathcal{S}_k|\times 1$ vector $\bm{\psi}_{kj}$. Denoting $\tilde{\bm{\beta}}_k = (\tilde{\bm{\beta}}^\top_{k1},...,\tilde{\bm{\beta}}^\top_{k|\mathcal{S}_k|})^\top$ and $\bm{\Psi}_k = (\bm{\psi}^\top_{k1},...,\bm{\psi}^\top_{k|\mathcal{S}_k|})^\top$, we have
\begin{align}
    \Var(\hat{\mu}^\pi_{a,\text{HT}}) & = \sum_{k=1}^{K}\tilde{\bm{\beta}}_k \Var(\bm{\Psi}_k)\tilde{\bm{\beta}}_k.
    \label{eq_A10_4.8_1}
\end{align}
We write $\Var(\hat{\mu}^\pi_{a,\text{HT}})$ as a function of the $\tilde{\bm{\beta}}_{k}$s, i.e., $\Var(\bm{\Psi}_k) = h(\tilde{\bm{\beta}}_{1},...,\tilde{\bm{\beta}}_{K})$. Equation \ref{eq_A10_4.8_1} shows that $h(\cdot)$ is convex. Now, let $\hat{\tilde{\bm{\beta}}}_k = (\hat{\tilde{\bm{\beta}}}^\top_{k1},...,\hat{\tilde{\bm{\beta}}}^\top_{k|\mathcal{S}_k|})^\top$ be the estimated $\tilde{\bm{\beta}}_k$. By construction, $\hat{\tilde{\bm{\beta}}}_k$ is unbiased for ${\tilde{\bm{\beta}}}_k$. 
\begin{align}
    \E\{\widehat{\Var}(\hat{\mu}^\pi_{a,\text{HT}})\} = \E\{h(\tilde{\hat{\bm{\beta}}}_{1},...,\hat{\tilde{\bm{\beta}}}_{K}) \} \geq h(\tilde{\bm{\beta}}_{1},...,\tilde{\bm{\beta}}_{K}) = \Var(\hat{\mu}^\pi_{a,\text{HT}}),
\end{align}
where the inequality holds due to Jensen's inequality.

\qed
\subsubsection{Proof of Theorem \ref{thm_general_additive}}

We can write
\begin{align}
\hat{\tau}^\pi_{\text{HT}} \ = \ & \frac{1}{K}\sum_{k=1}^{K}\left\{\frac{1}{|\mathcal{S}_k|}\sum_{j \in \mathcal{S}_k}\frac{\pi_{kj}(\bm{A}_k)}{f_k(\bm{A}_k)}(1,\bm{A}^\top_k)\tilde{\bm{\beta}}_{kj}\right\} \label{eq_A10_4.8_1_2}\\
\ = \  & \frac{1}{K}\sum_{k=1}^{K}\left\{\frac{1}{|\mathcal{S}_k|}\sum_{j \in \mathcal{S}_k}\sum_{\bm{a} \in \text{Supp}(f_k)}\mathbbm{1}(\bm{A}_k = \bm{a})\frac{\pi_{kj}(\bm{a})}{f_k(\bm{a})}(1,\bm{a}^\top)\tilde{\bm{\beta}}_{kj}\right\} \nonumber\\
\ = \ & \frac{1}{K}\sum_{k=1}^{K}\frac{1}{|\mathcal{S}_k|}\sum_{\bm{a} \in \text{Supp}(f_k)}\mathbbm{1}(\bm{A}_k = \bm{a})\left\{\sum_{j \in \mathcal{S}_k}\frac{\pi_{kj}(\bm{a})}{f_k(\bm{a})}(1,\bm{a}^\top)\tilde{\bm{\beta}}_{kj}\right\} \nonumber\\
\ = \ & \frac{1}{K}\sum_{k=1}^{K}\frac{1}{|\mathcal{S}_k|}\sum_{\bm{a} \in \text{Supp}(f_k)}\mathbbm{1}(\bm{A}_k = \bm{a})\zeta_k(\bm{a}),
\end{align}
where $\zeta_k(\bm{a}) = \sum_{j \in \mathcal{S}_k}\frac{\pi_{kj}(\bm{a})}{f_k(\bm{a})}(1,\bm{a}^\top)\tilde{\bm{\beta}}_{kj}$.
Therefore, 
\begin{align}
\Var(\hat{\tau}^\pi_{\text{HT}}) =   \frac{1}{K^2}\sum_{k = 1}^{K} \frac{1}{|\mathcal{S}_k|^2}\left[\sum_{\bm{a}}f_k(\bm{a}) \{1 - f_k(\bm{a})\}\zeta^2_k(\bm{a}) - \mathop{\sum\sum}_{\bm{a} \neq \bm{a}'}f_k(\bm{a})f_k(\bm{a}')\zeta_k(\bm{a})\zeta_k(\bm{a}') \right].
\end{align}
Next, Equation \ref{eq_A10_4.8_1_2} implies that $\hat{\tau}^\pi_{\text{HT}}$ has the form $\hat{\tau}^\pi_{\text{HT}} = \sum_{k}\sum_{j}\tilde{\bm{\beta}}^\top_{kj}\bm{\psi}_{kj}$, for some random vectors $\bm{\psi}_{kj}$. Thus, following the proof of Theorem \ref{thm_conservative_mu1}, we conclude that the variance estimator based on $\hat{\tilde{\bm{\beta}}}_{kj}$ is conservative.

\qed

\subsection{Additional theoretical results}
\label{sec_add_theory}

\subsubsection{Partial identification of the variance of the Horvitz-Thompson estimator}
\label{sec_partial_ht}

In this section, we focus on the partial identification of the variance of $\hat{\mu}^\pi_{a,\text{HT}}$. To this end, one possible approach is to assume that for a given treatment condition of the key-intervention unit $i^*$, if the two assignment vectors of the remaining intervention units are sufficiently similar, then the corresponding potential outcomes of unit $j$ should also be similar.
Specifically, within each cluster, we can partially identify $\Var(\hat{\mu}^\pi_{a,\text{HT}})$ by assuming the following form of Lipschitz continuity on the potential outcomes.  
\begin{assumption}[Lipschitz potential outcomes] \normalfont
For all $j \in \mathcal{S}_k$ and $\bm{s},\bm{s}' \in \{0,1\}^{n_k  -1}$,
\begin{equation*}
|Y_{kj}(A_{ki^*} = a, \bm{A}_{k(-i^*)} = \bm{s}) - Y_{kj}(A_{ki^*} = a, \bm{A}_{k(-i^*)} = \bm{s}')| \leq C(n_k) \times d(\bm{s},\bm{s}'),
\end{equation*}
where $C(n_k)$ is a function of $n_k$ that decreases to zero as $n_k \to \infty$, and $d(\cdot, \cdot)$ is a distance measure on $\mathbb{R}^{n_k -1}$.   
\label{assump_lipschitz}
\end{assumption}
For example, if $C(n_k) = c/\sqrt{n_k}$ for some constant $c>0$ and $d(\cdot, \cdot)$ is the $L_1$ distance, then Assumption \ref{assump_lipschitz} implies that $|Y_{kj}(A_{ki^*} = a, \bm{A}_{k(-i^*)} = \bm{s}) - Y_{kj}(A_{ki^*} = a, \bm{A}_{k(-i^*)} = \bm{s}')|$ is bounded by $c/\sqrt{n_k}$ times the number of intervention units for which the treatment assignment vectors $\bm{s}$ and $\bm{s}'$ differ.
Note that Assumption~\ref{assump_lipschitz} is implied by Assumption~\ref{assump_stratified}, and hence the former is a weaker assumption. Assumption~\ref{assump_lipschitz} is also related to the approximate neighborhood interference assumption of \cite{leung2022causal}, which assumes that in expectation, the difference in potential outcome for unit $j$ under any perturbation of the assignment of units sufficiently far apart (in terms of the path distance in the network) is negligible. Instead of focusing on the units being perturbed, Assumption~\ref{assump_lipschitz} focuses on the amount of perturbation and posits that small perturbations in assignments imply small differences in the potential outcomes.

Proposition \ref{prop_partial} shows that if the potential outcomes are bounded, then under Assumptions~\ref{assump_partial}--\ref{assump_lipschitz}, we can partially identify $\Var(\hat{\mu}^\pi_{a,\text{HT}})$ in completely randomized experiments. 
\begin{proposition}[Partial identification of the variance] \normalfont
Consider a completely randomized experiment in each cluster $k \in \{1,2,...,K\}$, where $n_{ka}$ and $p_k$ are the number and proportion of intervention units assigned to treatment $a \in \{0,1\}$, respectively. Under Assumptions \ref{assump_partial}--\ref{assump_lipschitz}, $\pi_k(\cdot) = f_k(\cdot)$, and bounded potential outcomes,
\begin{align}
\Var(\hat{\mu}^\pi_{a,\text{HT}}) \ \leq \ & \frac{1}{K^2}\sum_{k = 1}^{K} \frac{1}{|\mathcal{S}_k|^2} \left[ \left\{\frac{C(n_k)^2 |\mathcal{S}_k|}{2\binom{n_k-1}{n_{ka}-1}}^2 + \frac{C(n_k)^2 |\mathcal{S}_k| (|\mathcal{S}_k|- 1)}{2{\binom{n_k-2}{n_{ka}-2}}^2}\right\} \mathop{\sum\sum}_{\bm{s}\neq \bm{s}'}d^2(\bm{s},\bm{s}') \nonumber \right.\\
&+ \frac{1-p_k}{p_k {\binom{n_k - 1}{n_{ka} - 1}}} \sum_{j \in \mathcal{S}_k} \sum_{\bm{s}}Y^2_{kj}(A_{ki^*} = a, A_{k(-i^*)} = \bm{s})  \nonumber \\
  &+ \mathop{\sum\sum}_{j \neq j' \in \mathcal{S}_k} \left(\frac{1}{{\binom{n_k - 1}{n_{ka} - 1}}p_k} - \frac{1}{{\binom{n_k - 2}{n_{ka} - 2}}}\right)\sum_{\bm{s}} Y_{kj}(A_{ki^*} = a, A_{ki^{*'}} = a, A_{k(-i^*,i^{*'})} = \bm{s}) \nonumber \\ 
  & \hspace{7cm} Y_{kj'}(A_{ki^*} = a, A_{ki^{*'}} = a, A_{k(-i^*,i^{*'})} = \bm{s}) \Bigg]  \label{eq_lipupper}
\end{align}
\label{prop_partial}
\end{proposition}

\begin{proof}

Without loss of generality, we set $a = 1$. The proof for the case of $a =0$ is analogous.  Let $\bar{Y}_{kj}(1,\pi_k) = \sum_{\bm{s}}\pi_{k}(\bm{A}_{k(-i^*)} = \bm{s}\mid A_{ki^*} = 1)Y_{kj}(A_{ki^*} = 1, \bm{A}_{k(-i^*)} = \bm{s})$. For a completely randomized experiment, $\bar{Y}_{kj}(1,\pi_k) = \sum_{\bm{s}}Y_{kj}(A_{ki^*} = 1, \bm{A}_{k(-i^*)} = \bm{s})/ \binom{n_k-1}{n_{k1}-1}$. Now, using the notation of Theorem~\ref{thm_generalvar}, we get
\begin{align*}
    & \Lambda_{1,k,j} \\
    = & \frac{1}{\binom{n_k}{n_{k1}}p^2}\sum_{\bm{s}}Y^2_{kj}(A_{ki^*} = 1, \bm{A}_{k(-i^*)} = \bm{s}) - \bar{Y}^2_{kj}(1,\pi_k) \nonumber \\
   =  & \frac{1}{\binom{n_k -1}{n_{k1} -1}}\sum_{\bm{s}}\{Y_{kj}(A_{ki^*} = 1, \bm{A}_{k(-i^*)} = \bm{s}) - \bar{Y}_{kj}(1,\pi_k)\}^2 + \frac{1-p_k}{p_k\binom{n_k -1}{n_{k1} -1}}\sum_{\bm{s}}Y^2_{kj}(A_{ki^*} = 1, \bm{A}_{k(-i^*)} = \bm{s}) \nonumber \\
= & \frac{1}{2\binom{n_k -1}{n_{k1} -1}^2}\mathop{\sum\sum}_{\bm{s}\neq \bm{s}'}\{Y_{kj}(A_{ki^*} = 1, \bm{A}_{k(-i^*)} = \bm{s}) - Y_{kj}(A_{ki^*} = 1, \bm{A}_{k(-i^*)} = \bm{s}')\}^2 \\
& \quad + \frac{1-p_k}{p_k\binom{n_k -1}{n_{k1} -1}}\sum_{\bm{s}}Y^2_{kj}(A_{ki^*} = 1, \bm{A}_{k(-i^*)} = \bm{s}) \nonumber \\
\leq & \frac{C(n_k)^2}{2\binom{n_k -1}{n_{k1} -1}^2}\mathop{\sum\sum}_{\bm{s}\neq \bm{s}'}d^2(\bm{s},\bm{s}') + \frac{1-p_k}{p_k\binom{n_k -1}{n_{k1} -1}}\sum_{\bm{s}}Y^2_{kj}(A_{ki^*} = 1, \bm{A}_{k(-i^*)} = \bm{s}).
\end{align*}
The last equality holds since, for $n$ data points $x_1,....,x_n$ with mean $\bar{x}$, $\frac{1}{n}\sum_{i=1}^{n}(x_i - \bar{x})^2 = \frac{1}{2n^2}\mathop{\sum\sum}_{i \neq j}(x_i - x_j)^2$. The final inequality holds due to the Lipschitz condition. Therefore, we have
\begin{align}
    \sum_{j \in \mathcal{S}_k}\Lambda_{1,k,j} \leq \frac{C(n_k)^2|\mathcal{S}_k|}{2\binom{n_k -1}{n_{k1} -1}^2}\mathop{\sum\sum}_{\bm{s}\neq \bm{s}'}d^2(\bm{s},\bm{s}') + \frac{1-p_k}{p_k\binom{n_k -1}{n_{k1} -1}}\sum_{j \in \mathcal{S}_k}\sum_{\bm{s}}Y^2_{kj}(A_{ki^*} = 1, \bm{A}_{k(-i^*)} = \bm{s}). \label{eq_A3_0}
\end{align}
Next, for two units $j$ and $j'$, with key-intervention units $i^*$ and $i^{*'}$, denote 
$$\bar{Y}_{kj}(1,1,\pi_k) = \frac{1}{\binom{n_k-2}{n_{k1} -2}}\sum_{\bm{s}}Y_{kj}(A_{ki^*} = 1,A_{ki^{*'}} = 1, \bm{A}_{k(-i^*,-i^{*'})} = \bm{s}).$$ 
Now,
\begin{align}
    &\Lambda_{2,k,j,j'} \nonumber\\
    = & \frac{1}{\binom{n_k}{n_{k1}}p^2_k}\sum_{\bm{s}}Y_{kj}(A_{ki^*} = 1,A_{ki^{*'}} = 1, \bm{A}_{k(-i^*,-i^{*'})} = \bm{s})Y_{kj'}(A_{ki^*} = 1,A_{ki^{*'}} = 1, \bm{A}_{k(-i^*,-i^{*'})} = \bm{s}) \nonumber\\
    & \hspace{.7in} - \bar{Y}_{kj}(1,\pi_k)\bar{Y}_{kj'}(1,\pi_k). \nonumber \\
    = & \frac{1}{\binom{n_k-2}{n_{k1}-2}}\sum_{\bm{s}}\{Y_{kj}(A_{ki^*} = 1,A_{ki^{*'}} = 1, \bm{A}_{k(-i^*,-i^{*'})} = \bm{s}) - \bar{Y}_{kj}(1,1,\pi_k) \} \nonumber \\
    & \hspace{1in} \times \{Y_{kj'}(A_{ki^*} = 1,A_{ki^{*'}} = 1, \bm{A}_{k(-i^*,-i^{*'})} = \bm{s})- \bar{Y}_{kj'}(1,1,\pi_k) \} \nonumber \\
    & \quad + \left(\frac{1}{\binom{n_k-1}{n_{k1}-1}p_k} - \frac{1}{\binom{n_k-2}{n_{k1}-2}}\right) \sum_{\bm{s}}Y_{kj}(A_{ki^*} = 1,A_{ki^{*'}} = 1, \bm{A}_{k(-i^*,-i^{*'})} = \bm{s}) \nonumber\\
    & \hspace{1in} \times Y_{kj'}(A_{ki^*} = 1,A_{ki^{*'}} = 1, \bm{A}_{k(-i^*,-i^{*'})} = \bm{s}) \nonumber\\
    & \quad + \{\bar{Y}_{kj}(1,1,\pi_k)\bar{Y}_{kj'}(1,1,\pi_k)- \bar{Y}_{kj}(1,\pi_k)\bar{Y}_{kj'}(1,\pi_k)\}
    \label{eq_A3_1}
\end{align}
Using Cauchy-Schwarz inequality, the first term in Equation~\eqref{eq_A3_1} can be upper-bounded as follows,
\begin{align*}
 &   \sum_{\bm{s}}\{Y_{kj}(A_{ki^*} = 1,A_{ki^{*'}} = 1, \bm{A}_{k(-i^*,-i^{*'})} = \bm{s}) - \bar{Y}_{kj}(1,1,\pi_k) \} \nonumber\\
 & \hspace{1in} \times \{Y_{kj'}(A_{ki^*} = 1,A_{ki^{*'}} = 1, \bm{A}_{k(-i^*,-i^{*'})} = \bm{s}) - \bar{Y}_{kj'}(1,1,\pi_k) \}  \nonumber\\
 & \leq  \sqrt{\sum_{\bm{s}}\{Y_{kj}(A_{ki^*} = 1,A_{ki^{*'}} = 1, \bm{A}_{k(-i^*,-i^{*'})} = \bm{s}) - \bar{Y}_{kj}(1,1,\pi_k)\}^2} \nonumber\\
 & \hspace{1in} \times \sqrt{\sum_{\bm{s}}\{Y_{kj'}(A_{ki^*} = 1,A_{ki^{*'}} = 1 \mid \bm{A}_{k(-i^*,-i^{*'})} = \bm{s}) - \bar{Y}_{kj'}(1,1,\pi_k)\}^2} \nonumber\\
 & \leq \frac{C(n_k)^2}{2\binom{n_k-2}{n_{k1}-2}}\mathop{\sum\sum}_{\bm{s} \neq \bm{s}'}d(\bm{s},\bm{s}'),
\end{align*}
where the last inequality follows from similar steps as in the derivation for $\Lambda_{1,k,j}$. Now, suppose that the potential outcomes are bounded by a constant $M$. The third term can be written as,
\begin{align*}
   & |\bar{Y}_{kj}(1,1,\pi_k)\bar{Y}_{kj'}(1,1,\pi_k)- \bar{Y}_{kj}(1,\pi_k)\bar{Y}_{kj'}(1,\pi_k)| \nonumber \\
   = \ & |\bar{Y}_{kj}(1,1,\pi_k)\bar{Y}_{kj'}(1,1,\pi_k)- \bar{Y}_{kj}(1,\pi_k)\bar{Y}_{kj'}(1,1,\pi_k) + \bar{Y}_{kj}(1,\pi_k)\bar{Y}_{kj'}(1,1,\pi_k)- \bar{Y}_{kj}(1,\pi_k)\bar{Y}_{kj'}(1,\pi_k)| \nonumber \\
   \leq \ &  M (|\bar{Y}_{kj}(1,1,\pi_k) - \bar{Y}_{kj}(1,\pi_k)| + |\bar{Y}_{kj'}(1,1,\pi_k) - \bar{Y}_{kj'}(1,\pi_k)|).
\end{align*}

Now, $|\bar{Y}_{kj}(1,1,\pi_k) - \bar{Y}_{kj}(1,\pi_k)| = \frac{\binom{n_k - 2}{n_{k1} -1}}{\binom{n_k - 1}{n_{k1} -1}}|\bar{Y}_{kj}(1,1,\pi_k) - \bar{Y}_{kj}(1,0,\pi_k)| \leq \frac{\binom{n_k - 2}{n_{k1} -1}}{\binom{n_k - 1}{n_{k1} -1}}C(n_k)$, which is negligible for sufficiently large $n_k$. Therefore, for large $n_k$, we can write,
\begin{align}
\mathop{\sum\sum}_{j \neq j'}\Lambda_{2,k,j,j'} &\leq \frac{|\mathcal{S}_k|(|\mathcal{S}_k| - 1)C(n_k)^2}{2\binom{n_k-2}{n_{k1}-2}^2}\mathop{\sum\sum}_{\bm{s} \neq \bm{s}'}d(\bm{s},\bm{s}') \nonumber\\
    & \quad +\mathop{\sum\sum}_{j \neq j'} \left(\frac{1}{\binom{n_k-1}{n_{k1}-1}p_k} - \frac{1}{\binom{n_k-2}{n_{k1}-2}}\right) \sum_{\bm{s}}Y_{kj}(A_{ki^*} = 1,A_{ki^{*'}} = 1,  \bm{A}_{k(-i^*,-i^{*'})} = \bm{s}) \nonumber \\
    & \hspace{1in} \times Y_{kj'}(A_{ki^*} = 1,A_{ki^{*'}} = 1, \bm{A}_{k(-i^*,-i^{*'})} = \bm{s}).
    \label{eq_A3_2}
\end{align}
The proof of the proposition follows from Equations~\eqref{eq_A3_0}~and~\eqref{eq_A3_2}. 
\end{proof}

The upper bound in Equation \ref{eq_lipupper} is estimable. To see this, note that in the second term $\sum_{j \in \mathcal{S}_k} \sum_{\bm{s}}Y^2_{kj}(A_{ki^*} = a, A_{k(-i^*)} = \bm{s})$ can be estimated by $\sum_{j \in \mathcal{S}_k} \frac{\mathbbm{1}(A_{ki^*} = a)}{f_k(\bm{A}_k)}Y_{kj}^2$ without bias. Moreover, in the last term 
$\mathop{\sum\sum}_{j \neq j' \in \mathcal{S}_k} \sum_{\bm{s}} Y_j(A_{ki^*} = a, A_{ki^{*'}} = a, A_{k(-i^*,i^{*'})} = \bm{s}) Y_{j'}(A_{ki^*} = a, A_{ki^{*'}} = a, A_{k(-i^*,i^{*'})} = \bm{s})$ can be estimated without bias by $\mathop{\sum\sum}_{j \neq j' \in \mathcal{S}_k} \frac{\mathbbm{1}(A_{ki^*} = a, A_{ki^{*'}} = a)}{f_k(\bm{A}_k)}Y_{kj}Y_{kj'}$. Therefore, a conservative estimator of $\Var(\hat{\mu}^\pi_{a,\text{HT}})$ is given by
\begin{align}
\widehat{\Var}(\hat{\mu}^\pi_{a,\text{HT}}) 
= & \frac{1}{K^2}\sum_{k = 1}^{K}\frac{1}{|\mathcal{S}_k|^2} \left[\left\{\frac{C(n_k)^2 |\mathcal{S}_k|}{{\binom{n_k-1}{n_{ka}-1}}} + \frac{C(n_k)^2 |\mathcal{S}_k| (|\mathcal{S}_k|- 1)}{2{\binom{n_k-2}{n_{ka}-2}}^2}\right\} \mathop{\sum\sum}_{\bm{s}\neq \bm{s}'}d^2(\bm{s},\bm{s}') \nonumber \right.\\
&+ \frac{1-p_k}{p_k {\binom{n_k - 1}{n_{ka} - 1}}} \sum_{j \in \mathcal{S}_k} \frac{\mathbbm{1}(A_{ki^*} = a)}{f_k(\bm{A}_k)}Y_{kj}^2 \nonumber \\
& \left. + \mathop{\sum\sum}_{j \neq j' \in \mathcal{S}_k} \left(\frac{1}{{\binom{n_k - 1}{n_{ka} - 1}}p_k} - \frac{1}{{\binom{n_k - 2}{n_{ka} - 2}}}\right)\frac{\mathbbm{1}(A_{ki^*} = a, A_{ki^{*'}} = a)}{f_k(\bm{A}_k)}Y_{kj}Y_{kj'} \right]. 
  \label{eq_lipupper2}
\end{align}
\subsubsection{Variance estimation under stratified interference for a completely randomized design}

In this section, we consider the scenario where we use the same complete randomization for both the hypothetical intervention $\pi_k(\cdot)$ and the actual intervention $f_k(\cdot)$. In this special case, the exact variances of $\hat{\mu}^\pi_{a,\text{HT}}$ and $\widehat{\text{DE}}^\pi_{\text{HT}}$
and their estimators can be greatly simplified as shown in the next proposition. 
\begin{proposition}[Variance and its estimation under the same complete randomization] \normalfont
Let $f_k(\cdot) = \pi_k(\cdot)$ and both correspond to a completely randomized experiment with $n_{ka}$ intervention units assigned to treatment $a$. Then, under Assumptions~\ref{assump_partial}--\ref{assump_fixedprop}, and for $a \in \{0,1\}$,
\begin{enumerate}[label=(\alph*),leftmargin=*]
\item $\Var(\hat{\mu}_{a,\text{HT}}) = \frac{1}{K^2}\sum_{k=1}^{K} \left(\frac{n_k}{|\mathcal{S}_k|}\right)^2 \left(1- \frac{n_{ka}}{n_k}\right)\frac{\tilde{V}^2_{ka}}{n_{ka}}$, where $\tilde{V}^2_{ka} = \frac{1}{n_k-1}\sum_{i \in \mathcal{I}_k}\left\{\tilde{Y}_{ki}(a,p_k) - \bar{\tilde{Y}}(a,p_k)\right\}^2$, $\bar{\tilde{Y}}(a,p_k) = \frac{1}{n_k}\sum_{i \in \mathcal{I}_k}\tilde{Y}_{ki}(a,p_k)$. An unbiased estimator of $\Var(\hat{\mu}_{a,\text{HT}})$ is $$\widehat{\Var}(\hat{\mu}_{a,\text{HT}}) = \frac{1}{K^2}\sum_{k=1}^{K} \left(\frac{n_k}{|\mathcal{S}_k|}\right)^2 \left(1- \frac{n_{ka}}{n_k}\right)\frac{\hat{\tilde{V}}^2_{ka}}{n_{ka}},$$ where $\hat{\tilde{V}}^2_{ka} = \frac{1}{n_k-1}\sum_{i \in \mathcal{I}_k: A_{ki} = a}(\tilde{Y}^{\text{obs}}_{ki} - \bar{\tilde{Y}}_k)^2$, $\bar{\tilde{Y}}_k = \frac{1}{n_k}\sum_{i \in \mathcal{I}_k: A_{ki} = a}\tilde{Y}^{\text{obs}}_{ki}$. 
\item $\Var(\widehat{\text{DE}}^\pi_{\text{HT}}) = \frac{1}{K^2}\sum_{k=1}^{K} \left(\frac{n_k}{|\mathcal{S}_k|}\right)^2 \left(\frac{\tilde{V}^2_{k1}}{n_{k1}} + \frac{\tilde{V}^2_{k0}}{n_{k0}} - \frac{\tilde{V}^2_{k01}}{n_{k}}\right)$, where $\tilde{V}^2_{k1}$ and $\tilde{V}^2_{k,0}$ are as in part (a), and $\tilde{V}^2_{k01} = \frac{1}{n_k-1}\sum_{i \in \mathcal{I}_k}\{(\tilde{Y}_{ki}(1,p_k)-\tilde{Y}_{ki}(0,p_k)) - (\bar{\tilde{Y}}(1,p_k) - \bar{\tilde{Y}}(0,p_k)\}^2$. A conservative estimator of $\Var(\widehat{\text{DE}}^\pi)$ is $$\widehat{\Var}(\widehat{\text{DE}}^\pi)\ = \ \frac{1}{K^2}\sum_{k=1}^{K} \left(\frac{n_k}{m_k}\right)^2 \left(\frac{\hat{\tilde{V}}^2_{k1}}{n_{k1}} + \frac{\hat{\tilde{V}}^2_{k0}}{n_{k0}}\right).$$ 
\end{enumerate}
\label{prop_htvar1}
\end{proposition}

\begin{proof}
    The variance expressions in (a) and (b) follows from Theorems \ref{thm_var1} and \ref{thm_varDE} after setting $\pi_k(\cdot) = f_k(\cdot)$ and $f_k(\bm{a}) = \mathbbm{1}(\bm{a}^\top\bm{1} = n_{k1})/\binom{n_k}{n_{k1}}$. Moreover, the unbiasedness and conservativeness of the variance estimators in (a) and (b) follow from the properties of complete randomized design (see, e.g., \citealt{imbens2015causal}, Chapter 6). 
\end{proof}
The structure of the variance of $\hat{\mu}_{a,\text{HT}}$ resembles that of the estimated population mean in stratified random sampling without replacement, where the clusters act as the strata (see, e.g., \citealt{fuller2009sampling}, Chapter 1). Similarly, the variance of $\widehat{\text{DE}}^\pi_{\text{HT}}$ resembles that of the variance of the difference-in-means statistic in a stratified randomized experiment (see, e.g., \citealt{imbens2015causal}, Chapter 9). Moreover, the estimator of $\Var(\widehat{\text{DE}}^\pi_{\text{HT}})$ is unbiased if  $\tilde{Y}_{ki}(1,p_k)-\tilde{Y}_{ki}(0,p_k)$ in constant for all $i$, i.e., when the unit level causal effects based on the pooled potential outcomes are constant. This condition is analogous to the condition of unbiasedness for the standard Neyman's estimator of variance.

\subsubsection{Variance estimation of the indirect effect under stratified interference}

\begin{proposition}[Variance of the indirect effect estimator] \normalfont
 Under Assumptions~\ref{assump_partial},~\ref{assump_identify},~\ref{assump_indep},~\ref{assump_stratified},~and~\ref{assump_fixedprop}, 
\begin{align*}
\Var(\widehat{\text{IE}}^{\pi, \tilde{\pi}}_{a,\text{HT}})  = \frac{1}{K^2}\sum_{k=1}^{K}\frac{1}{|\mathcal{S}_k|^2}\left[\sum_{i=1}^{n_k}\tilde{c}_{i,a} \tilde{Y}^2_{ki}(a,p_k) + \mathop{\sum\sum}_{i \neq i'}\tilde{d}_{ii',a}\tilde{Y}_{ki}(a,p_k)\tilde{Y}_{ki'}(a,p_k)\right]. 
\end{align*}
where
$$\begin{aligned}
  \tilde{c}_{i,a} = & \sum_{\bm{s}}\frac{\left\{\frac{\pi_k(A_{ki} = a, \bm{A}_{k(-i)} = \bm{s})}{\pi_k(A_{ki} = a)} - \frac{\tilde{\pi}_k(A_{ki} = a, \bm{A}_{k(-i)} = \bm{s})}{\tilde{\pi}_k(A_{ki} = a)}\right\}^2}{f_k(A_{ki} = a, \bm{A}_{k(-i)} = \bm{s})},\\
  \tilde{d}_{ii',a} = & \sum_{\bm{s}} \frac{\left\{\frac{\pi_k(A_{ki} = a, \bm{A}_{k(-i)} = \bm{s})}{\pi_k(A_{ki} = a)} - \frac{\tilde{\pi}_k(A_{ki} = a, \bm{A}_{k(-i)} = \bm{s})}{\tilde{\pi}_k(A_{ki} = a)}\right\} \left\{\frac{\pi_k(A_{ki'} = a, \bm{A}_{k(-i')} = \bm{s})}{\pi_k(A_{ki'} = a)} - \frac{\tilde{\pi}_k(A_{ki'} = a, \bm{A}_{k(-i')} = \bm{s})}{\tilde{\pi}_k(A_{ki'} = a)}\right\}  }{f_k(A_{ki} = a, \bm{A}_{k(-i)} = \bm{s})} \\
  & \quad \times \mathbbm{1}(A_{ki} = a,A_{ki'} = a, \bm{A}_{k(-i)} = \bm{s},\bm{A}_{k(-i')} = \bm{s}).
  \end{aligned}
$$
\label{prop_varIE}
\end{proposition}

\begin{proof}
    
Without loss of generality, we set $a = 1$. Following similar steps as in the proof of Theorem~\ref{thm_var1}, we get
\begin{align*}
\Var(\widehat{\text{IE}}^{\pi,\Tilde{\pi}}_{1,\text{HT}})    
& = \frac{1}{K^2}\sum_{k=1}^{K}\frac{1}{|\mathcal{S}_k|^2}\Var\left\{\sum_{i=1}^{n_k}\sum_{\bm{s}}\gamma_{\bm{s}i}\mathbbm{1}(A_{ki} = 1, A_{k(-i) = \bm{s}}) \tilde{Y}_{ki}(1,p_k) \right\},
\end{align*}
where 
\begin{align*}
\gamma_{\bm{s}i} &= \frac{\pi_k(\bm{A}_{k(-i)} = \bm{s}\mid A_{ki} = 1) - \tilde{\pi}_k(\bm{A}_{k(-i)} = \bm{s}\mid A_{ki} = 1)}{f_k(A_{ki} = 1,\bm{A}_{k(-i)} = \bm{s})}
\end{align*}
Therefore,
\begin{align*}
\Var(\widehat{\text{IE}}^{\pi,\Tilde{\pi}}_{1,\text{HT}})    
& = \frac{1}{K^2}\sum_{k=1}^{K}\frac{1}{|\mathcal{S}_k|^2}(G_1+G_2),
\end{align*}
where
\begin{align*}
   G_1 
   & = \sum_{i=1}^{n_k} \tilde{Y}^2_{ki}(1,p_k)\Big[\sum_{\bm{s]}}\gamma^2_{\bm{s}i}f_k(A_{ki} = 1,\bm{A}_{k(-i)} = \bm{s})\{1 - f_k(A_{ki} = 1,\bm{A}_{k(-i)} = \bm{s})\} \nonumber\\
    & \quad - \mathop{\sum\sum}_{\bm{s} \neq \bm{s}'}\gamma_{\bm{s}i}\gamma_{\bm{s}'i}f_k(A_{ki} = 1,\bm{A}_{k(-i)} = \bm{s})f_k(A_{ki} = 1,\bm{A}_{k(-i)} = \bm{s}') \nonumber\\
    & = \sum_{i=1}^{n_k} \tilde{Y}^2_{ki}(1,p_k) \Big[\sum_{\bm{s}}\{\pi_k(\bm{A}_{k(-i)} = \bm{s}\mid A_{ki} = 1) - \tilde{\pi}_k(\bm{A}_{k(-i)} = \bm{s}\mid A_{ki} = 1)\}^2 \nonumber\\
    & \hspace{1in}  \times \frac{\{1-f_k(A_{ki} = 1,\bm{A}_{k(-i)} = \bm{s})\}}{f_k(A_{ki} = 1,\bm{A}_{k(-i)} = \bm{s})} \nonumber\\
    & \quad - \mathop{\sum\sum}_{\bm{s} \neq \bm{s}'}\{\pi_k(\bm{A}_{k(-i)} = \bm{s}\mid A_{ki} = 1) - \tilde{\pi}_k(\bm{A}_{k(-i)} = \bm{s}\mid A_{ki} = 1)\} \nonumber\\
    & \hspace{1in} \times \pi_k(\bm{A}_{k(-i)} = \bm{s}'\mid A_{ki} = 1) - \tilde{\pi}_k(\bm{A}_{k(-i)} = \bm{s}'\mid A_{ki} = 1)\}\Big] \nonumber\\
    &= \sum_{i=1}^{n_k} \tilde{c}_{i,1}\tilde{Y}^2_{ki}(1,p_k), \\
G_2 
&= \mathop{\sum\sum}_{i \neq i'}\tilde{Y}_{ki}(1,p_k)\tilde{Y}_{ki'}(1,p_k)\Big[ \sum_{\bm{s}}\gamma_{\bm{s}i}\gamma_{\bm{s}i'}\{\mathbbm{1}(A_{ki} = 1,\bm{A}_{k(-i)} = \bm{s}, A_{ki'} = 1,\bm{A}_{k(-i')} = \bm{s}) \nonumber\\
& \hspace{3in} \times f_k(A_{ki} = 1,\bm{A}_{k(-i)} = \bm{s}) \nonumber\\
    & \quad - f_k(A_{ki} = 1,\bm{A}_{k(-i)} = \bm{s})f_k(A_{ki'} = 1,\bm{A}_{k(-i')} = \bm{s})\} \nonumber\\
    & \quad + \mathop{\sum\sum}_{\bm{s} \neq \bm{s}'}\gamma_{\bm{s}i}\gamma_{\bm{s}'i}\{\mathbbm{1}(A_{ki} = 1,\bm{A}_{k(-i)} = \bm{s}, A_{ki'} = 1,\bm{A}_{k(-i')} = \bm{s}')\nonumber\\
   & \hspace{1in} \times f_k(A_{ki} = 1,\bm{A}_{k(-i)} = \bm{s}) - f_k(A_{ki} = 1,\bm{A}_{k(-i)} = \bm{s})f_k(A_{ki'} = 1,\bm{A}_{k(-i')} = \bm{s}')\}\Big] \nonumber\\
   &= \mathop{\sum\sum}_{i \neq i'}\tilde{Y}_{ki}(1,p_k)\tilde{Y}_{ki'}(1,p_k)\Big[\sum_{\bm{s}}\{\pi_k(\bm{A}_{k(-i)} = \bm{s}\mid A_{ki} = 1) - \tilde{\pi}_k(\bm{A}_{k(-i)} = \bm{s}\mid A_{ki} = 1)\} \nonumber \\
   & \hspace{1in} \times \{\pi_k(\bm{A}_{k(-i')} = \bm{s}\mid A_{ki'} = 1) - \tilde{\pi}_k(\bm{A}_{k(-i')} = \bm{s}|A_{ki'} = 1)\} \nonumber\\
   & \hspace{1in} \times \Big\{\frac{\mathbbm{1}(A_{ki} = 1,\bm{A}_{k(-i)} = \bm{s}, A_{ki'} = 1,\bm{A}_{k(-i')} = \bm{s})}{f_k(A_{ki} = 1,\bm{A}_{k(-i)} = \bm{s})}-1\Big\} \nonumber\\
   & \quad - \mathop{\sum\sum}_{\bm{s} \neq \bm{s}'}\{\pi_k(\bm{A}_{k(-i)} = \bm{s}\mid A_{ki} = 1) - \tilde{\pi}_k(\bm{A}_{k(-i)} = \bm{s}\mid A_{ki} = 1)\} \nonumber\\
   & \hspace{1in} \times \{\pi_k(\bm{A}_{k(-i')} = \bm{s}'\mid A_{ki'} = 1) - \tilde{\pi}_k(\bm{A}_{k(-i')} = \bm{s}'\mid A_{ki'} = 1)\} \nonumber\\
   &= \mathop{\sum\sum}_{i \neq i'}\tilde{d}_{ii',1}\tilde{Y}_{ki}(1,p_k)\tilde{Y}_{ki'}(1,p_k).
\end{align*}
\end{proof}

\subsubsection{Additional results on variance estimation under additive interference}
\label{appsec_var_additive}

In this section, we provide closed-form expressions of the variances of the Horvtiz-Thompson estimators of the direct and indirect effects under additive interference (Assumption \ref{assump_additive}). We provide estimators of these variances and show them that they are conservative in finite samples. 

Proposition \ref{prop_varDE_additive} provides a closed-form expression of the variance of $\widehat{\text{DE}}^\pi_{\text{HT}}$. 
\begin{proposition} \normalfont
    Under Assumptions~\ref{assump_partial},~\ref{assump_identify},~\ref{assump_indep},~and~\ref{assump_additive},
\begin{align}
    \Var(\widehat{\text{DE}}^\pi_{\text{HT}}) = & \Var(\hat{\mu}^\pi_{1,\text{HT}}) + \Var(\hat{\mu}^\pi_{0,\text{HT}}) - 2 \Cov(\hat{\mu}^\pi_{1,\text{HT}},\hat{\mu}^\pi_{0,\text{HT}}),
\end{align}
where $\Var(\hat{\mu}^\pi_{1,\text{HT}})$ and $\Var(\hat{\mu}^\pi_{0,\text{HT}})$ are as in Proposition \ref{prop_additive_mu1} and 
\begin{align}
\Cov(\hat{\mu}^\pi_{1,\text{HT}},\hat{\mu}^\pi_{0,\text{HT}}) &=   \frac{1}{K^2}\sum_{k = 1}^{K} \frac{1}{|\mathcal{S}_k|^2}\left(\sum_{j \in \mathcal{S}_k} \Lambda_{1,k,j} + \mathop{\sum\sum}_{j \neq j' \in \mathcal{S}_k}\Lambda_{2,k,j,j'} \right),   
\end{align}
where 
\begin{align}
\Lambda_{1,k,j} \ = \ & - \left\{(1,\bm{\pi}^\top_k(\cdot|A_{ki^*} = 1))\tilde{\bm{\beta}}_{kj} \right\}\left\{(1,\bm{\pi}^\top_k(\cdot|A_{ki^{*}} = 0))\tilde{\bm{\beta}}_{kj} \right\},\\
       \Lambda_{2,k,j,j'} \ = \ &\sum_{\tilde{\bm{s}}'}\frac{\mathbbm{\pi}^2_k(A_{ki^{*}} = 1, A_{ki^{*'}} = 0, \bm{A}_{k(-i^{*},-i^{*'})} = \tilde{\bm{s}}') \left\{(1,\bm{a}^\top_{kjj'})\tilde{\bm{\beta}}_{kj} \right\}\left\{(1,\bm{a}^\top_{kjj'})\tilde{\bm{\beta}}_{kj'} \right\}}{f_k(A_{ki^*} = 1, A_{ki^{*'}} = 0, \bm{A}_{k(-i^{*},-i^{*'})} = \tilde{\bm{s}}')\pi_k(A_{ki^*} = 1)\pi_k(A_{ki^{*'}} = 0)} \nonumber\\
    &\hspace{0.5cm} - \left\{(1,\bm{\pi}^\top_k(\cdot|A_{ki^*} = 1))\tilde{\bm{\beta}}_{kj} \right\}\left\{(1,\bm{\pi}^\top_k(\cdot|A_{ki^{*'}} = 0))\tilde{\bm{\beta}}_{kj'} \right\}. 
\end{align}
where $\bm{a}_{kjj'}$ is the vector of treatment assignments with $A_{ki^{*}} = 1, A_{ki^{*'}} = 0, \bm{A}_{k(-i^{*},-i^{*'})} = \tilde{\bm{s}}'$, and for $a \in \{0,1\}$, $\bm{\pi}_k(\cdot|A_{ki^*} = a)$ is the vector of conditional probabilities whose $i$th element is $\pi_k(A_{ki} = 1 \mid A_{ki^*} = a)$.
\label{prop_varDE_additive}
\end{proposition}

\begin{proof}
Now, following similar steps as in the proof of Proposition \ref{prop_additive_mu1},
\begin{align*}
&\Cov(\hat{\mu}^\pi_{1,\text{HT}},\hat{\mu}^\pi_{0,\text{HT}}) \nonumber\\
    & = \Cov\left(\frac{1}{K}\sum_{k=1}^{K}\frac{1}{|\mathcal{S}_k|}\sum_{j} \frac{\mathbbm{1}(A_{ki^*} = 1)\pi_k(\bm{A}_{k(-i^*)}\mid A_{i^*} = 1)}{f_k(\bm{A}_k)}{Y}^{\text{obs}}_{kj},\right. \nonumber\\
    & \quad \quad \left.\frac{1}{K}\sum_{k=1}^{K}\frac{1}{|\mathcal{S}_k|}\sum_{j} \frac{\mathbbm{1}(A_{ki^*} = 0)\pi_k(\bm{A}_{k(-i^*)}\mid A_{i^*} = 0)}{f_k(\bm{A}_k)}{Y}^{\text{obs}}_{kj} \right) \nonumber\\
    & = \frac{1}{K^2}\sum_{k = 1}^{K} \frac{1}{|\mathcal{S}_k|^2}\left(\sum_{j \in \mathcal{S}_k} \Lambda_{1,k,j} + \mathop{\sum\sum}_{j \neq j' \in \mathcal{S}_k}\Lambda_{2,k,j,j'} \right),
\end{align*}
where
\begin{align*}
\Lambda_{1,k,j} &= \Cov\left\{\frac{\mathbbm{1}(A_{ki^*} = 1)\pi_k(\bm{A}_{k(-i^*)}\mid A_{i^*} = 1)}{f_k(\bm{A}_k)}{Y}^{\text{obs}}_{kj}, \frac{\mathbbm{1}(A_{ki^*} = 0)\pi_k(\bm{A}_{k(-i^*)}\mid A_{i^*} = 0)}{f_k(\bm{A}_k)}{Y}^{\text{obs}}_{kj} \right\}  \nonumber\\
& = -\E\left\{\frac{\mathbbm{1}(A_{ki^*} = 1)\pi_k(\bm{A}_{k(-i^*)}\mid A_{i^*} = 1)}{f_k(\bm{A}_k)}{Y}^{\text{obs}}_{kj})\right\}\E\left\{ \frac{\mathbbm{1}(A_{ki^*} = 0)\pi_k(\bm{A}_{k(-i^*)}\mid A_{i^*} = 0)}{f_k(\bm{A}_k)}{Y}^{\text{obs}}_{kj}) \right\} \nonumber\\
& = -\left\{\sum_{\bm{s}}\pi_k(\bm{A}_{k(-i^*)} = \bm{s}\mid A_{i^*} = 1){Y}_{kj}(A_{ki^*} = 1, A_{k(-i^*)} = \bm{s})\right\} \nonumber \\
& \quad \times \left\{\sum_{\bm{s}}\pi_k(\bm{A}_{k(-i^*)} = \bm{s}\mid A_{i^*} = 0){Y}_{kj}(A_{ki^*} = 0, A_{k(-i^*)} = \bm{s})\right\} \nonumber \\
& = - \left\{(1,\bm{\pi}^\top_k(\cdot|A_{ki^*} = 1))\tilde{\bm{\beta}}_{kj} \right\}\left\{(1,\bm{\pi}^\top_k(\cdot|A_{ki^{*}} = 0))\tilde{\bm{\beta}}_{kj} \right\},
\end{align*}
   where the last equality follows from Equation \ref{eq_A4.6_1}.
Moreover, 
\begin{align}
        &\Lambda_{2,k,j,j'} \nonumber\\
       \ = \ & \nonumber \Cov\left\{\frac{\mathbbm{1}(A_{ki^*} = 1)\pi_k(\bm{A}_{k(-i^*)}\mid A_{i^*} = 1)}{f_k(\bm{A}_k)}{Y}^{\text{obs}}_{kj}, \frac{\mathbbm{1}(A_{ki^{*'}} = 0)\pi_k(\bm{A}_{k(-i^{*'})}\mid A_{i^{*'}} = 0)}{f_k(\bm{A}_k)}{Y}^{\text{obs}}_{kj'} \right\}\\
   \ = \ & \Cov \left\{ \sum_{\bm{s}} \frac{\mathbbm{1}(A_{ki^*} = 1, \bm{A}_{k(-i^*)} = \bm{s})\pi_k(\bm{A}_{k(-i^*)} = \bm{s}\mid A_{i^*} = 1)}{f_k(\bm{A}_k)}{Y}_{kj}(A_{ki^*} = 1, \bm{A}_{k(-i^*)} = \bm{s}),\right. \nonumber \\
    & \quad \quad \left. \sum_{\tilde{\bm{s}}} \frac{\mathbbm{1}(A_{ki^{*'}} = 0, \bm{A}_{k(-i^{*'})} = \tilde{\bm{s}})\pi_k(\bm{A}_{k(-i^{*'})} = \tilde{\bm{s}}\mid A_{i^{*'}} = 0)}{f_k(\bm{A}_k)}{Y}_{kj'}(A_{ki^{*'}} = 0, \bm{A}_{k(-i^{*'})} = \tilde{\bm{s}}) \right\} \nonumber\\
   \ = \ & \sum_{\bm{s}}\sum_{\tilde{\bm{s}}} \pi_k(\bm{A}_{k(-i^*)} = \bm{s}\mid A_{i^*} = 1)\pi_k(\bm{A}_{k(-i^{*'})} = \tilde{\bm{s}}\mid A_{i^*} = 0){Y}_{kj}(A_{ki^*} = 1, \bm{A}_{k(-i^*)} = \bm{s}) \nonumber\\
    & \quad \times {Y}_{kj'}(A_{ki^{*'}} = 0, \bm{A}_{k(-i^{*'})} = \bm{s})\left\{\frac{f_k(A_{ki^{*}} = 1, A_{ki^{*'}} = 0, \bm{A}_{k(-i^*)} = \bm{s}, \bm{A}_{k(-i^{*'})} = \tilde{\bm{s}})}{f_k(A_{ki^{*}} = 1,\bm{A}_{k(-i^*)} = \bm{s})f_k(A_{ki^{*'}} = 0,\bm{A}_{k(-i^{*'})} = \tilde{\bm{s}})} - 1   \right\} \nonumber\\
     \ = \ &  \sum_{\tilde{\bm{s}}'}\frac{\pi^2_k(A_{ki^{*}} = 1, A_{ki^{*'}} = 0, \bm{A}_{k(-i^*, -i^{*'})} = \tilde{\bm{s}}')}{f_k(A_{ki^{*}} = 1, A_{ki^{*'}} = 0, \bm{A}_{k(-i^*, -i^{*'})} = \tilde{\bm{s}}')\pi_k(A_{ki^*} = 1)\pi_k(A_{ki^{*'}} = 0)} \nonumber\\
& \quad  \times Y_{kj}(A_{ki^{*}} = 1, A_{ki^{*'}} = 0, \bm{A}_{k(-i^*,-i^{*'})} = \bm{s})Y_{kj'}(A_{ki^{*}} = 1,A_{ki^{*'}} = 1, \bm{A}_{k(-i^*,-i^{*'})} = \bm{s}) \nonumber\\
& - \left\{\sum_{\bm{s}}\pi_{k}(\bm{A}_{k(-i^*)} = \bm{s}\mid A_{ki^*} = 1)Y_{kj}(A_{ki^*} = 1, \bm{A}_{k(-i^*)} = \bm{s}) \right\} \nonumber \\ 
& \times \left\{\sum_{\bm{s}}\pi_{k}(\bm{A}_{k(-i^{*'})} = \bm{s}\mid A_{ki^{*'}} = 0)Y_{kj'}(A_{ki^{*'}} = 0, \bm{A}_{k(-i^{*'})} = \bm{s}) \right\} \nonumber \\
\ = \ & \sum_{\tilde{\bm{s}}'}\frac{\mathbbm{\pi}^2_k(A_{ki^{*}} = 1, A_{ki^{*'}} = 0, \bm{A}_{k(-i^{*},-i^{*'})} = \tilde{\bm{s}}') \left\{(1,\bm{a}^\top_{kjj'})\tilde{\bm{\beta}}_{kj} \right\}\left\{(1,\bm{a}^\top_{kjj'})\tilde{\bm{\beta}}_{kj'} \right\}}{f_k(A_{ki^*} = 1, A_{ki^{*'}} = 0, \bm{A}_{k(-i^{*},-i^{*'})} = \tilde{\bm{s}}')\pi_k(A_{ki^*} = 1)\pi_k(A_{ki^{*'}} = 0)} \nonumber\\
    &\hspace{0.5cm} - \left\{(1,\bm{\pi}^\top_k(\cdot|A_{ki^*} = 1))\tilde{\bm{\beta}}_{kj} \right\}\left\{(1,\bm{\pi}^\top_k(\cdot|A_{ki^{*'}} = 0))\tilde{\bm{\beta}}_{kj'} \right\}.
\end{align}
\end{proof}

Theorem \ref{thm_conservative_DE} shows that the estimated variance of the direct effect, based on the plug-in regression estimator is conservative in finite samples.
\begin{theorem} \normalfont
Let Assumptions~\ref{assump_partial},~\ref{assump_identify},~\ref{assump_indep},~and~\ref{assump_additive} hold, and let $\widehat{\Var}(\widehat{\text{DE}}^\pi_{\text{HT}})$ be the estimator of $\Var(\widehat{\text{DE}}^\pi_{\text{HT}})$ based on $\hat{\tilde{\bm{\beta}}}_{kj}$. Then,
$$\E\{\widehat{\Var}(\widehat{\text{DE}}^\pi_{\text{HT}})\} \geq \Var(\widehat{\text{DE}}^\pi_{\text{HT}}).$$
\label{thm_conservative_DE}
\end{theorem}
\begin{proof}
\begin{align}
\widehat{\text{DE}}^\pi_{\text{HT}} & = \sum_{k=1}^{K}\sum_{j \in \mathcal{S}_k}\frac{\mathbbm{1}(A_{ki^*} = 1)\pi_{k}(\bm{A}_{k(-i^*)} \mid A_{ki^*} = 1) - \mathbbm{1}(A_{ki^*} = 0)\pi_{k}(\bm{A}_{k(-i^*)} \mid A_{ki^*} = 0)}{K|\mathcal{S}_k|f_k(\bm{A}_k)}(1,\bm{A}^\top_{k})\tilde{\bm{\beta}}_{kj}  
\end{align}
Thus, akin to $\hat{\mu}^{\pi}_{a,\text{HT}}$, $\widehat{\text{DE}}^\pi_{\text{HT}}$ is of the form $\widehat{\text{DE}}^\pi_{\text{HT}} = \sum_{k}\sum_{j}\tilde{\bm{\beta}}^\top_{kj}\bm{\psi}_{kj}$, for some $|\mathcal{S}_k|\times 1$ vector $\bm{\psi}_{kj}$. Therefore, the proof follows directly from the proof of Theorem \ref{thm_conservative_mu1}.
\end{proof}

Next, we focus on the variance estimation problem for the estimated indirect effect.
Proposition \ref{prop_additive_IE} provides a closed-form expression of $\Var(\widehat{\text{IE}}^{\pi,\tilde{\pi}}_{a,\text{HT}})$.

\begin{proposition} \normalfont
 Under Assumptions~\ref{assump_partial},~\ref{assump_identify},~\ref{assump_indep},~and~\ref{assump_additive},  
\begin{align}
\Var(\widehat{\text{IE}}^{\pi,\tilde{\pi}}_{a,\text{HT}}) =   \frac{1}{K^2}\sum_{k = 1}^{K} \frac{1}{|\mathcal{S}_k|^2}\left(\sum_{j \in \mathcal{S}_k} \Lambda_{1,k,j} + \mathop{\sum\sum}_{j \neq j' \in \mathcal{S}_k}\Lambda_{2,k,j,j'} \right),  
\end{align}
where 
\begin{align}
\Lambda_{1,k,j} \ = \ & \sum_{\bm{s}}\frac{\{\pi_{k}(\bm{A}_{k(-i^*)} = \bm{s}\mid A_{ki^*} = a) - \tilde{\pi}_{k}(\bm{A}_{k(-i^*)} = \bm{s}\mid A_{ki^*} = a)\}^2}{f_k(A_{ki^*} = a, \bm{A}_{k(-i^*)} = \bm{s})}\left\{(1,\bm{a}^\top_{kj})\tilde{\bm{\beta}}_{kj} \right\}^2 \nonumber\\
 & \quad  - \left\{(0,\{\bm{\pi}_k(\cdot|A_{ki^*} = a)-\tilde{\bm{\pi}}_k(\cdot|A_{ki^*} = a)\}^\top)\tilde{\bm{\beta}}_{kj} \right\}^2,\\
       \Lambda_{2,k,j,j'} \ = \ &\sum_{\tilde{\bm{s}}}\frac{\left\{(1,\bm{a}^\top_{kjj'})\tilde{\bm{\beta}}_{kj} \right\}\left\{(1,\bm{a}^\top_{kjj'})\tilde{\bm{\beta}}_{kj'} \right\}}{f_k(A_{ki^*} = a, A_{ki^{*'}} = a, \bm{A}_{k(-i^{*},-i^{*'})} = \tilde{\bm{s}})\pi_k(A_{ki^*} = a)\pi_k(A_{ki^{*'}} = a)} \nonumber\\
    &\hspace{0.5cm} \times \left\{\frac{\pi_k(A_{ki^{*}} = a, A_{ki^{*'}} = a, \bm{A}_{k(-i^{*},-i^{*'})} = \tilde{\bm{s}})}{\pi_k(A_{ki^*} = a)} - \frac{\tilde{\pi}_k(A_{ki^{*}} = a, A_{ki^{*'}} = a, \bm{A}_{k(-i^{*},-i^{*'})} = \tilde{\bm{s}})}{\tilde{\pi}_k(A_{ki^*} = a)} \right\} \nonumber\\
    &\hspace{0.5cm}\times \left\{\frac{\pi_k(A_{ki^{*}} = a, A_{ki^{*'}} = a, \bm{A}_{k(-i^{*},-i^{*'})} = \tilde{\bm{s}})}{\pi_k(A_{ki^{*'}} = a)} - \frac{\tilde{\pi}_k(A_{ki^{*}} = a, A_{ki^{*'}} = a, \bm{A}_{k(-i^{*},-i^{*'})} = \tilde{\bm{s}})}{\tilde{\pi}_k(A_{ki^{*'}} = a)} \right\} \nonumber\\
    & \quad - (0,\{\bm{\pi}_k(\cdot|A_{ki^*} = a)-\tilde{\bm{\pi}}_k(\cdot|A_{ki^*} = a)\}^\top)\tilde{\bm{\beta}}_{kj} \nonumber\\
    & \quad \quad \times (0,\{\bm{\pi}_k(\cdot|A_{ki^{*'}} = a)-\tilde{\bm{\pi}}_k(\cdot|A_{ki^{*'}} = a)\}^\top)\tilde{\bm{\beta}}_{kj'}. 
\end{align}
where $\bm{a}_{kj}$ is the vector of treatment assignments with $A_{ki^{*}} = a$ and $\bm{A}_{k(-i^{*})} = \bm{s}$; $\bm{a}_{kjj'}$ is the vector of treatment assignments with $A_{ki^{*}} = a, A_{ki^{*'}} = a, \bm{A}_{k(-i^{*},-i^{*'})} = \tilde{\bm{s}}$; $\bm{\pi}_k(\cdot|A_{ki^*} = a)$ and $\tilde{\bm{\pi}}_k(\cdot|A_{ki^*} = a)$ are the vectors of conditional probabilities whose $i$th elements are $\pi_k(A_{ki} = 1 \mid A_{ki^*} = a)$ and $\tilde{\pi}_k(A_{ki} = 1 \mid A_{ki^*} = a)$, respectively.
\label{prop_additive_IE}
\end{proposition}

\begin{proof}
The estimator of the indirect effect can be written as,
\begin{align}
   & \widehat{\text{IE}}^{\pi,\tilde{\pi}}_{a,\text{HT}} \nonumber\\
   & = \frac{1}{K}\sum_{k=1}^{K}\frac{1}{|\mathcal{S}_k|}\sum_{j \in \mathcal{S}_k}\sum_{\bm{s}}\mathbbm{1}(A_{ki^*} = a, \bm{A}_{k(-i^*)} = \bm{s})\frac{\pi_{k}(\bm{A}_{k(-i^*)} = \bm{s} \mid A_{ki^*} = a) - \tilde{\pi}_{k}(\bm{A}_{k(-i^*)} = \bm{s} \mid A_{ki^*} = a)}{f_k(A_{ki^*} = a, \bm{A}_{k(-i^*)} =\bm{s})} \nonumber\\
   &\hspace{0.5cm}\times Y_{kj}(A_{ki^*} = a, \bm{A}_{k(-i^*)} =\bm{s}).
    \label{eq_A_IEhat}
\end{align}
We note that, $\widehat{\text{IE}}^{\pi,\tilde{\pi}}_{a,\text{HT}}$ has the same form as $\hat{\mu}^{\pi}_{a,\text{HT}}$ in Equation \ref{eq_A_muahat}, with $\pi_{k}(\bm{A}_{k(-i^*)} = \bm{s} \mid A_{ki^*} = a)$ being replaced by $\pi_{k}(\bm{A}_{k(-i^*)} = \bm{s} \mid A_{ki^*} = a) - \tilde{\pi}_{k}(\bm{A}_{k(-i^*)} = \bm{s} \mid A_{ki^*} = a)$. Thus, the desired variance expression can be derived by following the proofs of Theorem \ref{thm_generalvar} an \ref{prop_additive_mu1} exactly.
\end{proof}

Theorem \ref{thm_conservative_IE} shows that the estimated variance of the indirect effect, based on the plug-in regression estimator is conservative in finite samples.
\begin{theorem} \normalfont
Let Assumptions~\ref{assump_partial},~\ref{assump_identify},~\ref{assump_indep},~and~\ref{assump_additive} hold, and let $\widehat{\Var}(\widehat{\text{IE}}^{\pi,\tilde{\pi}}_{a,\text{HT}})$ be the estimator of $\Var(\widehat{\text{IE}}^{\pi,\tilde{\pi}}_{a,\text{HT}})$ based on $\hat{\tilde{\bm{\beta}}}_{kj}$. Then,
$$\E\{\widehat{\Var}(\widehat{\text{IE}}^{\pi,\tilde{\pi}}_{a,\text{HT}})\} \geq \Var(\widehat{\text{IE}}^{\pi,\tilde{\pi}}_{a,\text{HT}}).$$
\label{thm_conservative_IE}
\end{theorem}
\begin{proof}
\begin{align}
\widehat{\text{IE}}^{\pi,\tilde{\pi}}_{a,\text{HT}} & = \sum_{k=1}^{K}\sum_{j \in \mathcal{S}_k}\frac{\mathbbm{1}(A_{ki^*} = a)\pi_{k}(\bm{A}_{k(-i^*)} \mid A_{ki^*} = a) - \mathbbm{1}(A_{ki^*} = a)\tilde{\pi}_{k}(\bm{A}_{k(-i^*)} \mid A_{ki^*} = a)}{K|\mathcal{S}_k|f_k(\bm{A}_k)}(1,\bm{A}^\top_{k})\tilde{\bm{\beta}}_{kj}  
\end{align}
Thus, akin to $\hat{\mu}^{\pi}_{a,\text{HT}}$, $\widehat{\text{IE}}^{\pi,\tilde{\pi}}_{a,\text{HT}}$ is of the form $\widehat{\text{IE}}^{\pi,\tilde{\pi}}_{a,\text{HT}} = \sum_{k}\sum_{j}\tilde{\bm{\beta}}^\top_{kj}\bm{\psi}_{kj}$, for some $|\mathcal{S}_k|\times 1$ vector $\bm{\psi}_{kj}$. Therefore, the proof follows directly from the proof of Theorem \ref{thm_conservative_mu1}.
\end{proof}

\subsubsection{Inference on treatment effects with multiple key-intervention units.}
\label{sec_appendix_multiple}

In this section, we consider the setting with multiple key-intervention units and focus on the Horvitz-Thompson estimator of $\tau^\pi = \frac{1}{K}\sum_{k=1}^{K}\left[\frac{1}{|\mathcal{S}_k|}\sum_{j \in \mathcal{S}_k}\left\{\sum_{\bm{a}\in \{0,1\}^{n_k}}\pi_{k}(\bm{a}|\mathcal{C}_{kj})Y_{kj}(\bm{a})\right\}\right],$ where $\mathcal{C}_{kj} = \{\bm{a} \in \{0,1\}^{n_k}: \sum_{s=1}^{r}A_{ki^*_s}/|\bm{i}^*| = p^*\}$, where $\bm{i}^*$ is the set of key-intervention units of unit $j$ (of size $|\bm{i}^*|$), and $p^* \in [0,1]$. Here, for each unit $j \in \mathcal{S}_k$, the stochastic intervention treats a fixed proportion $p^*$ of its key-intervention units. For simplicity, we set $\pi_k(\cdot\mid\mathcal{C}_{kj}) = f_k(\cdot\mid\mathcal{C}_{kj})$, i.e., given that $p^*$ proportion of key-intervention units are treated, the assignment mechanism under the stochastic intervention
is the same as that under the actual intervention.
The resulting Horvitz-Thompson estimator can be written as,
\begin{align*}
    \hat{\tau}^\pi_{\text{HT}} = \frac{1}{K}\sum_{k=1}^{K}\frac{1}{|\mathcal{S}_k|}\sum_{j \in \mathcal{S}_k}\frac{\mathbbm{1}(\bm{A}^\top_{k\bm{i}^*}\bm{1} = |\bm{i}^*|p^*)}{f_k(\bm{A}^\top_{k\bm{i}^*}\bm{1} = |\bm{i}^*|p^*)}Y^{\text{obs}}_{kj}.
\end{align*}

To point identify the variance of $\hat{\tau}^\pi_{\text{HT}}$ in this case, we consider an analog of the stratified interference assumption for multiple key-intervention units.
\begin{assumption}[Stratified interference for multiple key-intervention units] \normalfont
    For unit $j \in \mathcal{S}_k$, if $\bm{a}, \bm{a}' \in \{0,1\}^{n_k}$ are such that $\bm{a}^\top_{\bm{i}^*}\bm{1} = {\bm{a}'}^\top_{\bm{i}^*}\bm{1}$ and $\bm{a}^\top\bm{1} = \bm{a}'^\top\bm{1}$, then $Y_{kj}(\bm{a}) = Y_{kj}(\bm{a}')$. 
    \label{assump_doublestr}
\end{assumption}
Assumption~\ref{assump_doublestr} states that the potential outcome of a unit $j \in \mathcal{S}_k$ depends on the treatment assignment of the intervention units in cluster $k$ only through the proportion of treated key-intervention units and the proportion of overall treated intervention units. In the single key-intervention unit case, this assumption becomes equivalent to Assumption~\ref{assump_stratified}. 
Under Assumption~\ref{assump_doublestr}, we can write the potential outcome $Y_{kj}(\bm{a})$ as $Y_{kj}\left(\frac{a^\top_{\bm{i}^*}\bm{1}}{|\bm{i}^*|}, \frac{\bm{a}^\top \bm{1}}{n_k}\right)$. 

Similar to the single key-intervention unit case, we now define the pooled potential outcome. However, unlike the previous case, here the pooled potential outcomes are indexed by subsets of the intervention units. Formally, the pooled potential outcome for a subset $\bm{i}$ of $\mathcal{I}_k$ and fixed $p^*,p_k \in (0,1)$ is $\tilde{Y}_{k\bm{i}}(p^*,p_k) = \sum_{j\in \mathcal{S}_k} \mathbbm{1}(j \leftarrow \bm{i})Y_{j}(\frac{a^\top_{\bm{i}^*}\bm{1}}{|\bm{i}^*|} = p^*, \frac{\bm{a}^\top \bm{1}}{n_k} = p_k)$. The corresponding pooled observed outcome is $\tilde{Y}^{\text{obs}}_{k\bm{i}} = \tilde{Y}_{k\bm{i}}(\frac{\bm{A}_{k\bm{i}}^\top\bm{1}}{|\bm{i}|},p_k)$.
Also, let $\mathcal{G}_k = \{\bm{i}\subseteq \mathcal{I}_k: f_k(\bm{A}^\top_{k\bm{i}}\bm{1} = |\bm{i}|p^*)>0\}$ be the subset of intervention units $\bm{i}$ in cluster $k$ for which there is a strictly positive probability of observing $p^*$ proportion of treated units. 
In Theorem \ref{thm_multiple}, we obtain a closed-form expression of the variance of $\hat{\tau}^\pi_{\text{HT}}$. 
\begin{theorem} \normalfont
Under Assumptions~\ref{assump_partial},~\ref{assump_identify},~\ref{assump_indep},~\ref{assump_fixedprop},~\ref{assump_doublestr}, and $\pi_k(\cdot|\mathcal{C}_{kj}) = f_k(\cdot|\mathcal{C}_{kj})$,
\begin{align*}
 \Var(\hat{\tau}^\pi_{\text{HT}}) = \frac{1}{K^2}\sum_{k=1}^{K}\frac{1}{|\mathcal{S}_k|^2}\left[\sum_{\bm{i} \in \mathcal{G}_k}\tilde{c}_{\bm{i}} \tilde{Y}^2_{k\bm{i}}(p^*,p_k) + \mathop{\sum\sum}_{\bm{i} \neq \bm{i}' \in \mathcal{G}_k}\tilde{d}_{\bm{i}\bm{i}'}\tilde{Y}_{k\bm{i}}(p^*,p_k)\tilde{Y}_{k\bm{i}'}(p^*,p_k)\right],   
\end{align*}
where $\tilde{c}_{\bm{i}} = \frac{1}{f_k(\bm{A}^\top_{k\bm{i}}\bm{1} = |\bm{i}|p^*)}-1$ and $\tilde{d}_{\bm{i}\bm{i}'} = \frac{f_k(\bm{A}^\top_{k\bm{i}}\bm{1} = |\bm{i}|p^*, \bm{A}^\top_{k\bm{i}'}\bm{1} = |\bm{i}|p^*)}{f_k(\bm{A}^\top_{k\bm{i}}\bm{1} = |\bm{i}|p^*)f_k(\bm{A}^\top_{k\bm{i}'}\bm{1} = |\bm{i}'|p^*)} - 1$.
    \label{thm_multiple}
\end{theorem}
The term $\sum_{\bm{i} \in \mathcal{G}_k}\tilde{c}_{\bm{i}} \tilde{Y}^2_{k\bm{i}}(p^*,p_k)$ can be estimated unbiasedly using the Horvitz-Thompson estimator $\sum_{\bm{i} \in \mathcal{G}_k}\tilde{c}_{\bm{i}} \frac{\mathbbm{1}(\bm{A}^\top_{k\bm{i}}\bm{1} = |\bm{i}|p^*)}{f_k(\bm{A}^\top_{k\bm{i}}\bm{1} = |\bm{i}|p^*)} \tilde{Y}_{k\bm{i}}^2$. 

Similarly, the Horvitz-Thompson estimator $\mathop{\sum\sum}_{\bm{i} \neq \bm{i}' \in \mathcal{G}_k}\tilde{d}_{\bm{i}\bm{i}'}\frac{\mathbbm{1}(\bm{A}^\top_{k\bm{i}}\bm{1} = |\bm{i}|p^*, \bm{A}^\top_{k\bm{i}'}\bm{1} = |\bm{i}'|p^*)}{f_k(\bm{A}^\top_{k\bm{i}}\bm{1} = |\bm{i}|p^*, \bm{A}^\top_{k\bm{i}'}\bm{1} = |\bm{i}'|p^*)} \tilde{Y}^{\text{obs}}_{k\bm{i}}\tilde{Y}^{\text{obs}}_{k\bm{i}'}$ is unbiased for the term $\mathop{\sum\sum}_{\bm{i} \neq \bm{i}' \in \mathcal{G}_k}\tilde{d}_{\bm{i}\bm{i}'}\tilde{Y}_{k\bm{i}}(p^*,p_k)\tilde{Y}_{k\bm{i}'}(p^*,p_k)$, provided the design satisfies the \textit{measurability} condition $f_k(\bm{A}^\top_{k\bm{i}}\bm{1} = |\bm{i}|p^*, \bm{A}^\top_{k\bm{i}'}\bm{1} = |\bm{i}'|p^*)>0$, i.e., for all subsets of intervention units $\bm{i},\bm{i}' \in \mathcal{G}_k$, the design allows for assignments that treat $p^*$ proportion of units in both $\bm{i}$ and $\bm{i}'$. 
If the design is not measurable, then we can instead obtain a conservative estimator of the variance. 
Finally, for some subsets $\bm{i}$, $|\bm{i}|p^*$ may not be an integer. In that case, we replace it with its nearest integer $\text{int}(|\bm{i}|p^*)$. Thus, for a measurable design, we can estimate ${\Var}(\hat{\tau^\pi_{\text{HT}}})$ as 
\begin{align*}
 \widehat{\Var}(\hat{\tau}^\pi_{\text{HT}}) &= \frac{1}{K^2}\sum_{k=1}^{K}\frac{1}{|\mathcal{S}_k|^2}\Big[\sum_{\bm{i} \in \mathcal{G}_k}\tilde{c}_{\bm{i}} \frac{\mathbbm{1}(\bm{A}^\top_{k\bm{i}}\bm{1} = \text{int}(|\bm{i}|p^*))}{f_k(\bm{A}^\top_{k\bm{i}}\bm{1} = \text{int}(|\bm{i}|p^*))} (\tilde{Y}^{\text{obs}}_{k\bm{i}})^2 \nonumber \\
 & \quad + \mathop{\sum\sum}_{\bm{i} \neq \bm{i}' \in \mathcal{G}_k}\tilde{d}_{\bm{i}\bm{i}'}\frac{\mathbbm{1}(\bm{A}^\top_{k\bm{i}}\bm{1} = \text{int}(|\bm{i}|p^*), \bm{A}^\top_{k\bm{i}'}\bm{1} = \text{int}(|\bm{i}'|p^*))}{f_k(\bm{A}^\top_{k\bm{i}}\bm{1} = \text{int}(|\bm{i}|p^*), \bm{A}^\top_{k\bm{i}'}\bm{1} = \text{int}(|\bm{i}'|p^*))} \tilde{Y}^{\text{obs}}_{k\bm{i}}\tilde{Y}^{\text{obs}}_{k\bm{i}'}\Big].   
\end{align*}


\subsubsection{Proof of Theorem \ref{thm_multiple}}

Following the proof of Theorem \ref{thm_var1}, under Assumption \ref{assump_doublestr}, we can write
\begin{align*}
    \hat{\tau}^\pi_{\text{HT}} &= \frac{1}{K}\sum_{k=1}^{K}\frac{1}{|\mathcal{S}_k|}\sum_{\bm{i} \in \mathcal{G}_k}\frac{\mathbbm{1}(\bm{A}^\top_{k\bm{i}}\bm{1} = |\bm{i}|p^*)}{f_k(\bm{A}^\top_{k\bm{i}}\bm{1} = |\bm{i}|p^*)}\Tilde{Y}_{k\bm{i}}(p^*,p_k).
\end{align*}
Thus,
\begin{align*}
\Var(\hat{\tau}^\pi_{\text{HT}}) &= \frac{1}{K^2}\sum_{k=1}^{K}\frac{1}{|\mathcal{S}_k|^2}\Big[ \sum_{\bm{i} \in \mathcal{G}_k}\frac{1 - f_k(\bm{A}^\top_{k\bm{i}}\bm{1} = |\bm{i}|p^*)}{f_k(\bm{A}^\top_{k\bm{i}}\bm{1} = |\bm{i}|p^*)}\tilde{Y}^2_{k\bm{i}}(p^*,p_k) \nonumber\\
& \quad + \mathop{\sum\sum}_{\bm{i} \neq \bm{i}'\in \mathcal{G}_k}\tilde{Y}_{k\bm{i}}(p^*,p_k)\tilde{Y}_{k\bm{i}'}(p^*,p_k) \left\{\frac{\Pr(\bm{A}^\top_{k\bm{i}}\bm{1} = |\bm{i}|p^*,\bm{A}^\top_{k\bm{i}'}\bm{1} = |\bm{i}'|p^*)}{f_k(\bm{A}^\top_{k\bm{i}}\bm{1} = |\bm{i}|p^*)f_k(\bm{A}^\top_{k\bm{i}'}\bm{1} = |\bm{i}'|p^*)} - 1  \right\} \Big]\nonumber \\
& = \frac{1}{K^2}\sum_{k=1}^{K}\frac{1}{|\mathcal{S}_k|^2}\left[\sum_{\bm{i} \in \mathcal{G}_k}\tilde{c}_{\bm{i}} \tilde{Y}^2_{k\bm{i}}(p^*,p_k) + \mathop{\sum\sum}_{\bm{i} \neq \bm{i}' \in \mathcal{G}_k}\tilde{d}_{\bm{i}\bm{i}'}\tilde{Y}_{k\bm{i}}(p^*,p_k)\tilde{Y}_{k\bm{i}'}(p^*,p_k)\right].
\end{align*}
\qed 
\subsubsection{Inference for the H\'{a}jek estimator}
\label{appsec_hajek}
In this section, we derive the design-based variances of the H\'{a}jek estimators.
To this end, we first focus on the H\'{a}jek estimator of $\mu^\pi_{a}$ and note that, $\hat{\mu}^\pi_{a,\text{H\'{a}jek}}$ can be written as the ratio of two Horvitz-Thompson estimators; that is,
$$\hat{\mu}^\pi_{a,\text{H\'{a}jek}} = \frac{\hat{\mu}^\pi_{a,\text{HT}}}{\hat{\lambda}^\pi_{a,\text{HT}}}, \quad \text{where} \quad \hat{\lambda}^\pi_{a,\text{HT}} = \frac{1}{K}\sum_{k=1}^{K}\frac{1}{|\mathcal{S}_k|}\sum_{j \in \mathcal{S}_k}\mathbbm{1}(A_{ki^*} = a)\frac{\pi_{k}(\bm{A}_{k(-i^*)})\mid A_{ki^*} = a)}{f_k(\bm{A}_k)}$$ with $\mathbb{E}(\hat{\lambda}^\pi_{a,\text{HT}})  =1$. Since $\hat{\mu}^\pi_{a,\text{H\'{a}jek}}$ is the ratio of two random quantities, in general, $\hat{\mu}^\pi_{a,\text{H\'{a}jek}}$ is not design-unbiased for $\mu^\pi_a$, and we cannot obtain its design-based variance in closed form. However, we show that it is design-consistent for $\mu^\pi_a$ and approximate its variance by linearization, provided the estimators $\hat{\mu}^\pi_{a,\text{HT}}$ and $\hat{\lambda}^\pi_{a,\text{HT}}$ are design-consistent (see, e.g., \citealt{lohr2021sampling}, Chapter 9, for related analyses). Here, we first illustrate the estimation of this variance under stratified interference.

When $\pi_k(\cdot) = f_k(\cdot)$ and $f_k(\cdot)$ corresponds to a completely randomized experiment,  $\Var(\hat{\mu}_{a,\text{HT}}) = \frac{1}{K^2}\sum_{k=1}^{K} \left(\frac{n_k}{|\mathcal{S}_k|}\right)^2 \left(1- \frac{n_{ka}}{n_k}\right)\frac{\tilde{V}^2_{ka}}{n_{ka}}$. 
Hence, a sufficient condition for design-consistency of $\hat{\mu}^\pi_{a,\text{HT}}$ is that $\frac{1}{K}\sum_{k=1}^{K} \left(\frac{n_k}{|\mathcal{S}_k|}\right)^2 \left(1- \frac{n_{ka}}{n_k}\right)\frac{\tilde{V}^2_{ka}}{n_{ka}}$ is bounded, which holds when, e.g., $n_k,|\mathcal{S}_{k}| \to \infty$, $\frac{n_k}{|\mathcal{S}_k|} \to \gamma (<\infty)$, and $\tilde{V}^2_{ka}$ is bounded. Design-consistency of $\hat{\lambda}^\pi_{a,\text{HT}}$ holds under analogous conditions. 
For general expressions of $\pi_k(\cdot)$ and $f_k(\cdot)$, the following theorem establishes consistency of $\hat{\mu}^\pi_{a,\text{H\'{a}jek}}$ under similar conditions, and provides an approximate closed-form expression of its variance using linearization.
\begin{theorem}[Design-consistent estimator and its variance under stratified interference] \normalfont
Let $D_{ki} = \sum_{j \in \mathcal{S}_k}\mathbbm{1}(j \leftarrow i)$, $i \in \mathcal{I}_k$.  Assume that the second-order terms $\frac{1}{K}\sum_{k=1}^{K}\frac{1}{|\mathcal{S}_k|^2}\Big\{\sum_{i=1}^{n_k}c_{i,a} \tilde{Y}^2_{ki}(a,p_k) + \mathop{\sum\sum}_{i \neq i'}d_{ii',a}\tilde{Y}_{ki}(a,p_k)\tilde{Y}_{ki'}(a,p_k)\Big\}$ and $\frac{1}{K}\sum_{k=1}^{K}\frac{1}{|\mathcal{S}_k|^2}\Big\{\sum_{i=1}^{n_k}c_{i,a} D^2_{ki} + \mathop{\sum\sum}_{i \neq i'}d_{ii',a}D_{ki}D_{ki'}\Big\}$ are bounded.
Then, as $K \to \infty$, under Assumptions~\ref{assump_partial}--\ref{assump_fixedprop}, we have
\begin{equation*}
\hat{\mu}^\pi_{a,\text{H\'{a}jek}} \xrightarrow{P} \mu^\pi_a 
\end{equation*}
and
\begin{align*}
\Var(\hat{\mu}^\pi_{a,\text{H\'{a}jek}}) = \Var(\hat{\mu}^\pi_{a,\text{HT}}) + (\mu^\pi_{a})^2\Var(\hat{\lambda}^\pi_{a,\text{HT}}) - 2\mu^\pi_{a} \Cov(\hat{\mu}^\pi_{a,\text{HT}}, \hat{\lambda}^\pi_{a,\text{HT}}) + o_P(1).     
\end{align*}
    \label{thm_Hajek1}
\end{theorem}
\begin{proof}
    Without loss of generality, we set $a = 1$. Now,
\begin{align}
\Var(\hat{\mu}^\pi_{1,\text{HT}}) = \frac{1}{K^2}\sum_{k=1}^{K}\frac{1}{|\mathcal{S}_k|^2}\left\{\sum_{i=1}^{n_k}c_{i,1} \tilde{Y}^2_{ki}(1,p_k) + \mathop{\sum\sum}_{i \neq i'}d_{ii',1}\tilde{Y}_{ki}(1,p_k)\tilde{Y}_{ki'}(1,p_k)\right\},
\end{align}
By the given condition, we have $\hat{\mu}^\pi_{1,\text{HT}} -\mu^\pi_1 = O_P(1/K)$ and $\hat{\lambda}^\pi_{1,\text{HT}} -1 = O_P(1/K)$ as $K \to  \infty$. Therefore, by Slutsky's theorem, $\hat{\mu}^\pi_{1,\text{H\'{a}jek}} = \frac{\hat{\mu}^\pi_{1,\text{HT}}}{\hat{\lambda}^\pi_{1,\text{HT}}} \xrightarrow{P} \mu^\pi_1$ as $K \to \infty$.
Now, using Taylor's expansion, for $h(\cdot): \mathbb{R}^2 \to \mathbb{R}$,
\begin{align*}
h(\hat{\mu}^\pi_{1,\text{HT}},\hat{\lambda}^\pi_{1,\text{HT}}) & = h(\mu^\pi_1,1) + (\hat{\mu}^\pi_{1,\text{HT}} - \mu^\pi_1, \hat{\lambda}^\pi_{1,\text{HT}} - 1)\nabla h(\mu^\pi_1,1) + O_P(1/K).
\end{align*}
Thus, we have
\begin{align*}
\Var\{h(\hat{\mu}^\pi_{1,\text{HT}},\hat{\lambda}^\pi_{1,\text{HT}})\} =   (\nabla h(\mu^\pi_1,1))^\top \Var\{(\hat{\mu}^\pi_{1,\text{HT}}, \hat{\lambda}^\pi_{1,\text{HT}})^\top\} \nabla h(\mu^\pi_1,1) + o_P(1). 
\end{align*}
Setting $h(x,y) = x/y$, we get $\nabla h(x,y) = (1/y,-x/y^2)$, which implies,
\begin{align*}
 \Var(\hat{\mu}^\pi_{1,\text{H\'{a}jek}}) & = \Var(\hat{\mu}^\pi_{1,\text{HT}}) + (\mu^\pi_{1})^2\Var(\hat{\lambda}^\pi_{1,\text{HT}}) - 2\mu^\pi_{1} \Cov(\hat{\mu}^\pi_{1,\text{HT}}, \hat{\lambda}^\pi_{1,\text{HT}}) + o_P(1).   
\end{align*}

\end{proof}

Leveraging this result, we compute the variance by the plug-in estimator
\begin{align}
    \widehat{\Var}(\hat{\mu}^\pi_{a,\text{H\'{a}jek}}) = \widehat{\Var}(\hat{\mu}^\pi_{a,\text{HT}}) + (\hat{\mu}^\pi_{a,\text{H\'{a}jek}})^2\widehat{\Var}(\hat{\lambda}^\pi_{a,\text{HT}}) - 2\hat{\mu}^\pi_{a, \text{H\'{a}jek}} \widehat{\Cov}(\hat{\mu}^\pi_{a,\text{HT}}, \hat{\lambda}^\pi_{a,\text{HT}}),
\end{align}
where 
\begin{align}
    \widehat{\Var}(\hat{\mu}^\pi_{a,\text{HT}}) = \frac{1}{K^2}\sum_{k=1}^{K}\frac{1}{|\mathcal{S}_k|^2}\left\{\sum_{i=1}^{n_k}\frac{\mathbbm{1}(A_{ki} = a)}{f_k(A_{ki} = a)}c_{i,a}\tilde{Y}_{ki}^2 + \mathop{\sum\sum}_{i \neq i'} \frac{\mathbbm{1}(A_{ki} = a, A_{ki'} = a)}{f_k(A_{ki} = a,A_{ki'} = a)} d_{ii',a}\tilde{Y}_{ki}\tilde{Y}_{ki'}\right\},
    \label{eq_app_muahatht}
\end{align}
$\widehat{\Var}(\hat{\lambda}^\pi_{a,\text{HT}})$ replaces $\tilde{Y}_{ki}$ in Equation~\eqref{eq_app_muahatht} by $D_{ki}$, and 
\begin{align}
\widehat{\Cov}(\hat{\mu}^\pi_{a,\text{HT}},\hat{\lambda}^\pi_{a,\text{HT}}) &= \frac{1}{K^2}\sum_{k=1}^{K}\frac{1}{|\mathcal{S}_k|^2}\Big\{\sum_{i=1}^{n_k}\frac{\mathbbm{1}(A_{ki} = a)}{f_k(A_{ki} = a)}c_{i,a}\tilde{Y}^{\text{obs}}_{ki}D_{ki} \nonumber \\
& \quad + \mathop{\sum\sum}_{i \neq i'} \frac{\mathbbm{1}(A_{ki} = a, A_{ki'} = a)}{f_k(A_{ki} = a,A_{ki'} = a)} d_{ii',a}\tilde{Y}^{\text{obs}}_{ki}D_{ki'}\Big\}.  
\end{align}
The estimators of the above variances and covariances can also be obtained analogously under additive interference. In particular, $\widehat{\Var}(\hat{\mu}^\pi_{a,\text{HT}})$ and $\widehat{\Var}(\hat{\lambda}^\pi_{a,\text{HT}})$ are computed by plugging in the estimated coefficients of the additive model in the general variance expression in Proposition \ref{prop_additive_mu1}. As for the covariance, following the proof of Proposition \ref{prop_additive_mu1}, we get
\begin{align}
\Cov(\hat{\mu}^\pi_{a,\text{HT}}, \hat{\lambda}^\pi_{a,\text{HT}}) =   \frac{1}{K^2}\sum_{k = 1}^{K} \frac{1}{|\mathcal{S}_k|^2}\left(\sum_{j \in \mathcal{S}_k} \Lambda_{1,k,j} + \mathop{\sum\sum}_{j \neq j' \in \mathcal{S}_k}\Lambda_{2,k,j,j'} \right),  
\label{eq_app_cov_mulambda}
\end{align}
where 
\begin{align}
\Lambda_{1,k,j} \ = \ & \sum_{\bm{s}}\frac{\pi^2_{k}(\bm{A}_{k(-i^*)} = \bm{s}\mid A_{ki^*} = a)}{f_k(A_{ki^*} = a, \bm{A}_{k(-i^*)} = \bm{s})}\left\{(1,\bm{a}^\top_{kj})\tilde{\bm{\beta}}_{kj} \right\} - \left\{(1,\bm{\pi}^\top_k(\cdot|A_{ki^*} = a))\tilde{\bm{\beta}}_{kj} \right\},
\label{eq_lambda1_add} \\
       \Lambda_{2,k,j,j'} \ = \ &\sum_{\bm{s}}\frac{\mathbbm{\pi}^2_k(A_{ki^{*}} = a, A_{ki^{*'}} = a, \bm{A}_{k(-i^{*},-i^{*'})} = \bm{s}) \left\{(1,\bm{a}^\top_{kjj'})\tilde{\bm{\beta}}_{kj} \right\}}{f_k(A_{ki^*} = a, A_{ki^{*'}} = a, \bm{A}_{k(-i^{*},-i^{*'})} = \bm{s})\pi_k(A_{ki^*} = a)\pi_k(A_{ki^{*'}} = a)} \nonumber\\
    &\hspace{0.5cm} - \left\{(1,\bm{\pi}^\top_k(\cdot|A_{ki^*} = a))\tilde{\bm{\beta}}_{kj} \right\},
\end{align}
In this case, $\widehat{\Cov}(\hat{\mu}^\pi_{a,\text{HT}}, \hat{\lambda}^\pi_{a,\text{HT}})$ substitutes $\tilde{\bm{\beta}}_{kj}$ by its estimator $\hat{\tilde{\bm{\beta}}}_{kj}$.

We now derive the approximate design-based variance of the H\'{a}jek estimator of $\text{DE}^\pi$.
The H\'{a}jek estimator of $\text{DE}^\pi$ is given by,
\begin{align}
    \widehat{\text{DE}}^\pi_{\text{H\'{a}jek}} &= \hat{\mu}^\pi_{1,\text{H\'{a}jek}} - \hat{\mu}^\pi_{0,\text{H\'{a}jek}} = \frac{\hat{\mu}^\pi_{1,\text{HT}}}{\hat{\lambda}^\pi_{1,\text{HT}}} - \frac{\hat{\mu}^\pi_{0,\text{HT}}}{\hat{\lambda}^\pi_{0,\text{HT}}} .
\end{align}
Under the assumptions of Theorem \ref{thm_Hajek1}, we have, for $a \in \{0,1\}$, $\hat{\mu}^\pi_{a,\text{HT}} - \mu^\pi_a = O_p(1/K)$ and $\hat{\lambda}^\pi_{a,\text{HT}} - \lambda^\pi_a = O_p(1/K)$. Thus, using Taylor expansion, we get
\begin{align*}
&\Var(h(\hat{\mu}^\pi_{0,\text{HT}},\hat{\lambda}^\pi_{0,\text{HT}},\hat{\mu}^\pi_{1,\text{HT}},\hat{\lambda}^\pi_{1,\text{HT}}) \nonumber \\
& = (\nabla h(\mu^\pi_0,1,\mu^\pi_1,1))^\top \Var\{(\hat{\mu}^\pi_{0,\text{HT}},\hat{\lambda}^\pi_{0,\text{HT}},\hat{\mu}^\pi_{1,\text{HT}},\hat{\lambda}^\pi_{1,\text{HT}})^\top \} \nabla h(\mu^\pi_0,1,\mu^\pi_1,1) + o_P(1).
\end{align*}
Now, setting $h(x_0,y_0,x_1,y_1) = (x_1/y_1) - (x_0/y_0)$, we have $\nabla h(\mu^\pi_0,1,\mu^\pi_1,1) = (-1,\mu^\pi_0,1,-\mu^\pi_1)^\top$. Therefore,
\begin{align*}
\Var(\widehat{\text{DE}}^\pi_{\text{H\'{a}jek}}) & =  \Var(\hat{\mu}^\pi_{0,\text{HT}}) + (\mu^\pi_{0})^2\Var(\hat{\lambda}^\pi_{0,\text{HT}}) - 2\mu^\pi_{0} \Cov(\hat{\mu}^\pi_{0,\text{HT}}, \hat{\lambda}^\pi_{0,\text{HT}})  \nonumber \\
& \quad + \Var(\hat{\mu}^\pi_{1,\text{HT}}) + (\mu^\pi_{1})^2\Var(\hat{\lambda}^\pi_{1,\text{HT}}) - 2\mu^\pi_{1} \Cov(\hat{\mu}^\pi_{1,\text{HT}}, \hat{\lambda}^\pi_{1,\text{HT}}) \nonumber \\
& \quad -2\Cov(\hat{\mu}^\pi_{0,\text{HT}},\hat{\mu}^\pi_{1,\text{HT}}) + 2\mu^\pi_1 \Cov(\hat{\mu}^\pi_{0,\text{HT}},\hat{\lambda}^\pi_{1,\text{HT}}) + 2\mu^\pi_0 \Cov(\hat{\lambda}^\pi_{0,\text{HT}},\hat{\mu}^\pi_{1,\text{HT}}) \nonumber\\
& \quad - 2\mu^\pi_0\mu^\pi_1 \Cov(\hat{\lambda}^\pi_{0,\text{HT}},\hat{\lambda}^\pi_{1,\text{HT}}) + o_P(1).
\end{align*}
Following the proof of Theorem \ref{thm_var1} and \ref{thm_varDE}, under stratified interference, we can compute the covariance terms as follows.
\begin{align*}
\Cov(\hat{\mu}^\pi_{a,\text{HT}},\hat{\lambda}^\pi_{a,\text{HT}}) & = \frac{1}{K^2}\sum_{k=1}^{K}\frac{1}{|\mathcal{S}_k|^2}\Big\{\sum_{i=1}^{n_k}c_{i,a}\tilde{Y}_{ki}(a,p_k)D_{ki} + \mathop{\sum\sum}_{i \neq i'} d_{ii',a}\tilde{Y}_{ki}(a,p_k)D_{ki'}\Big\}.  \\
    \Cov(\hat{\mu}^\pi_{1,\text{HT}}, \hat{\mu}^\pi_{0,\text{HT}}) & = \frac{1}{K^2}\sum_{k=1}^{K}\frac{1}{|\mathcal{S}_k|^2}\left[\mathop{\sum\sum}_{i \neq i'}g_{ii'}\tilde{Y}_{ki}(1,p_k)\tilde{Y}_{ki'}(0,p_k) -   \sum_{i = 1}^{n_k} \tilde{Y}_{ki}(1,p_k)\tilde{Y}_{ki}(0,p_k)     \right], \\
    \Cov(\hat{\mu}^\pi_{1,\text{HT}}, \hat{\lambda}^\pi_{0,\text{HT}}) & = \frac{1}{K^2}\sum_{k=1}^{K}\frac{1}{|\mathcal{S}_k|^2}\left[\mathop{\sum\sum}_{i \neq i'}g_{ii'}\tilde{Y}_{ki}(1,p_k)D_{ki'} - \sum_{i = 1}^{n_k} \tilde{Y}_{ki}(1,p_k)D_{ki}\right],\\
    \Cov(\hat{\lambda}^\pi_{1,\text{HT}}, \hat{\mu}^\pi_{0,\text{HT}}) & = \frac{1}{K^2}\sum_{k=1}^{K}\frac{1}{|\mathcal{S}_k|^2}\left[\mathop{\sum\sum}_{i \neq i'}g_{ii'}D_{ki}\tilde{Y}_{ki'}(0,p_k) - \sum_{i = 1}^{n_k} D_{ki}\tilde{Y}_{ki}(0,p_k)\right], \\
    \Cov(\hat{\lambda}^\pi_{1,\text{HT}}, \hat{\lambda}^\pi_{0,\text{HT}}) & = \frac{1}{K^2}\sum_{k=1}^{K}\frac{1}{|\mathcal{S}_k|^2}\left[\mathop{\sum\sum}_{i \neq i'}g_{ii'}D_{ki}D_{ki'} -   \sum_{i = 1}^{n_k} D^2_{ki} \right].
\end{align*}

We estimate the above variance by the plug-in estimator
\begin{align*}
\widehat{\Var}(\widehat{\text{DE}}^\pi_{\text{H\'{a}jek}}) & =  \widehat{\Var}(\hat{\mu}^\pi_{0,\text{HT}}) + (\hat{\mu}^\pi_{0, \text{H\'{a}jek}})^2\widehat{\Var}(\hat{\lambda}^\pi_{0,\text{HT}}) - 2\hat{\mu}^\pi_{0,\text{H\'{a}jek}} \widehat{\Cov}(\hat{\mu}^\pi_{0,\text{HT}}, \hat{\lambda}^\pi_{0,\text{HT}})  \nonumber \\
& + \widehat{\Var}(\hat{\mu}^\pi_{1,\text{HT}}) + (\hat{\mu}^\pi_{1,\text{H\'{a}jek}})^2\widehat{\Var}(\hat{\lambda}^\pi_{1,\text{HT}}) - 2\hat{\mu}^\pi_{1,\text{H\'{a}jek}} \widehat{\Cov}(\hat{\mu}^\pi_{1,\text{HT}}, \hat{\lambda}^\pi_{1,\text{HT}}) \nonumber \\
& -2\widehat{\Cov}(\hat{\mu}^\pi_{0,\text{HT}},\hat{\mu}^\pi_{1,\text{HT}}) + 2\hat{\mu}^\pi_{1,\text{H\'{a}jek}} \widehat{\Cov}(\hat{\mu}^\pi_{0,\text{HT}},\hat{\lambda}^\pi_{1,\text{HT}}) + 2\hat{\mu}^\pi_{0,\text{H\'{a}jek}} \widehat{\Cov}(\hat{\lambda}^\pi_{0,\text{HT}},\hat{\mu}^\pi_{1,\text{HT}}) \nonumber\\
&- 2\hat{\mu}^\pi_{0,\text{H\'{a}jek}}\hat{\mu}^\pi_{1,\text{H\'{a}jek}} \widehat{\Cov}(\hat{\lambda}^\pi_{0,\text{HT}},\hat{\lambda}^\pi_{1,\text{HT}}).
\end{align*}
Here, following the proof of Theorem \ref{thm_var1} and \ref{thm_varDE}, we can estimate the covariance terms as,
\begin{align*}
&\widehat{\Cov} (\hat{\mu}^\pi_{a,\text{HT}},\hat{\lambda}^\pi_{a,\text{HT}})  \nonumber\\
&= \frac{1}{K^2}\sum_{k=1}^{K}\frac{1}{|\mathcal{S}_k|^2}\Big\{\sum_{i=1}^{n_k}\frac{\mathbbm{1}(A_{ki} = a)}{f_k(A_{ki} = a)}c_{i,a}\tilde{Y}^{\text{obs}}_{ki}D_{ki} + \mathop{\sum\sum}_{i \neq i'} \frac{\mathbbm{1}(A_{ki} = a, A_{ki'} = a)}{f_k(A_{ki} = a,A_{ki'} = a)} d_{ii',a}\tilde{Y}^{\text{obs}}_{ki}D_{ki'}\Big\}, \\
\end{align*}
\begin{align*}
\widehat{\Cov}(\hat{\mu}^\pi_{1,\text{HT}}, \hat{\mu}^\pi_{0,\text{HT}}) & = \frac{1}{K^2}\sum_{k=1}^{K}\frac{1}{|\mathcal{S}_k|^2}\Big[\mathop{\sum\sum}_{i \neq i'}g_{ii'}\frac{\mathbbm{1}(A_{ki} = 1, A_{ki'} = 0)}{f_k(A_{ki} = 1,A_{ki'} = 0)}\tilde{Y}^{\text{obs}}_{ki}\tilde{Y}^{\text{obs}}_{ki'} \nonumber\\
& \quad - \frac{1}{2}\sum_{i = 1}^{n_k}\left\{\frac{\mathbbm{1}(A_{ki} = 1)}{f_k(A_{ki} = 1)} + \frac{\mathbbm{1}(A_{ki} = 0)}{f_k(A_{ki} = 0)}\right\} (\tilde{Y}^{\text{obs}}_{ki})^2  \Big], \\
\widehat{\Cov}(\hat{\mu}^\pi_{1,\text{HT}}, \hat{\lambda}^\pi_{0,\text{HT}}) & = \frac{1}{K^2}\sum_{k=1}^{K}\frac{1}{|\mathcal{S}_k|^2}\left[\mathop{\sum\sum}_{i \neq i'}g_{ii'}\frac{\mathbbm{1}(A_{ki} = 1, A_{ki'} = 0)}{f_k(A_{ki} = 1,A_{ki'} = 0)}\tilde{Y}^{\text{obs}}_{ki}D_{ki'} - \sum_{i = 1}^{n_k} \frac{\mathbbm{1}(A_{ki} = 1)}{f_k(A_{ki} = 1)}\tilde{Y}^{\text{obs}}_{ki}D_{ki}\right], \\
\widehat{\Cov}(\hat{\lambda}^\pi_{1,\text{HT}}, \hat{\mu}^\pi_{0,\text{HT}}) & = \frac{1}{K^2}\sum_{k=1}^{K}\frac{1}{|\mathcal{S}_k|^2}\left[\mathop{\sum\sum}_{i \neq i'}g_{ii'}\frac{\mathbbm{1}(A_{ki} = 1, A_{ki'} = 0)}{f_k(A_{ki} = 1,A_{ki'} = 0)}D_{ki}\tilde{Y}^{\text{obs}}_{ki'} - \sum_{i = 1}^{n_k} \frac{\mathbbm{1}(A_{ki} = 0)}{f_k(A_{ki} = 0)}D_{ki}\tilde{Y}^{\text{obs}}_{ki}\right],\\
\widehat{\Cov}(\hat{\lambda}^\pi_{1,\text{HT}}, \hat{\lambda}^\pi_{0,\text{HT}}) & = \frac{1}{K^2}\sum_{k=1}^{K}\frac{1}{|\mathcal{S}_k|^2}\Big[\mathop{\sum\sum}_{i \neq i'}g_{ii'}\frac{\mathbbm{1}(A_{ki} = 1, A_{ki'} = 0)}{f_k(A_{ki} = 1,A_{ki'} = 0)}D_{ki}D_{ki'} \nonumber\\
&- \frac{1}{2}\sum_{i = 1}^{n_k}\left\{\frac{\mathbbm{1}(A_{ki} = 1)}{f_k(A_{ki} = 1)} + \frac{\mathbbm{1}(A_{ki} = 0)}{f_k(A_{ki} = 0)}\right\} (D_{ki})^2  \Big].
\end{align*}

Under additive interference, the above covariance terms can be computed analogously to those in the proofs of Propositions \ref{prop_additive_mu1} and \ref{prop_varDE_additive}. For instance, $\Cov(\hat{\mu}^\pi_{1,\text{HT}},\hat{\lambda}^\pi_{0,\text{HT}})$ has the following closed-form expression.
\begin{align}
\Cov(\hat{\mu}^\pi_{1,\text{HT}},\hat{\lambda}^\pi_{0,\text{HT}}) &=   \frac{1}{K^2}\sum_{k = 1}^{K} \frac{1}{|\mathcal{S}_k|^2}\left(\sum_{j \in \mathcal{S}_k} \Lambda_{1,k,j} + \mathop{\sum\sum}_{j \neq j' \in \mathcal{S}_k}\Lambda_{2,k,j,j'} \right),   
\end{align}
where 
\begin{align}
\Lambda_{1,k,j} \ = \ & - \left\{(1,\bm{\pi}^\top_k(\cdot|A_{ki^*} = 1))\tilde{\bm{\beta}}_{kj} \right\},\\
       \Lambda_{2,k,j,j'} \ = \ &\sum_{\bm{s}}\frac{\mathbbm{\pi}^2_k(A_{ki^{*}} = 1, A_{ki^{*'}} = 0, \bm{A}_{k(-i^{*},-i^{*'})} = \bm{s}) \left\{(1,\bm{a}^\top_{kjj'})\tilde{\bm{\beta}}_{kj} \right\}}{f_k(A_{ki^*} = 1, A_{ki^{*'}} = 0, \bm{A}_{k(-i^{*},-i^{*'})} = \bm{s})\pi_k(A_{ki^*} = 1)\pi_k(A_{ki^{*'}} = 0)} \nonumber\\
    &\hspace{0.5cm} - \left\{(1,\bm{\pi}^\top_k(\cdot|A_{ki^*} = 1))\tilde{\bm{\beta}}_{kj} \right\}. 
\end{align}
Moreover, $\Cov(\hat{\mu}^\pi_{a,\text{HT}}, \hat{\lambda}^\pi_{a,\text{HT}})$ has a closed-form expression given in Equation \ref{eq_app_cov_mulambda}. The estimators of these covariance terms are obtained by plugging in the estimators of $\tilde{\bm{\beta}}_{kj}$.

Finally, we conclude the section by considering the multiple key-intervention unit case as in Section \ref{sec_appendix_multiple}.
The H\'{a}jek estimator of $\tau^\pi$ is given by
\begin{align*}
\hat{\tau}^\pi_{\text{H\'{a}jek}} =   \frac{\sum_{k=1}^{K}\frac{1}{|\mathcal{S}_k|}\sum_{j \in \mathcal{S}_k}\frac{\mathbbm{1}(\bm{A}^\top_{k\bm{i}^*}\bm{1} = |\bm{i}^*|p^*)}{f_k(\bm{A}^\top_{k\bm{i}^*}\bm{1} = |\bm{i}^*|p^*)}Y^{\text{obs}}_{kj} } {\sum_{k=1}^{K}\frac{1}{|\mathcal{S}_k|}\sum_{j \in \mathcal{S}_k}\frac{\mathbbm{1}(\bm{A}^\top_{k\bm{i}^*}\bm{1} = |\bm{i}^*|p^*)}{f_k(\bm{A}^\top_{k\bm{i}^*}\bm{1} = |\bm{i}^*|p^*)}} = \frac{\hat{\tau}^\pi_{\text{HT}}}{\hat{\lambda}^\pi_{\text{HT}}}.  
\end{align*}
The form of the variance of $\hat{\tau}^\pi_{\text{H\'{a}jek}}$ and its estimator is analogous to those in Theorem \ref{thm_Hajek1}, where $\hat{\mu}^\pi_{a,\text{HT}}$ is replaced by  $\hat{\tau}^\pi_{\text{HT}}$ and $\hat{\lambda}^\pi_{a,\text{HT}}$ is replaced by  $\hat{\lambda}^\pi_{\text{HT}}$, and the derivation is analogous to the proof of Theorem~\ref{thm_Hajek1}. 


\subsection{Additional results from the simulation study}
\label{sec_add_simu}

\begin{figure}[H]
    \centering
    \includegraphics[scale = 0.53]{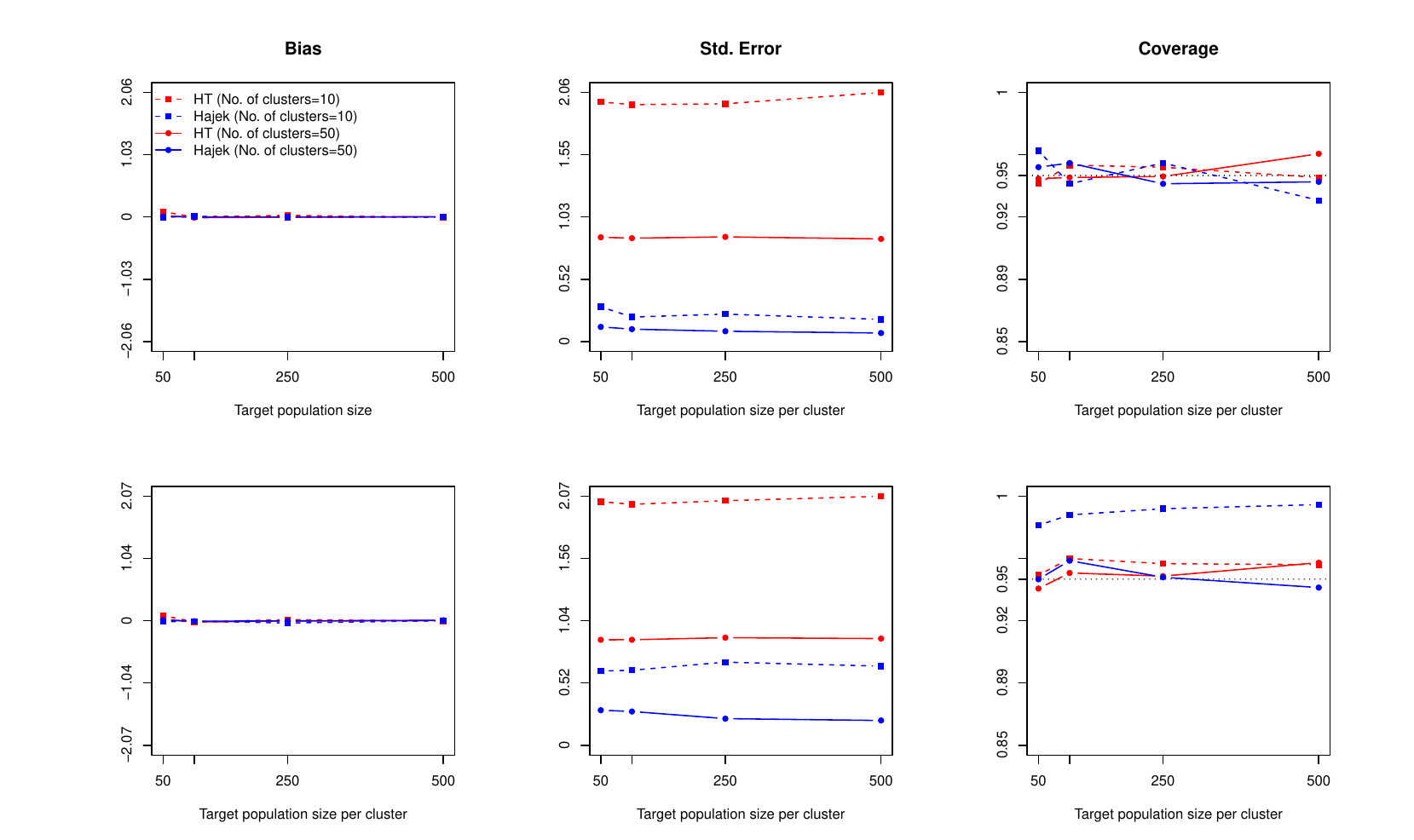}
    \caption{Bias, standard error, and coverage of 95\% confidence intervals for the Horvitz-Thompson and H\'{a}jek estimators of $\mu^\pi_1$ under outcome models M1 and M2 and stochastic intervention $\pi^{(2)}_k(\cdot)$. The first and second row correspond to outcome models M1 and M2, respectively.}
    \label{fig:simu_mu1_strat}
\end{figure}
\begin{figure}[H]
    \centering
    \includegraphics[scale = 0.53]{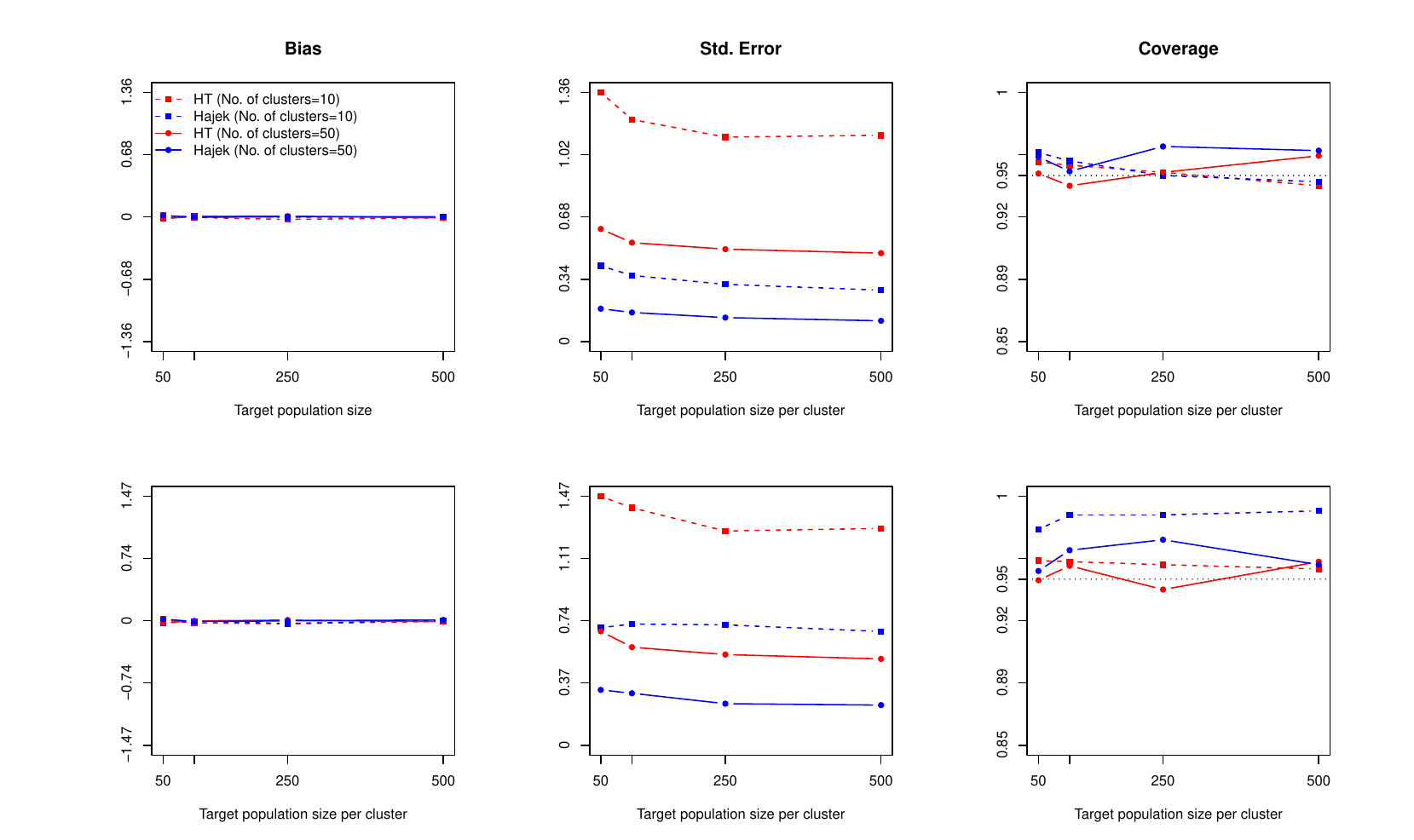}
    \caption{Bias, standard error, and coverage of 95\% confidence intervals for the Horvitz-Thompson and H\'{a}jek estimators of $\text{DE}^\pi$ under outcome models M1 and M2 and stochastic intervention $\pi^{(2)}_k(\cdot)$. The first and second row correspond to outcome models M1 and M2, respectively.}
    \label{fig:simu_DE_strat}
\end{figure}

\subsection{Additional results from the case study}
\label{sec_add_case}

\begin{singlespacing}
\begin{table}[!ht]
\centering
\scalebox{0.75}{
\begin{tabular}{cccccccc}
\toprule
  &   &  \multicolumn{3}{c}{Seed-ineligible population} &  \multicolumn{3}{c}{Seed-eligible population}\\
  \cline{3-8}
Outcome &  & Estimate & Std. Error & 95\% CI & Estimate & Std. Error & 95\% CI \\
\hline
\multirow{4}{*}{Talking about conflict} & $\hat{\mu}^{\pi}_{1,\text{HT}}$ & 0.41 & 0.51 & ( -0.6 , 1.41 ) & 0.44 & 0.56 & ( -0.65 , 1.54 ) \\ 
& $\hat{\mu}^{\pi}_{1,\text{H\'{a}jek}}$ & 0.40 & 0.51 & ( -0.61 , 1.41 ) & 0.44 & 0.56 & ( -0.65 , 1.54 ) \\ 
& $\widehat{\text{DE}}^{\pi}_{\text{HT}}$ & 0.04 & 0.13 & ( -0.23 , 0.3 ) & 0.07 & 0.17 & ( -0.25 , 0.4 ) \\ 
& $\widehat{\text{DE}}^{\pi}_{\text{H\'{a}jek}}$ & 0.02 & 0.14 & ( -0.25 , 0.28 ) & 0.07 & 0.17 & ( -0.25 , 0.4 ) \\ 
\hline
\multirow{4}{*}{Wearing anti-conflict wristbands} & $\hat{\mu}^{\pi}_{1,\text{HT}}$ & 0.19 & 0.25 & ( -0.3 , 0.68 ) & 0.28 & 0.30 & ( -0.31 , 0.87 ) \\ 
& $\hat{\mu}^{\pi}_{1,\text{H\'{a}jek}}$ & 0.19 & 0.25 & ( -0.3 , 0.68 ) & 0.28 & 0.30 & ( -0.31 , 0.87 ) \\ 
& $\widehat{\text{DE}}^{\pi}_{\text{HT}}$ & 0.03 & 0.07 & ( -0.12 , 0.17 ) & 0.14 & 0.15 & ( -0.15 , 0.42 ) \\ 
  & $\widehat{\text{DE}}^{\pi}_{\text{H\'{a}jek}}$ & 0.02 & 0.07 & ( -0.13 , 0.16 ) & 0.14 & 0.15 & ( -0.15 , 0.42 ) \\ 
  \hline
\multirow{4}{*}{Cases of conflict  }  & $\hat{\mu}^{\pi}_{1,\text{HT}}$  & 0.16 & 0.30 & ( -0.42 , 0.75 ) & 0.15 & 0.33 & ( -0.51 , 0.8 ) \\ 
  & $\hat{\mu}^{\pi}_{1,\text{H\'{a}jek}}$ & 0.16 & 0.30 & ( -0.43 , 0.74 ) & 0.15 & 0.33 & ( -0.51 , 0.8 ) \\ 
   & $\widehat{\text{DE}}^{\pi}_{\text{HT}}$ & 0.00 & 0.15 & ( -0.3 , 0.3 ) & 0.00 & 0.24 & ( -0.47 , 0.48 ) \\ 
 & $\widehat{\text{DE}}^{\pi}_{\text{H\'{a}jek}}$ & -0.01 & 0.15 & ( -0.3 , 0.29 ) & 0.00 & 0.24 & ( -0.47 , 0.48 ) \\

\bottomrule
\end{tabular}
}
\caption{Estimates, standard errors (SE) and $95\%$ confidence intervals (CI) of the average potential outcomes and direct effects under additive interference for the two target populations, where the stochastic intervention equals the actual intervention.}
\label{tab:q1_additive}
\end{table}
\end{singlespacing}

\begin{figure}[!ht]
    \centering
    \includegraphics[scale = 0.8]{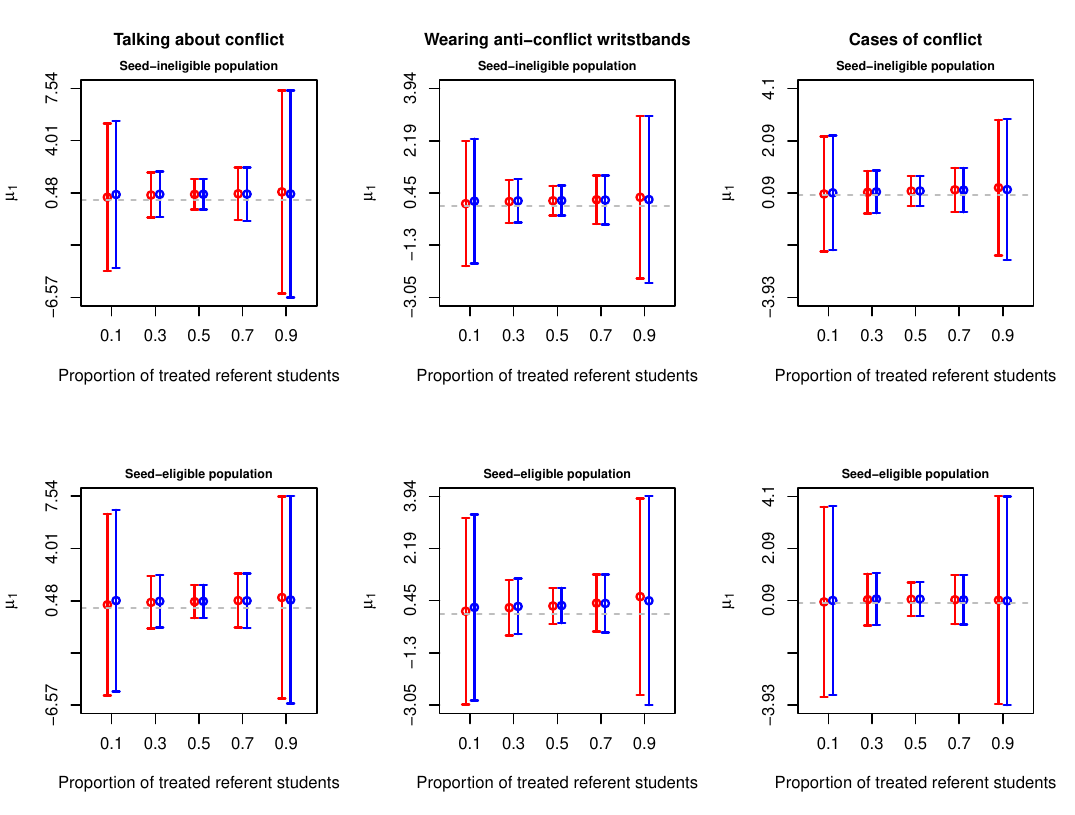}
    \caption{Point estimates and $95\%$ confidence intervals under additive interference for the Horvitz-Thompson (red) and H\'{a}jek (blue) estimators of $\mu^\pi_1$ for two target populations.}
    \label{fig:mu1_estimates_additive}
\end{figure}

\begin{figure}[!ht]
    \centering
    \includegraphics[scale = 0.8]{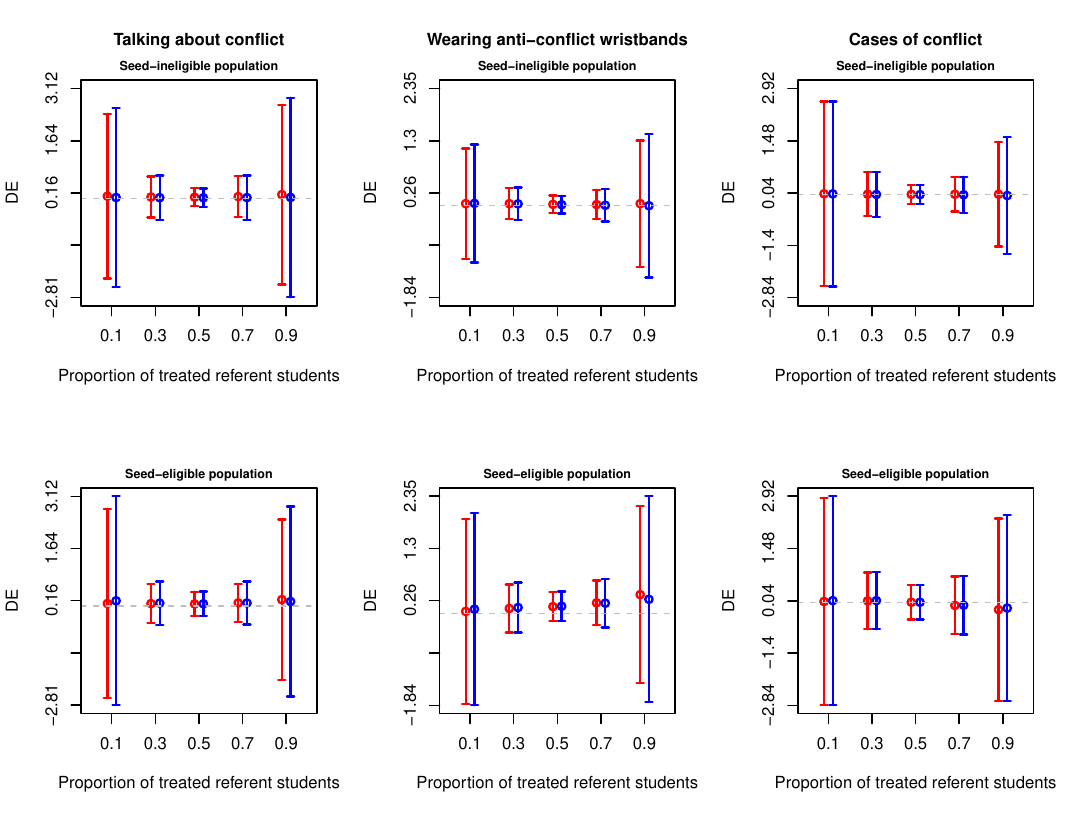}
    \caption{Point estimates and $95\%$ confidence intervals under additive interference for the Horvitz-Thompson (red) and H\'{a}jek (blue) estimators of $\text{DE}^\pi$ for two target populations.}
    \label{fig:DE_estimates_additive}
\end{figure}


\end{document}